%% file: nonShannon.tex
\documentclass[10pt,journal]{IEEEtran}
\usepackage{amsmath,amssymb} 
\usepackage{color}
\usepackage{graphicx}
\usepackage{epsfig}
\begin{document}

\newcommand{\red}{\color{red}}
\newcommand{\black}{\color{black}}
\newcommand{\blue}{\color{blue}}
\newcommand{\green}{\color{green}}

\newtheorem{theorem}{Theorem}
\newtheorem{definition}{Definition}
\newtheorem{corollary}[theorem]{Corollary}
\newtheorem{lemma}[theorem]{Lemma}
\newtheorem{proposition}[theorem]{Proposition}
\newtheorem{conjecture}{Conjecture}

\newcommand{\ShannonSpace}[1]{\Gamma_{#1}}
\newcommand{\EntropySpace}[1]{\Gamma^*_{#1}}
\newcommand{\EntropyCone} [1]{\bar{\Gamma}^*_{#1}}
\newcommand{\NormalizedCone} [1]{\hat{\Gamma}^*_{#1}}
\newcommand{\NormalizedReducedCone} [1]{\hat{\mathcal{C}_{#1}}}
\newcommand{\ZhangYeungOuterBound}{\Gamma_4^{\mathrm{(1)}}}
\newcommand{\DoFrZeOuterBound}{\Gamma_4^{\mathrm{(2)}}}

\newcommand{\given}{$|$}
\newcommand{\semi}{,}
\newcommand{\Z}{\mathbf{Z}}
\newcommand{\HH}{h}
\newcommand{\Comment}[1]{& [\mbox{from  #1}]}
\newcommand{\R}{\mathbf{R}}
\newcommand{\Vamos}{V\'{a}mos}
\newcommand{\alphabet}{\mathcal{A}}

\title{Non-Shannon Information Inequalities in Four Random Variables}

\author{\IEEEauthorblockN{Randall Dougherty}
\IEEEauthorblockA{
Center for Communications Research\\
4320 Westerra Court\\
San Diego, CA 92121-1969\\
Email: rdough@ccrwest.org}\\
\IEEEauthorblockN{Chris Freiling}
\IEEEauthorblockA{
Department of Mathematics\\
California State University\\
5500 University Parkway\\
San Bernardino, CA 92407-2397\\
Email: cfreilin@csusb.edu}\\
\IEEEauthorblockN{Kenneth Zeger}
\IEEEauthorblockA{
Dept. of Electrical and Computer Eng.\\
University of California, San Diego\\
La Jolla, CA 92093-0407\\
Email: zeger@ucsd.edu}
 }

\maketitle

\begin{abstract}
Any unconstrained information inequality in three or fewer random variables can be written
as a linear combination of instances of Shannon's inequality $I(A;B|C) \ge 0$.
Such inequalities are sometimes
referred to as ``Shannon" inequalities.
In 1998, Zhang and Yeung gave the first example of
a ``non-Shannon" information inequality in four variables.  Their technique was to add
two auxiliary variables with special properties and then apply Shannon inequalities to
the enlarged list.  Here we will show that the Zhang-Yeung inequality can actually be derived from
just one auxiliary variable.  Then we use their same basic technique of adding auxiliary variables
to give many other non-Shannon inequalities in four variables.  Our list includes the inequalities found
by Xu, Wang, and Sun, but it is by no means exhaustive.
Furthermore, some of the inequalities obtained may be superseded by stronger inequalities that have
yet to be found.  Indeed, we show that the Zhang-Yeung inequality is one of those that is superseded.
We also present several infinite
families of inequalities.  This list includes some, but not all of the infinite families found by Matus.
Then we will give a description of what additional information these inequalities tell us about entropy space.
This will include a conjecture on the maximum possible failure of Ingleton's inequality.
Finally, we will present an application of non-Shannon inequalities to network coding.  We will demonstrate
how these inequalities are useful in finding bounds on the information that can flow through a particular network
called the Vamos network.
\end{abstract}

\footnotetext[2]{Mathematical Reviews: 26A12.}

\date{}
\maketitle

\section{Introduction}
\label{sec:information-inequalities}

For collections $A$, $B$, and $C$ of jointly related discrete random variables,
denote the \textit{entropy} of $A$ by
$$H(A):=\sum_{a\in A}-p(a)\log_2(p(a)), $$
where $0\log_2(0):=0$.
The \textit{conditional entropy} of $A$ given $B$ is defined by
\begin{eqnarray}
H(A|B)&:=&H(AB)-H(A) \label{condEntropy}\\
&=&\sum_{(a,b)\in AB}-p((a,b))\log_2(p((a,b))) \nonumber\\
&-& \sum_{a\in A}-p(a)\log_2(p(a)),\nonumber
\end{eqnarray}
the \textit{mutual information} between random variables $A$ and $B$ by
$$I(A;B):= H(A)+H(B)-H(AB),$$
and the \textit{conditional mutual information} between random variables
$A$ and $B$ given $C$ by
$$I(A;B|C):=H(AC)+H(BC)-H(C)-H(ABC).$$

The basic inequalities
$H(A)\ge 0$, $H(A|B)\ge 0$, and $I(A;B) \ge 0$
were originally proved in 1948 by Shannon \cite{Shannon-1948} 
and can all be obtained as special cases
(e.g. see \cite{Yeung-book}) 
of the inequality
\begin{align}
I(A;B|C) &\ge 0.                \label{eq:55}
\end{align}
For example, letting $C$ be a trivial random variable with just one element, we obtain
\begin{eqnarray*}
I(A;B|C) &=& H(AC)+H(BC)-H(C)-H(ABC)\\
&=&H(A)+H(B)-0-H(AB)\\
&=&I(A;B)
\end{eqnarray*}
yielding that $I(A;B)\geq 0$. Similarly, $H(A)=I(A;A)\geq 0$ and $H(A|B)=I(A;A|B)\geq 0$.

By combining together instances of these basic inequalities, other valid inequalities can be formed.
For example,
\begin{equation}\label{in:1}
I(A;B|C)+2I(A;C)\geq 0
\end{equation}
holds for any random variables $A$, $B$, $C$, since each of the terms separately is at least zero.
Applying the definitions,
this can be rewritten as
\begin{eqnarray}\label{in:2}
0 &\leq& H(AC)+H(BC)-H(ABC)-H(C)\\
\nonumber &&+2H(A)+2H(C)-2H(AC).
\end{eqnarray}
Canceling and rearranging the terms gives us
\begin{eqnarray}\label{in:3}
0 &\leq&H(BC)-H(ABC)+ H(C)\\
\nonumber &&+2H(A)-H(AC).
\end{eqnarray}
If we now permute the variables $A\to B$, $B\to C$, $C\to A$ then this becomes
\begin{eqnarray}\label{in:4}
0 &\leq&H(AC)-H(ABC)+ H(A)\\
\nonumber &&+2H(B)-H(AB).
\end{eqnarray}
Throughout our discussion, we will identify inequalities that can be derived from
each other using definitions, basic algebraic manipulation, and rearrangement of the variables.
Thus, the inequalities (\ref{in:1}), (\ref{in:2}), (\ref{in:3}), (\ref{in:4}) are really
just different forms of the same inequality.

Since all conditional entropies and all conditional mutual informations
can be written as linear combinations of joint entropies, we give the following
definition.
\begin{definition}
Let $n$ be a positive integer,
and let $S_1, \ldots, S_k$ be subsets of $\{1, \ldots, n\}$.
Let $\alpha_i\in\R$ for $1\le i \le k$.
An inequality of the form
$$
\alpha_1 H( \{A_i: i\in S_1\}) + \cdots + \alpha_k H( \{A_i: i\in S_k\}) \ge 0
$$
is called an \textit{information inequality}
if it holds for all jointly distributed random variables $A_1, \ldots  A_n$.
\label{def:1}
\end{definition}
The textbook \cite{Yeung-book} 
refers to information inequalities as
``the laws of information theory''.
As an example,
taking
$p=2$,
$S_1 = \{1\}$,
$S_2 = \{2\}$,
$S_3 = \emptyset$,
$S_4 = \{1,2\}$,
$\alpha_1 = \alpha_2 = 1$,
and $\alpha_4 = -1$,
one obtains
$H(A_1) + H(A_2) - H(A_1,A_2) \ge 0$,
which is an information inequality since it is always true
(this can be more succinctly expressed as $I(A_1;A_2) \ge 0$).

Information inequalities that can be derived
by adding special cases of Shannon's original
inequality will be given a special designation bearing his name.

\begin{definition}
A \textit{Shannon information inequality}
is any information inequality of the form
\begin{align}
\sum_i \alpha_i I(A_i;B_i|C_i) \ge 0
\label{eq:13.5}
\end{align}
where each $\alpha_i\geq 0$.
\end{definition}

Any information inequality that cannot be
expressed in the form \eqref{eq:13.5}
will be called a \textit{non-Shannon information inequality}.
It is known \cite[p.~308]{Yeung-book} 
that all information inequalities containing three or fewer random variables are Shannon inequalities.

These were the only known types of information inequalities until
Zhang and Yeung in 1998 published a ``non-Shannon" information inequality
\cite{Zhang-Yeung-July98} \cite[Theorem 14.7 on p.310]{Yeung-book}.
\begin{theorem}[Zhang-Yeung Theorem]
The following is a $4$-variable non-Shannon information inequality:
\begin{align*}
2I(C;D) &\le I(A;B){+}I(A;C,D){+}3I(C;D|A){+}I(C;D|B).
\end{align*}
\label{thm:Zhang-Yeung}
\end{theorem}

There are two parts to Theorem~\ref{thm:Zhang-Yeung}.  The first part claims that the inequality is
valid, and the other part claims that it is non-Shannon.
To prove this inequality is valid, Zhang and Yeung added two auxiliary variables
and then applied Shannon inequalities to the enlarged list of variables.  This
technique of adding auxiliary variables will be encapsulated by the
Copy Lemma ~\ref{copyLemma} given in Section~\ref{sec:Copy} below.   We will then give a new proof of the first part of the
Zhang-Yeung theorem in Section~\ref{ZY}. Unlike the original proof,
this proof requires only one auxiliary variable.  The second part of Theorem~\ref{thm:Zhang-Yeung}
can be proved using the Information Theory Inequality Prover (ITIP)(see \cite{ITIP}),
which is a MATLAB \cite{MATLAB} program for verifying testing whether an inequality is a Shannon inequality.
It was written and made freely available by Yeung and Yan.

Since the seminal work \cite{Zhang-Yeung-July98}, many other
non-Shannon iformation inequalities have been found.  See for example,
Ln\v{e}ni\v{c}ka \cite{Lnenicka03}.
Makarychev, Makarychev, Romashchenko, and Vereshchagin \cite{Makarychev-Makarychev-Romashchenko-Vereshchagin02},
Zhang \cite{Zhang03},
Zhang and Yeung \cite{Zhang-Yeung-Nov97},
Dougherty, Freiling, and Zeger \cite{SixNew},
and
Matus \cite{Matus}.

In this paper we will give the results of a systematic search for
additional four-variable non-Shannon inequalities using the same basic technique of Zhang and Yeung, as
given in the Copy Lemma.
In Section~\ref{sec:search} with a will present the general methodology
that was used for this search.  The depth of the search is measured by how many auxiliary variables were used
and how many instances of the Copy Lemma were used.  Ignoring inequalities that can be derived from others
by a permutation of the variables, it turns out that the Zhang-Yeung inequality is the only
one that can be derived from just one auxiliary variable.

In Section~\ref{TwoCopies} we exhaustively search
for all inequalities that can be derived from just two auxiliary variables.
The result of this search (after weeding out redundant inequalities that can be
deduced from the others) is the list of six two-copy-variable inequalities
that appear in ~\cite{SixNew}.

As the list of inequalities grows, it turns out that some of the earlier inequalities are no longer needed.
To see if an inequality is still necessary, we
temporarily remove it from the list and perform a linear program to see if the inequality in question can fail when all of the
others are satisfied.  If not, then it has been ``superseded" by the others and is trimmed from the list permanently.
The Zhang-Yeung inequality is an example of one that has been superseded.
Sharing the fate of the Zhang-Yeung inequality,
each of the six inequalities from Section~\ref{TwoCopies}
has also been superseded.

Next, we exhaustively searched for all
inequalities that can be derived using three auxiliary variables.  These are split into two sections.
In Section~\ref{ThreeCopies} we present those that only require two instances of the Copy Lemma, and in Section~\ref{ThreeCopiesB}
we give those that require three instances of the Copy Lemma.  In both of these sections, we list only the inequalities that have
not yet been superseded.  We also give a computer generated proof of each one.  To save space, the first few proofs
are given in detail while the latter proofs are abbreviated.

At this point, we were reaching the limit of what could be feasibly done with the resources available.  To speed things up,
the use of the ITIP program was eventually replaced with a faster C-program.
We began searching for inequalities using four variables and only three instances of the Copy
Lemma.  This search eventually finished; it took an estimated 50-100 CPU-years to complete.
We were not able to complete the entire exhaust over four auxiliary variables.
We estimate that this search would take about 70 times longer.  Our final list of inequalities
appears in Section~\ref{FourCopies}.  To save space, rather than giving the proof of each one, we provide the sequence of copy steps
used in the proof.

The first two infinite families of non-Shannon inequalities were found by Matus \cite{Matus}.
Each of these families is given in list form, indexed by the positive integers.  The first list was used
by Matus to prove the fundamental fact that no finite collection of linear inequalities will ever be able
to describe the entropy space completely.
In Section \ref{infinite} we also look at entire classes of inequalities.  These will be presented as rules, which
allow us to automatically generate new inequalities from old ones.  These rules can also be iterated leading to
uncountable collections of information inequalities.
These rules can also be used to generate many individual lists of inequalities.  As an example, we will show how to derive
the first list of Matus from these rules.  Matus' second list, however, was not uncovered by this process.

Additional information inequalities and a third infinite list have been discovered by Xu-Wang-Sun~\cite{XuWangSun}.
In Section~\ref{XWS} we will show  how to derive stronger versions of their inequalities
from the list in Section~\ref{ThreeCopiesB} and from the inequality rules of Section~\ref{infinite}.

In Section~\ref{structure} we will summarize what we have learned about the structure of entropy space.
We will give some volume computations showing that, in fact, these inequalities do not
seem to be closing the gap between the space satisfying the Shannon inequalities and the
space of known entropic vectors.  We present a certain probability distribution with just four atoms, that
we believe gives the maximum possible failure of
Ingleton's inequality.  We conclude that some new techniques will likely be necessary
in order to settle this ``four-atom" conjecture.

The Zhang-Yeung inequality has recently been applied to network coding
to demonstrate that Shannon information inequalities are in general insufficient
for computing the coding capacity of a network
(see \cite{Dougherty-Freiling-Zeger06-Matroids} and \cite{ChanGrant}).
In Section~\ref{Vamos} we will apply our list of inequalities to improve the information
theoretic upper bound for the coding capacity of a simple network called the Vamos network.


\section{Copy Lemmas}\label{sec:Copy}

There is only one known basic technique for coming up with new information inequalities.
\begin{enumerate}
\item Start with a set of arbitrary random variables.
\item Add auxiliary random variables
with special properties.
\item Apply known information inequalities to the enlarged set of random variables.
\end{enumerate}
In this section, we will present the methods for obtaining the auxiliary variables.
These are encapsulated in the following lemma, which is essentially due to Zhang and Yeung.


\begin{lemma}[Copy Lemma]\label{copyLemma}
Let $A$, $B$, $C$, $D$ be jointly distributed random variables.  Then there is another random variable
$R$, jointly distributed with $A,B,C,D$  with the following properties.

\begin{description}
\item [C1.] The marginal distributions of $(A,B,C)$ and $(A,B,R)$ are the same with $R$ replacing $C$.
\item [C2.] $I(CD;R|AB)=0$
\end{description}
In this case we say that $R$ is a $D$-copy of $C$ over $(A,B)$.
\end{lemma}
\begin{proof}
Let $A$, $B$, $C$, $D$, denote the alphabets of the random variables $A$, $B$, $C$, $D$ resp.
Let $a,b,c,d$ denote arbitrary elements of $A,B,C,D$, resp.  with probability $p(a,b,c,d)$.
Let $R$ be a new random variable and let $r$ denote an arbitrary element of its alphabet, which is $C$.
Define the joint probability distribution of $A,B,C,D,R$ by
$$
p'(a,b,c,d,r) = \frac{p(a,b,c,d)\sum_d p(a,b,r,d)}{\sum_{c,d}p(a,b,c,d)}.
$$
It is clear that these are nonnegative.
Summing over $r$ we get
\begin{eqnarray*}
\sum_r  p'(a,b,c,d,r) &=& \frac{p(a,b,c,d)\sum_{r,d} p(a,b,r,d)}{\sum_{c,d}p(a,b,c,d)}\\
&=& p(a,b,c,d)
\end{eqnarray*}
so that $p'$ is an extension of $p$, which also implies that the sum of all of the probabilities $p'$ is 1.
Similarly, the marginal distribution of $(A,B,R)$ is given by
\begin{eqnarray*}
\sum_{c,d}  p'(a,b,c,d,r) &=& \frac{\sum_{c,d} p(a,b,c,d)\sum_{d} p(a,b,r,d)}{\sum_{c,d}p(a,b,c,d)}\\
&=& \sum_{d} p(a,b,r,d)
\end{eqnarray*}
while the marginal distribution of $(A,B,C)$ is given by $\sum_{d} p(a,b,c,d)$, demonstrating (C1).

If we write (C2) in terms of entropies, we get
$H(ABCD)+H(ABR)-H(AB)-H(ABCDR)=0$.  But $H(A,B,R)=H(A,B,C)$ by (C1), so it remains to show that
$H(ABCDR)=H(ABCD)+H(ABC)-H(AB)$.  We compute $H(ABCDR)$ as
\begin{eqnarray*}
&=& \sum_{a,b,c,d,r} -p'(a,b,c,d,r)\log_2(p'(a,b,c,d,r))\\
&=&  \sum_{a,b,c,d,r} -p'(a,b,c,d,r)\log_2(p(a,b,c,d))\\
&+&  \sum_{a,b,c,d,r} -p'(a,b,c,d,r)\log_2(\sum_d p(a,b,r,d))\\
&-&  \sum_{a,b,c,d,r} -p'(a,b,c,d,r)\log_2(\sum_{c,d} p(a,b,c,d)).\\
\end{eqnarray*}
But
\begin{eqnarray*}
&& \sum_{a,b,c,d,r} -p'(a,b,c,d,r)\log_2(p(a,b,c,d))\\
&=&  \sum_{a,b,c,d} -p(a,b,c,d)\log_2(p(a,b,c,d))\\
&=& H(ABCD)
\end{eqnarray*}
\begin{eqnarray*}
&&\sum_{a,b,c,d,r} -p'(a,b,c,d,r)\log_2(\sum_d p(a,b,r,d))\\
&=&  \sum_{a,b,r} \sum_{c,d}-p'(a,b,c,d,r)\log_2(\sum_d p(a,b,r,d))\\
&=&  \sum_{a,b,r} \sum_d -p(a,b,r,d)\log_2(\sum_d p(a,b,r,d))\\
&=& H(ABC)
\end{eqnarray*}
\begin{eqnarray*}
&&  \sum_{a,b,c,d,r}-p'(a,b,c,d,r)\log_2(\sum_{c,d} p(a,b,c,d))\\
&=&  \sum_{a,b,} \sum_{c,d,r}-p'(a,b,c,d,r)\log_2(\sum_{c,d} p(a,b,c,d))\\
&-&  \sum_{a,b,} \sum_{d,r}-p(a,b,r,d)\log_2(\sum_{c,d} p(a,b,c,d))\\
&=& H(AB)
\end{eqnarray*}
Therefore, $H(ABCDR)=H(ABCD)+H(ABC)-H(AB)$ as desired.
\end{proof}


\section{The Zhang-Yeung Inequality}\label{ZY}

Zhang and Yeung were the first to discover non-Shannon inequalities.  Their original version of the inequality
took the following form.
$$
2I(C;D) \leq I(A;B)+I(A;CD)+3I(C;D|A) + I(C;D|B).
$$

Their proof used two copy lemmas.  Here we give
a proof that is not necessarily shorter, but it uses only one copy.
We will show below that the Zhang-Yeung inequality is the only inequality that
results from one instance of the copy lemma.

We will write down the proof twice.  The first proof makes it easier for the reader to
see at a glance, how the inequality is achieved by applying Shannon inequalities after adding an auxiliary variable.
However, it also takes longer to check it for accuracy since it requires tedious expansions of conditional informations
into entropies.  The second proof is easier to check for accuracy, but harder to get the general idea of the proof.
Future proofs will be given only the first way.

We also write the inequality
in a different, but equivalent form.  This is so it will match the pattern of the other inequalities
that we have found.  The equivalence of the two forms can be seen by swapping $C \leftrightarrow A$, $B \leftrightarrow  D$
and rewriting in terms of entropies.

\begin{theorem}[Zhang-Yeung Inequality]\label{th:ZY2}
Let $A$, $B$, $C$, $D$ be random variables.  Then
\begin{eqnarray*}
I(A;B) &\leq& 2I(A;B|C) + I(A;C|B) + I(B;C|A)\\
&+&I(A;B|D)+I(C;D)
\end{eqnarray*}

\end{theorem}
\begin{proof}(A)
By expanding mutual informations into entropies and canceling terms,
one can verify the following $6$-variable identity:
\begin{align*}
      & I(A;B)  & \\
       &\ \ + I(C;R|A)   &    \text{(S)} \\
       &\ \ + I(C;D|R)     & \text{(S)}  \\
       &\ \ + I(AB;R|CD)      &   \text{(S)}  \\
       &\ \ + I(D;R|B)       &    \text{(S)}  \\
       &\ \ + I(A;B|RD)     &  \text{(S)}  \\
       &\ \ + I(D;R|A)     &    \text{(S)}  \\
       &\ \ + I(R;C|B)     & \text{(S)}  \\
      &\ \  + I(A;B|CR)     &  \text{(S)}  \\
       &\ \  + I(C;R|ABD)     &  \text{(S)}  \\
   &= 2I(A;B|C)+I(A;C|B)+I(B;C|A)   \\
   &\ \ +I(A;B|D)+I(C;D)\\
       &\ \ + 2I(CD;R|AB)      &    \text{(C2)}  \\
      &\ \ + I(A;B|R)-I(A;B|C)      &  \text{(C1)}  \\
       &\ \ + I(A;R|B)-I(A;C|B)      &   \text{(C1)}  \\
       &\ \ + I(B;R|A)-I(B;C|A)     &  \text{(C1)}
\end{align*}
Each of the terms marked (S)
is a conditional mutual information and is therefore nonnegative.
Thus, if the terms marked (S) are erased and the ``$=$'' is replaced by ``$\le$'',
then we obtain a $5$-variable Shannon-type inequality.
By the Copy Lemma we may choose
$R$ to be a $D$-copy of $C$ over $AB$.
Then, the term marked (C2) is zero by condition (C2),
and each of the terms marked (C1) equals zero by condition (C1).
\end{proof}

In order to read the second proof, it will be convenient to review some basic properties of
entropies.  Each of these can be verified by rewriting in terms of entropies, and applying Shannon's
inequality \eqref{eq:55}, if necessary.
\begin{eqnarray}
I(A;B) &=& H(A)-H(A|B)\label{first}\\
I(A;B|C)&=& H(A|C)-H(A|BC)\label{second}\\
I(A;BC|D) &=& I(A;C|BD)+ I(A;B|D)\label{third}\\
H(A|BC) &\leq& H(A|B) \leq H(AC|B)\label{fourth}\\
I(A;B|CD) &\leq& I(A;B|C)\leq I(AD;B|C)\label{fifth}
\end{eqnarray}
It will also be convenient to note that property (C2) can be rewritten as
\begin{equation}\label{C2A}
H(R|AB)=H(R|ABCD).
\end{equation}
By combining this with \eqref{third} and \eqref{fifth} it follows that
\begin{equation}
I(A;C|B)=I(AR;C|B)\label{C2B}
\end{equation}
and
\begin{equation}\label{C2C}
I(B;R|A)=I(BD;R|A).
\end{equation}

\begin{proof}(\em B)
Let $R$ be a $D$-copy of $C$ over $(A,B)$.  Then:

\begin{tabbing}
\emph{(I)}:\\
\hspace{.2in} \= $2I(A;B|C)$  \hspace{2in}   \=  \\
=           \> $I(A;B|C){+}I(A;B|R)$      \> \ \emph{(C1)} \\
=           \> $H(A|C){-}H(A|BC){+}I(A;B|R)$    \> \ \eqref{second}           \\
=           \> $H(A|C) {-}H(A|B) {+}I(A;C|B) {+}I(A;B|R)$   \> \ \eqref{second}   \\
=           \> $H(A|C) {-} H(A|B) {+}I(AR;C|B) {+}I(A;B|R)$     \> \ \eqref{C2B} \\
$\geq$      \> $H(A|C) {-}H(A|B) {+}I(A;C|BR) {+}I(A;B|R)$ \> \ \eqref{fifth}       \\
=           \> $H(A|C) {-}H(A|B) {+}I(A;BC|R)$    \> \ \eqref{third}            \\
$\geq$      \> $H(A|C) {-}H(A|B) {+}I(A;C|R)$    \> \ \eqref{fifth}             \\
\>\>\\
\emph{(II)}:\\
            \> $I(A;C|B){+}I(B;C|A){+}I(A;B|D)$               \\
=           \> $I(A;R|B){+}I(B;R|A){+}I(A;B|D)$              \> \ \emph{(C1)} \\
=           \> $I(A;R|B){+}I(BD;R|A){+}I(A;B|D)$               \> \ \eqref{C2C}   \\
$\geq$      \> $I(A;R|B){+}I(B;R|AD){+}I(A;B|D)$            \> \ \eqref{fifth}     \\
=           \> $I(A;R|B){+}I(AR;B|D)$            \> \ \eqref{third}       \\
$\geq$      \> $I(A;R|B){+}I(R;B|D)$             \> \ \eqref{fifth}  \\
=           \> ${-}H(R|AB){+}H(R|B){+}I(R;B|D)$     \> \ \eqref{second}          \\
$\geq$      \> ${-}H(R|AB){+}H(R|BD){+}I(R;B|D)$     \> \ \eqref{fourth}             \\
=           \> ${-}H(R|AB){+}H(R|D)$                       \> \ \eqref{second}     \\
=           \> ${-}H(R|ABCD){+}H(R|D)$                         \> \ \eqref{C2A}   \\
$\geq$      \> ${-}H(R|CD) {+}H(R|D)$                  \> \ \eqref{fourth}       \\
=           \> $I(R;C|D)$                     \> \ \eqref{second}     \\
=           \> $H(C|D){-}H(C|DR)$              \> \ \eqref{second}          \\
$\geq$      \> $H(C|D){-}H(C|R)$               \> \ \eqref{fourth}           \\
\\
Finally,\\
            \> $2I(A;B|C){+}I(A;C|B){+}I(B;C|A)$\\
            \> ${+}I(A;B|D){+}I(C;D)$  \\
$\geq$      \> $2I(A;B|C){+}H(C|D){-}H(C|R){+}I(C;D)$  \> \ \emph{(II)}\\
=           \> $2I(A;B|C){+} H(C){-}H(C|R)$ \> \ \eqref{first}   \\
$\geq$      \> $H(A|C) {-}H(A|B) {+} I(A;C|R)$ \\
            \> ${+}H(C) {-}H(C|R)$  \> \ \emph{(I)}\\
=           \> $H(A|C) {-}H(A|B) {+} H(C) {-}H(C|AR)$ \> \ \eqref{second} \\
=           \> ${-}H(A|B) {+}H(AC) {-}H(C|AR)$ \> \ \eqref{condEntropy} \\
$\geq$      \> ${-}H(A|B) {+}H(AC) {-}H(C|A)$   \> \ \eqref{fourth}  \\
=            \> ${-}H(A|B) {+}H(A)$ \> \ \eqref{condEntropy}\\
=           \> $I(A;B)$\> \ \eqref{first}  \\
\end{tabbing}
\end{proof}


\section{Search Methodology}
\label{sec:search}

In this section we will describe the search method used to generate new inequalities.

We assume that we have a present list of known inequalities.  For example, at the start, this list is just the Shannon inequalities.
We consider the set of vectors, listed in the order
\begin{eqnarray}
&&\langle H(A),H(B),H(AB),H(C),H(AC),H(BC),\nonumber\\
&&H(ABC),H(D),H(AD),H(BD),H(ABD),\nonumber\\
&&H(CD),H(ACD),H(BCD),H(ABCD)\label{vertex} \rangle
\end{eqnarray}
that satisfy the current list of inequalities.
Since all of the inequalities are homogeneous, this set of vectors is closed
under multiplication by positive constants.
This set of satisfying vectors forms a region that is a polytopic cone in fifteen dimensions.
An alternate way to view this space is to hold $H(ABCD)$ at a fixed value of 1, and then
consider the satisfying vectors $\langle H(A), \ldots, H(BCD)\rangle$ which now form a polytope in
$\R^{14}$.

Using software, such as Fukuda's Cddlib \cite{cddlib}, we can generate a list
of extreme rays for the current region.  Two rays are called
equivalent if one can be obtained from the other under a permutation of the variables $A,B,C,D$.
For example, $0 0 0 1 1 1 1 1 1 1 1 1 1 1 1$ is equivalent to $0 1 1 0 0 1 1 1 1 1 1 1 1 1 1$
since the second one can be obtained from the first by swapping $B$ and $C$.
At the start of the process, when we just have the Shannon inequalities, there are twelve extreme rays, up to symmetry:
\begin{eqnarray*}
 0 0 0 0 0 0 0 0 0 0 0 0 0 0 0\\
 0 0 0 0 0 0 0 1 1 1 1 1 1 1 1\\
 0 0 0 1 1 1 1 1 1 1 1 1 1 1 1\\
 0 1 1 1 1 1 1 1 1 1 1 1 1 1 1\\
 1 1 1 1 1 1 1 1 1 1 1 1 1 1 1\\
 0 1 1 1 1 2 2 1 1 2 2 2 2 2 2\\
 1 1 2 1 2 2 2 1 2 2 2 1 2 2 2\\
 1 1 2 1 2 2 2 1 2 2 2 2 2 2 2\\
 2 1 2 1 2 2 2 1 2 2 2 2 2 2 2\\
 1 1 2 1 2 2 3 1 2 2 3 2 3 3 3\\
 1 1 2 1 2 2 3 2 3 3 3 3 3 3 3\\
 2 2 4 2 3 3 4 2 3 3 4 3 4 4 4
\end{eqnarray*}

Next, generate a list of possible one-variable copy specifications
that will create a new variable $R$.  Again with just one from each symmetry family these are:
\begin{align*}
      &27   \ \ \ \ \  8  \ \ \ \ \ 7        \\
        &36   \ \ \ \ \ 12   \ \ \ \ 3        \\
        &39   \ \ \ \ \ 14   \ \ \ \ 1        \\
        &63   \ \ \ \ \  4  \ \ \ \ \ 3        \\
        &66   \ \ \ \ \  6  \ \ \ \ \ 1        \\
       &75   \ \ \ \ \  2  \ \ \ \ \ 1
\end{align*}
Each of these codes up an instance of the Copy Lemma of the form ``$R$ is an $X$-copy of $Y$ over $Z$".
The first column is just an index and is irrelevant. The second and third columns
code the $Y$ and $Z$ respectively.
These variables are coded in binary, combining $1=A$, $2=B$, $4=C$, $8=D$.
So the line labeled by 27 gives us that $Y$ is $D$ and $Z$ is $ABC$, while the
line labeled 36 gives us that $Y$ is $CD$ and $Z$ is $AB$.
The variable $X$ can be deduced from the other two.  For each of the variables
$A$, $B$, $C$, $D$, if it does not appear in $Z$ and is not equal to $Y$ then it is included
in $X$.  Then the above reads as:
\begin{align*}
        &27   \ \ \ \ \ \text{$R$ is a copy of $D$ over $ABC$}      \\
        &36   \ \ \ \ \ \text{$R$ is a $CD$-copy of ($CD$) over $AB$ }      \\
        &39   \ \ \ \ \ \text{$R$ is a $BCD$-copy of ($BCD$) over $A$}       \\
        &63   \ \ \ \ \ \text{$R$ is a $D$-copy of $C$ over $AB$}       \\
        &66   \ \ \ \ \ \text{$R$ is a $BCD$-copy of ($BC$) over $A$}       \\
        &75   \ \ \ \ \ \text{$R$ is a $CD$-copy of $B$ over $A$}
\end{align*}
(Some possibilities were elimintated in advance: a last copy over no variables can
never produce any new inequalities, and we may assume that no variable
appears in both $Y$ and $Z$.)

Now go through a full search, testing each of the copy specifications against
all of the extreme rays (including all symmetric forms) to see whether any
contradictions are obtained.  In the current example, the results of the run are:
\begin{align*}
&27   \ \ \ \ \  8  \ \ \ \ \ 7    \ \ \ \ \ \text{--PASSED--}     \\
&36   \ \ \ \ \ 12   \ \ \ \ 3  \ \ \ \ \ \text{--PASSED--}     \\
       &39   \ \ \ \ \ 14   \ \ \ \ 1  \ \ \ \ \ \text{--PASSED--}     \\
 &\text{----} \ \ \ \ \ \text{2 2 3 2 3 3 4 2 4 3 4 3 4 4 4}\\
&\text{Inequality: -1  -2   3  -2   3   3  -4   0  -1   1   0   1   0  -1   0}  \\
        &63   \ \ \ \ \  4  \ \ \ \ \ 3        \\
               &66   \ \ \ \ \  6  \ \ \ \ \ 1  \ \ \ \ \ \text{--PASSED--}        \\
       &75   \ \ \ \ \  2  \ \ \ \ \ 1  \ \ \ \ \ \text{--PASSED--}
\end{align*}

The 'PASSED' entries are those copy specifications which did not contradict
any of the extreme rays.  The remaining copy specification,
"R is a D-copy of C over AB", was incompatible with the extreme
ray 2 2 3 2 3 3 4 2 4 3 4 3 4 4 4 (and possibly with others; the search
stops for a given copy specification when one extreme ray fails).
In other words, the equations forcing an entropy vector to lie on this
extreme ray, namely
\begin{eqnarray}
 \label{eq_A}     H(A) &=& 2H(ABCD)/4\\
 \label{eq_B}       H(B) &=& 2H(ABCD)/4\\
     H(AB) &=& 3H(ABCD)/4\\
      H(C) &=& 2H(ABCD)/4\\
     H(AC) &=& 3H(ABCD)/4\\
     H(BC) &=& 3H(ABCD)/4\\
    H(ABC) &=& 4H(ABCD)/4\\
\label{eq_D}      H(D) &=& 2H(ABCD)/4\\
     H(AD) &=& 4H(ABCD)/4\\
     H(BD) &=& 3H(ABCD)/4\\
\label{eq_ABD}     H(ABD) &=& 4H(ABCD)/4\\
     H(CD) &=& 3H(ABCD)/4\\
\label{eq_ACD}    H(ACD) &=& 4H(ABCD)/4\\
    H(BCD) &=& 4H(ABCD)/4
\end{eqnarray}
together with the copy specification requirements coming from the Copy Lemma, when
$R$ is a $D$-copy of $C$ over $AB$,
\begin{eqnarray}
       H(R) &=& H(C)\\
      H(AR) &=& H(AC)\\
      H(BR) &=& H(BC)\\
     H(ABR) &=& H(ABC)\\
 \label{eq_CDR}      I(CD;R|AB) &=& 0
\end{eqnarray}
and the Shannon inequalities, force all of the entropies to be 0.

>From this one can generate a new inequality.  The software has already gone
through this process to produce the inequality
\begin{eqnarray*}
   -H(A)-2H(B)+3H(AB)-2H(C)&&\\
   +3H(AC)+3H(BC)-4H(ABC)-H(AD)&&\\
   +H(BD)+H(CD)-H(BCD) &>=& 0,
\end{eqnarray*}
To explain this process in more detail, here is how one could
arrive at the same inequality using a program called ITIP \cite{ITIP}.
ITIP is a MATLAB \cite{MATLAB} program for testing Shannon inequalitites.  It was written
in made available by Raymond Yeung and Ying-On Yan.
One feeds ITIP an entropy inequality followed by a list of linear constraints.
ITIP returns ``True" if the inequality follows from the constraints and the
Shannon inequalities. Otherwise it returns ``Not provable by ITIP".

Start with:
\begin{tabbing}
IT\=IP(`$H(ABCD)<=0$',...\\
    \>`$4H(A)=2H(ABCD)$',`$4H(B)=2H(ABCD)$',...\\
    \> `$4H(AB)=3H(ABCD)$',`$4H(C)=2H(ABCD)$',...\\
    \>`$4H(AC)=3H(ABCD)$',`$4H(BC)=3H(ABCD)$',...\\
      \>`$4H(ABC)=4H(ABCD)$',`$4H(D)=2H(ABCD)$',...\\
      \>`$4H(AD)=4H(ABCD)$',`$4H(BD)=3H(ABCD)$',...\\
      \>`$4H(ABD)=4H(ABCD)$',...\\
      \>`$4H(CD)=3H(ABCD)$',...\\
      \>`$4H(ACD)=4H(ABCD)$',...\\
      \>`$4H(BCD)=4H(ABCD)$',`$I(R;CD|AB)=0$',...\\
      \>`$H(R)=H(C)$',`$H(AR)=H(AC)$',...\\
      \>`$H(BR)=H(BC)$',`$H(ABR)=H(ABC)$')\\
   True
\end{tabbing}

Note that each inequality and each constraint is contained in single quotes, and
they are separated by commas.  The ellipses at the end of each line is part of the
MATLAB syntax and merely denotes that the command is continued on the next line.

This first run of ITIP merely justifies what has already been claimed, namely
that ITIP can prove $H(ABCD)<=0$ from equations \eqref{eq_A} to
\eqref{eq_CDR}.

Next, we test (by trial and error) whether some of these constraint equations
can be eliminated.  We find that, in this case, equations \eqref{eq_A},
\eqref{eq_D}, \eqref{eq_ABD},  \eqref{eq_ACD} are not necessary.  The inquiry to ITIP
looks like:

\begin{tabbing}
IT\=IP(`$H(ABCD)<=0$',...\\
    \>`$4H(B)=2H(ABCD)$',`$4H(AB)=3H(ABCD)$',...\\
    \> `$4H(C)=2H(ABCD)$',`$4H(AC)=3H(ABCD)$',...\\
    \>`$4H(BC)=3H(ABCD)$',...\\
      \>`$4H(ABC)=4H(ABCD)$',...\\
      \>`$4H(AD)=4H(ABCD)$',`$4H(BD)=3H(ABCD)$',...\\
      \>`$4H(CD)=3H(ABCD)$',...\\
      \>`$4H(BCD)=4H(ABCD)$',`$I(R;CD|AB)=0$',...\\
      \>`$H(R)=H(C)$',`$H(AR)=H(AC)$',...\\
      \>`$H(BR)=H(BC)$',`$H(ABR)=H(ABC)$')\\
   True
\end{tabbing}

Next, one can eliminate each equation constraint by adding an appropriate multiple
of it (determined perhaps by trial and error; we modified ITIP to output
some additional information that would be helpful for this) to the
right-hand side of the inequality.  For example, for the first remaining constraint
\eqref{eq_B}, we first try using a multiple of zero:

\begin{tabbing}
IT\=IP(`$H(ABCD)<=-0H(B)+0H(ABCD)$',...\\
    \>`$4H(AB)=3H(ABCD)$',...\\
    \> `$4H(C)=2H(ABCD)$',`$4H(AC)=3H(ABCD)$',...\\
    \>`$4H(BC)=3H(ABCD)$',...\\
      \>`$4H(ABC)=4H(ABCD)$',...\\
      \>`$4H(AD)=4H(ABCD)$',`$4H(BD)=3H(ABCD)$',...\\
      \>`$4H(CD)=3H(ABCD)$',...\\
      \>`$4H(BCD)=4H(ABCD)$',`$I(R;CD|AB)=0$',...\\
      \>`$H(R)=H(C)$',`$H(AR)=H(AC)$',...\\
      \>`$H(BR)=H(BC)$',`$H(ABR)=H(ABC)$')\\
   Not provable by ITIP
\end{tabbing}
This of course fails since it was already determined that this equation was necessary.
Then we gradually raise the multiple until we get a success.  Thus,
\begin{tabbing}
IT\=IP(`$H(ABCD)<=-8H(B)+4H(ABCD)$',...\\
    \>`$4H(AB)=3H(ABCD)$',...\\
    \> `$4H(C)=2H(ABCD)$',`$4H(AC)=3H(ABCD)$',...\\
    \>`$4H(BC)=3H(ABCD)$',...\\
      \>`$4H(ABC)=4H(ABCD)$',...\\
      \>`$4H(AD)=4H(ABCD)$',`$4H(BD)=3H(ABCD)$',...\\
      \>`$4H(CD)=3H(ABCD)$',...\\
      \>`$4H(BCD)=4H(ABCD)$',`$I(R;CD|AB)=0$',...\\
      \>`$H(R)=H(C)$',`$H(AR)=H(AC)$',...\\
      \>`$H(BR)=H(BC)$',`$H(ABR)=H(ABC)$')\\
   True
\end{tabbing}

We repeat this for the next equation, finding
\begin{tabbing}
IT\=IP(`$H(ABCD)<=-8H(B)+4H(ABCD)$',...\\
    \>`$+12H(AB)-9H(ABCD)$',...\\
    \> `$4H(C)=2H(ABCD)$',`$4H(AC)=3H(ABCD)$',...\\
    \>`$4H(BC)=3H(ABCD)$',...\\
      \>`$4H(ABC)=4H(ABCD)$',...\\
      \>`$4H(AD)=4H(ABCD)$',`$4H(BD)=3H(ABCD)$',...\\
      \>`$4H(CD)=3H(ABCD)$',...\\
      \>`$4H(BCD)=4H(ABCD)$',`$I(R;CD|AB)=0$',...\\
      \>`$H(R)=H(C)$',`$H(AR)=H(AC)$',...\\
      \>`$H(BR)=H(BC)$',`$H(ABR)=H(ABC)$')\\
   True
\end{tabbing}

After passing through all of the equations not involving $R$,
we get:
\begin{tabbing}
IT\=IP([`$H(ABCD)<=-8H(B)+4H(ABCD)$',...\\
    \>`$+12H(AB)-9H(ABCD)$',...\\
    \> `$-8H(C)+4H(ABCD)+12H(AC)-9H(ABCD)$',...\\
    \>`$+12H(BC)-9H(ABCD)$',...\\
      \>`$-16H(ABC)+16H(ABCD)$',...\\
      \>`$-8H(AD)+8H(ABCD)+8H(BD)-6H(ABCD)$',...\\
      \>`$+8H(CD)-6H(ABCD)$',...\\
      \>`$-8H(BCD)+8H(ABCD)$'],`$I(R;CD|AB)=0$',...\\
      \>`$H(R)=H(C)$',`$H(AR)=H(AC)$',...\\
      \>`$H(BR)=H(BC)$',`$H(ABR)=H(ABC)$')\\
   True
\end{tabbing}

Note that brackets are used to enclose the first inequality when it does
not fit on a single line.
Thus ITIP gives a proof of a new inequality under the copy conditions.
We collect like terms and divide through by the common divisor, 4, to get:
\begin{tabbing}
IT\=IP([`$-2H(B)+3H(AB)-2H(C)+3H(AC)$',...\\
    \>`$+3H(BC)-4H(ABC)-2H(AD)+2H(BD)$',...\\
    \>`$+2H(CD)-2H(BCD)>=0$'],...\\
    \>`$I(R;CD|AB)=0$',`$H(R)=H(C)$',...\\
    \>`$H(AR)=H(AC)$',`$H(BR)=H(BC)$',...\\
    \>`$H(ABR)=H(ABC)$')
\\
   True
\end{tabbing}

For the last step, we see if this inequality can be strengthened by adding
one or more Shannon terms to the right-hand side.
This final inequality simplifies to the form
\begin{eqnarray*}
   -H(A)-2H(B)+3H(AB)-2H(C)&&\\
   +3H(AC)+3H(BC)-4H(ABC)-H(AD)&&\\
   +H(BD)+H(CD)-H(BCD) &>=& 0,
\end{eqnarray*}
mentioned earlier.

One can express this inequality (which is in fact the Zhang-Yeung inequality)
in the form
 \begin{eqnarray*}
   I(B;C) &<=& 2I(B;C|A)+I(B;A|C)+I(C;A|B)\\
   &&+I(B;C|D)+I(A;D).
   \end{eqnarray*}
Also, one can do further ITIP work to extract an intelligible proof of
this inequality from the Copy Lemma.

Once one has this inequality, one can produce a new list of extreme rays up to symmetry:
\begin{eqnarray}\label{newList}
0 0 0 0 0 0 0 0 0 0 0 0 0 0 0\\
\nonumber 0 0 0 0 0 0 0 1 1 1 1 1 1 1 1\\
\nonumber  0 0 0 1 1 1 1 1 1 1 1 1 1 1 1\\
\nonumber  0 1 1 1 1 1 1 1 1 1 1 1 1 1 1\\
\nonumber  1 1 1 1 1 1 1 1 1 1 1 1 1 1 1\\
\nonumber  0 1 1 1 1 2 2 1 1 2 2 2 2 2 2\\
\nonumber  1 1 2 1 2 2 2 1 2 2 2 1 2 2 2\\
\nonumber  1 1 2 1 2 2 2 1 2 2 2 2 2 2 2\\
\nonumber  2 1 2 1 2 2 2 1 2 2 2 2 2 2 2\\
\nonumber  1 1 2 1 2 2 3 1 2 2 3 2 3 3 3\\
\nonumber  1 1 2 1 2 2 3 2 3 3 3 3 3 3 3\\
\nonumber  3 2 4 4 5 6 6 4 5 5 6 5 6 6 6\\
\nonumber  3 3 6 2 4 4 6 4 5 5 6 5 6 6 6\\
\nonumber  3 3 6 3 4 5 6 3 5 4 6 5 6 6 6\\
\nonumber  3 3 6 4 5 5 6 4 5 5 6 5 6 6 6
\end{eqnarray}
A test of this list against the remaining one-variable copy specifications shows that
there are no further contradictions -- the Zhang-Yeung inequality is all
that one can get using one copy variable.

To get more inequalities, we try using two copy variables but just one copy step.
Then $RS$ is a $W$-copy of $(X)(Y)$ over $Z$.
The possible copy specifications are:
\begin{tabbing}
       \qquad \qquad \qquad \qquad \= 18630 \qquad \= 4 \qquad \= 8 \qquad \= 3\\
        \> 18666 \> 6 \> 10 \> 1\\
       \> 18675 \> 6 \> 8 \> 1\\
       \> 18837 \> 2 \> 12 \> 1\\
       \> 19566 \> 2 \> 4 \> 1\\
\end{tabbing}
As before, the first column is an index and is irrelevant.
The variables $X$, $Y$, and $Z$ are coded by the second, third, and fourth columns respectively.
The variable $W$ is deduced from the others by including any of $A$, $B$, $C$, $D$ which do not
appear in $Z$ and are not equal to $X$ or $Y$.  Thus, the above translates as:
\begin{tabbing}
    \qquad \= $RS$ is a copy of $CD$ over $AB$\\
    \>  $RS$ is a $BCD$-copy of ($BC$)($BD$) over $A$\\
     \> $RS$ is a $BC$-copy of ($BC$)$D$ over $A$\\
    \>  $RS$ is a $CD$-copy of $B(CD)$ over $A$      [redundant]\\
    \>  $RS$ is a $D$-copy of $BC$ over $A$\\
\end{tabbing}
It turns out that these also do not give any more than the Zhang-Yeung inequality.

For two copy variables in two steps, there is a much longer list of possibilities
(753 of them).  A first pass against the Shannon extreme rays reduces this list
to 198, and then a run against the Zhang-Yeung extreme rays shows that only
15 of the copy specifications give contradictions (new inequalities):

\begin{tabbing}
 \quad \=---- \=4  4  5  3  5  5  6  2  5  6  6  4  6  6  6 \\
 \> \>Inequality: -7  -2   8  -4   9   5  -9   0   3  -2   0   3  -4   0   0\\
 \> \>4636 \qquad 4 \qquad 3 \qquad 1 \qquad 28 \\
 \>\> $R$ is a $D$-copy of $C$ over $AB$\\
  \>\>$S$ is a $B$-copy of $A$ over $CDR$\\
\end{tabbing}
\begin{tabbing}
 \quad \=---- \=4 2 5 3 5 4 6 4 5 6 6 5 6 6 6 \\
 \> \>Inequality: -7  -1  7  -4   9   5  -9   -1   4  -2   0   3  -4   0   0\\
 \> \>4716 \qquad 4 \qquad 3 \qquad 4 \qquad 25 \\
 \>\> $R$ is a $D$-copy of $C$ over $AB$\\
  \>\>$S$ is a $B$-copy of $C$ over $ADR$\\
\end{tabbing}
\begin{tabbing}
 \quad \=---- \=4 2 5 3 5 4 6 3 5 4 6 6 6 6 6 \\
 \> \>Inequality: -6  -3   7  -2   7   5  -8   -1   4  3  4   -2  0   0   0\\
 \> \>4788 \qquad 4 \qquad 3 \qquad 2 \qquad 25 \\
\end{tabbing}
\begin{tabbing}
 \quad \=---- \=2 4 5 3 4 5 6 3 4 5 6 6 6 6 6 \\
 \> \>Inequality: -3 -6 7 -3 6 8 -9 -1 3 4 -4 -2 0 0 0\\
 \> \>4798 \qquad 4 \qquad 3 \qquad 1 \qquad 24 \\
\end{tabbing}
\begin{tabbing}
 \quad \=---- \=4 4 5 3 5 5 6 2 5 6 6 4 6 6 6  \\
 \> \>Inequality: -17 -2 16 -8 21 7 -17 -2 9 -4 0 9 -12 0 0 \\
 \> \>5364 \qquad 4 \qquad 3 \qquad 8 \qquad 21 \\
\end{tabbing}
\begin{tabbing}
 \quad \=---- \=4 4 5 3 5 5 6 3 5 5 6 6 6 6 6\\
 \> \>Inequality: -7 -7 11 -1 6 6 -8 -1 6 6 -8 -3 0 0 0 \\
 \> \>5508 \qquad 4 \qquad 3 \qquad 8 \qquad 19 \\
\end{tabbing}
\begin{tabbing}
 \quad \=---- \=2 4 5 3 4 5 6 3 4 5 6 6 6 6 6 \\
 \> \>Inequality: -3 -6 7 -3 6 8 -9 -1 3 4 -4 -2 0 0 0 \\
 \> \>5526 \qquad 4 \qquad 3 \qquad 8 \qquad 17 \\
\end{tabbing}
\begin{tabbing}
 \quad \=---- \=3 2 4 4 5 5 6 4 5 6 6 5 6 6 6 \\
 \> \>Inequality: -6 -1 7 -5 9 6 -10 0 2 -2 0 2 -2 0 0 \\
 \> \>6084 \qquad 4 \qquad 3 \qquad 2 \qquad 21 \\
\end{tabbing}
\begin{tabbing}
 \quad \=---- \=3 3 4 5 4 5 5 6 2 4 6 6 5 6 6 6 \\
 \> \>Inequality: -3 -4 7 -6 7 10 -11 0 2 -2 0 2 -2 0 0 \\
 \> \>6094 \qquad 4 \qquad 3 \qquad 1 \qquad 20 \\
\end{tabbing}
\begin{tabbing}
 \quad \=---- \=3 4 5 2 4 5 6 4 5 5 6 6 6 6 6 \\
 \> \>Inequality: -5 -5 8 -1 6 6 -9 0 2 2 -2 -2 0 0 0\\
 \> \>6156 \qquad 4 \qquad 3 \qquad 4 \qquad 19 \\
\end{tabbing}
\begin{tabbing}
 \quad \=---- \= 3 4 5 4 5 5 6 2 4 6 6 5 6 6 6\\
 \> \>Inequality: -3 -4 7 -6 7 10 -11 0 2 -2 0 2 -2 0 0 \\
 \> \>6174 \qquad 4 \qquad 3 \qquad 4 \qquad 17 \\
\end{tabbing}
\begin{tabbing}
 \quad \=---- \= 2 4 5 3 4 5 6 3 4 5 6 6 6 6 6 \\
 \> \>Inequality: -3 -5 6 -2 5 6 -7 -1 3 4 -4 -2 0 0 0 \\
 \> \>11358 \qquad 4 \qquad 3 \qquad 16 \qquad 9 \\
\end{tabbing}
\begin{tabbing}
 \quad \=---- \= 4 4 5 3 5 5 6 3 5 5 6 6 6 6 6 \\
 \> \>Inequality:-8 -9 14 -1 7 8 -11 -1 6 6 -8 -3 0 0 0 \\
 \> \>12069 \qquad 4 \qquad 3 \qquad 24 \qquad 3 \\
\end{tabbing}
\begin{tabbing}
 \quad \=---- \= 3 4 5 4 5 5 6 2 4 6 6 5 6 6 6 \\
 \> \>Inequality: -3 -3 6 -5 6 8 -9 0 2 -2 0 2 -2 0 0 \\
 \> \>12654 \qquad 4 \qquad 3 \qquad 16 \qquad 5 \\
\end{tabbing}
\begin{tabbing}
 \quad \=---- \= 3 4 5 2 4 5 6 4 5 5 6 6 6 6 6  \\
 \> \>Inequality: -5 -5 8 -1 6 6 -9 0 2 2 -2 -2 0 0 0 \\
 \> \>12717 \qquad 4 \qquad 3 \qquad 20 \qquad 3 \\
\end{tabbing}

One can now continue the iteration, adding these new inequalities to the
known list, producing a new set of extreme rays, re-running them against
the 15 remaining copy specifications, and repeating until no contradictions
remain.  The result (after weeding out redundant inequalities that can be
deduced from the others) is the list of six two-copy-variable inequalities
given in our previous paper.


\section{Inequalities Using Two Copies}\label{TwoCopies}

In our conference paper, we gave six additional inequalities that used two copy lemmas to prove.
For this article, we will rewrite these inequalities in a different form that matches our other inequalities.
We summarize as follows:
\begin{eqnarray*}
2I(A;B) &\leq& aI(A;B|C) + bI(A;C|B) + cI(B;C|A) \\
      &+&  dI(A;B|D) + eI(A;D|B)+ fI(B;D|A)\\
      &+& gI(C;D)
\end{eqnarray*}
for each of the follwing values of $(a,b,c,d,e,f,g)$.
\begin{eqnarray}
(a,b,c,d,e,f,g) &=& (5,3,1,2,0,0,2) \label{Matus2}\\
(a,b,c,d,e,f,g) &=& (4,2,1,3,1,0,2) \label{survivor}\\
(a,b,c,d,e,f,g) &=& (4,4,1,2,1,1,2) \\
(a,b,c,d,e,f,g) &=& (3,3,3,2,0,0,2) \\
(a,b,c,d,e,f,g) &=& (3,4,2,3,1,0,2) \\
(a,b,c,d,e,f,g) &=& (3,2,2,2,1,1,2)
\end{eqnarray}
These are the only additional inequalities that are obtained from this technique using two copy lemmas.
As it turns out, all of these inequalities except for \eqref{survivor} have since been superseded by
other inequalities.  We give here a proof only for \eqref{survivor}.

\begin{lemma}\label{le:1}
The following is a 5-variable information inequality:
\begin{eqnarray*}
    I(A;B) &\leq& I(A;B|C)+I(A;B|D)+I(C;D)\\
    &+& I(A;D|R)+I(A;R|D)+I(D;R|A).
\end{eqnarray*}
\end{lemma}
\begin{proof}
By expanding mutual informations into entropies and cancelling
terms, one can verify the following 6-variable identity:
\begin{align*}
      & I(A;B)  & \\
       &\ \ + I(A;B|CS)   &                                                           \text{(S)} \\
       &\ \ + I(B;S|A)     &                                                          \text{(S)}  \\
       &\ \ + I(B;S|D)      &                                                          \text{(S)}  \\
       &\ \ + I(A;D|BS)       &                                                       \text{(S)}  \\
       &\ \ + I(C;D|S)     &                                                           \text{(S)}  \\
       &\ \ + I(C;S|A)     &                                                           \text{(S)}  \\
       &\ \ + I(C;S|B)     &                                                           \text{(S)}  \\
      &\ \  + I(ABR;S|CD)     &                                                        \text{(S)}  \\
      &\ \  + I(CR;S|ABD)     &                                                        \text{(S)} \\
      &\ \  + I(DR;S|ABC)     &                                                        \text{(S)}  \\
   &= I(A;B|C)+I(A;B|D)+I(C;D)   \\
   &\ \ +I(A;D|R)+I(A;R|D)+I(D;R|A)\\
       &\ \ + 3I(S;BCR|AD)      &                                                      \text{(C2)}  \\
       &\ \ - (H(S)-H(R))     &                                                        \text{(C1)}  \\
       &\ \ + 2(H(SA)-H(RA))      &                                                    \text{(C1)}  \\
       &\ \ + 2(H(SD)-H(RD))      &                                                    \text{(C1)}  \\
       &\ \ - 3(H(SAD)-H(RAD))     &                                                   \text{(C1)}
\end{align*}
Each of the terms in the lines marked (S) is a conditional mutual
information and is therefore nonnegative by Shannon's inequality \eqref{eq:55}.  Thus, if the terms marked (S)
are erased and the ``=" is replaced by ``$\leq$", then we obtain a 6-variable
Shannon-type inequality.  By the Copy Lemma, we may choose the random variable S such
that S is a BC-copy of R over AD.  Then the term marked (C2) is zero
by condition (C2), and each of the terms marked (C1) equals zero
by condition (C1).
\end{proof}

\begin{theorem}
The following is a 4-variable information inequality:
\begin{eqnarray*}
   2I(A;B) &\leq& 4I(A;B|C)+2I(A;C|B)+I(B;C|A)\\
   &+&3I(A;B|D)+I(A;D|B)+2I(C;D).
\end{eqnarray*}
\end{theorem}
\begin{proof}
By expanding mutual informations into entropies and cancelling
terms, one can verify the following 5-variable identity:
\begin{align*}
 & 2I(A;B) & \\
      &\ \ + I(A;B|CR)   &   \text{(S)} \\
      &\ \ + 3I(A;B|DR)   &   \text{(S)} \\
     &\ \  + I(C;D|R)      &   \text{(S)} \\
     &\ \  + I(D;R|ABC)     &   \text{(S)} \\
     &\ \  + 3I(D;R|B)       &   \text{(S)} \\
     &\ \  + I(AB;R|CD)        &   \text{(S)} \\
     &\ \  + 3I(C;R|ABD)       &   \text{(S)} \\
     &\ \  + I(C;R|A)          &   \text{(S)} \\
     &\ \  + I(C;R|B)         &   \text{(S)} \\
     &\ \  + [-I(A;B)+I(A;B|C)+I(A;B|D) &\\
     &\ \ +I(C;D)+I(A;D|R)+I(A;R|D)+I(D;R|A)]    &   \text{(L)} \\
   & = 4I(A;B|C)+2I(A;C|B)+I(B;C|A) \\
    &\ \ +3I(A;B|D)+I(A;D|B)+2I(C;D) \\
     &\ \  + 5I(R;CD|AB)     &   \text{(C2)} \\
    &\ \   - 2(H(R)-H(C))    &   \text{(C1)} \\
    &\ \   + 3(H(RA)-H(CA))  &   \text{(C1)} \\
     &\ \  + 4(H(RB)-H(CB))   &   \text{(C1)} \\
     &\ \  - 5(H(RAB)-H(CAB)) &   \text{(C1)} \\
\end{align*}
The expression enclosed in brackets ``[]" immediately before the ``=" sign
and marked with (L) is nonnegative by Lemma \ref{le:1}.
Each of the terms in the lines marked (S) is a conditional mutual
information and is therefore nonnegative.  Thus, if the terms marked (S)
are erased and the "=" is replaced by "$\le$", then we obtain a 5-variable
Shannon-type inequality.  By the Copy Lemma, we may choose R such
that R is a D-copy of C over AB.  Then the term marked (C2) is zero
by condition (C2), and each of the terms marked (C1) equals zero
by condition (C1).
\end{proof}


\section{Inequalities Using Three Copy Variables with at most Two Copy Steps}\label{ThreeCopies}

In this section and the next we will present all of the information inequalities that can be proved from this technique
using three copies.  Some of these inequalities have been superseded by other inequalities, and won't be
listed here.  In this section we will concentrate on those that involve only two instances of the Copy Lemma.
We will start out by giving the proof of the first inequality in full detail.  The proofs are all very similar.  Therefore, in order to
save space, as we go further down the list we will gradually take shortcuts, eventually giving only abbreviated proofs.
However, enough detail will given in these outlines so that the full proofs can be reconstructed.

\begin{theorem}  \label{thm:many-inequalities}
The following is an information inequality
\begin{eqnarray*}
aI(A;B) &\leq& bI(A;B|C) + cI(A;C|B) + dI(B;C|A) \\
      &+&  eI(A;B|D) + fI(A;D|B)+ gI(B;D|A)\\
      &+& hI(C;D) + iI(C;D|A)
\end{eqnarray*}
for each of the following values of $(a,b,c,d,e,f,g,h,i)$.
\begin{align}
&&(3&,&4&,&4&,&4&,&3&,&1&,&1&,&3&,&0)&&&& \label{ineq3.1} \\
&&(3&,&9&,&6&,&1&,&3&,&0&,&0&,&3&,&0)&&&& \label{ineq3.2}\\
&&(3&,&4&,&6&,&6&,&3&,&0&,&0&,&3&,&0)&&&& \label{ineq3.3}\\
&&(2&,&3&,&3&,&1&,&5&,&2&,&0&,&2&,&0)&&&& \label{ineq3.4}\\
&&(3&,&4&,&3&,&3&,&3&,&3&,&3&,&3&,&0)&&&& \label{ineq3.5}\\
&&(3&,&6&,&3&,&1&,&6&,&3&,&0&,&3&,&0)&&&& \label{ineq3.6}\\
&&(4&,&5&,&8&,&8&,&4&,&1&,&1&,&4&,&0)&&&& \label{ineq3.7}\\
&&(2&,&4&,&2&,&1&,&2&,&0&,&0&,&2&,&3)&&&& \label{ineq3.8}\\
&&(4&,&5&,&5&,&5&,&4&,&4&,&4&,&4&,&0)&&&& \label{ineq3.9}\\
&&(2&,&3&,&3&,&2&,&2&,&0&,&0&,&2&,&0)&&&& \label{ineq3.10}\\
&&(3&,&7&,&5&,&1&,&3&,&1&,&1&,&3&,&0)&&&& \label{ineq3.11}\\
&&(4&,&6&,&4&,&3&,&4&,&2&,&1&,&4&,&0)&&&& \label{ineq3.12}\\
&&(2&,&5&,&2&,&1&,&2&,&0&,&0&,&2&,&0)&&&& \label{ineq3.13}\\
&&(2&,&4&,&3&,&1&,&2&,&0&,&0&,&2&,&0)&&&& \label{ineq3.14}\\
&&(2&,&4&,&1&,&2&,&2&,&3&,&0&,&2&,&0)&&&& \label{ineq3.15}\\
&&(2&,&4&,&2&,&1&,&2&,&4&,&1&,&2&,&0)&&&& \label{ineq3.16}
\end{align}
\end{theorem}
\begin{proof}
{\bf Proof of (\ref{ineq3.1}):}
We first show the following 5-variable information inequality:
\begin{eqnarray}\label{lem3.1}
   2I(A;B) &\le& I(AB;R)+I(A;B|C)+I(A;B|D) \nonumber\\
      && +2I(A;B|CR)+2I(A;B|DR)+I(A;C|B)\nonumber\\
      && +I(A;D|B)+I(B;C|A)+I(B;D|A)\nonumber\\
      && +I(C;D)+2I(C;D|R).
\end{eqnarray}
To see this, verify the following 7-variable identity by expanding into entropies
and canceling:
\begin{align*}
&2I(A;B)&\\
      &\ \  + I(CR;S|A)                                                         & \text{(S)} \\
      &\ \  + I(CR;S|B)                                                         & \text{(S)} \\
      &\ \  + I(DR;S|A)                                                         & \text{(S)} \\
      &\ \  + I(DR;S|B)                                                         & \text{(S)} \\
      &\ \  + I(CD;ST|ABR)                                                      & \text{(S)} \\
      &\ \  + I(A;B|CRS)                                                        & \text{(S)} \\
      &\ \  + I(A;B|DRS)                                                        & \text{(S)} \\
      &\ \  + I(AB;S|CDR)                                                       & \text{(S)} \\
      &\ \  + I(C;D|RS)                                                         & \text{(S)} \\
      &\ \  + I(D;S|ABCR)                                                       & \text{(S)} \\
      &\ \  + I(C;S|ABDR)                                                       & \text{(S)} \\
      &\ \  + I(CR;T|A)                                                         & \text{(S)} \\
      &\ \  + I(CR;T|B)                                                         & \text{(S)} \\
      &\ \  + I(DR;T|A)                                                         & \text{(S)} \\
      &\ \  + I(DR;T|B)                                                         & \text{(S)} \\
      &\ \  + I(A;B|CRT)                                                        & \text{(S)} \\
      &\ \  + I(A;B|DRT)                                                        & \text{(S)} \\
      &\ \  + I(C;D|RT)                                                         & \text{(S)} \\
      &\ \  + I(AB;T|CDR)                                                       & \text{(S)} \\
      &\ \  + + 3I(CDR;T|ABS)                                                     & \text{(S)} \\
      &\ \  + I(CDR;S|ABT)                                                      & \text{(S)} \\
      &\ \  + I(S;T|R)                                                          & \text{(S)} \\
      &\ \  + I(AB;R|ST)                                                        & \text{(S)} \\
      &\ \  + I(CR;S|ABT)                                                       & \text{(S)} \\
      &\ \  + I(DR;S|ABT)                                                       & \text{(S)} \\
      &\ \  + I(C;ST|ABDR)                                                      & \text{(S)} \\
      &\ \  + I(D;ST|ABCR)                                                      & \text{(S)} \\
      & = I(AB;R)+I(A;B|C)+I(A;B|D)&\\
      & +2I(A;B|CR)+2I(A;B|DR)+I(A;C|B)&\\
      & +I(A;D|B) +I(B;C|A)+I(B;D|A)&\\
      & +I(C;D)+2I(C;D|R)&\\
      &\ \   + 7I(ST;CDR|AB)                                                     & \text{(C2)} \\
      &\ \   - (H(ST)-H(CD))                                                     & \text{(C1)} \\
      &\ \   + 2(H(SA)-H(CA))                                                    & \text{(C1)} \\
      &\ \   + 2(H(TA)-H(DA))                                                    & \text{(C1)} \\
      &\ \   + 2(H(SB)-H(CB))                                                    & \text{(C1)} \\
      &\ \   + 2(H(TB)-H(DB))                                                    & \text{(C1)} \\
      &\ \   - 3(H(SAB)-H(CAB))                                                  & \text{(C1)} \\
      &\ \   - 3(H(TAB)-H(DAB))                                                  & \text{(C1)} \\
\end{align*}
Each of the terms in the lines marked (S) is a conditional mutual
information and is therefore nonnegative.  Thus, if the terms marked (*S)
are erased and the "=" is replaced by "$\leq$", then we obtain a 7-variable
Shannon-type inequality.  By the Copy Lemma, we may choose ST such
that ST is a R-copy of CD over AB.  Then the term marked (C2) is zero
by condition (C2), and each of the terms marked (C1) equals zero
by condition (C1).

With \eqref{lem3.1} verified, we now continue with the proof of \eqref{ineq3.1}.
By expanding mutual informations into entropies and cancelling
terms, one can verify the following 5-variable identity:
\begin{align*}
&3I(A;B)&\\
      &\ \ + 2I(C;R|A)                                                         & \text{(S)} \\
      &\ \ + 2I(C;R|B)                                                         & \text{(S)} \\
      &\ \ + 2I(D;R|A)                                                         & \text{(S)} \\
      &\ \ + 2I(D;R|B)                                                         & \text{(S)} \\
      &\ \ + 2I(AB;R|CD)                                                       & \text{(S)} \\
      &\ \ + 2I(D;R|ABC)                                                       & \text{(S)} \\
      &\ \ + 2I(C;R|ABD)                                                       & \text{(S)} \\
      &\ \ + [-2I(A;B)+I(AB;R)+I(A;B|C)\\
      &\ \ +I(A;B|D)+2I(A;B|CR)+2I(A;B|DR)\\
      &\ \ +I(A;C|B)+I(A;D|B)+I(B;C|A)\\
      &\ \ +I(B;D|A)+I(C;D)+2I(C;D|R)]                  & \text{(L)} \\
    & = 4I(A;B|C)+4I(A;C|B)+4I(B;C|A)\\
      &\ \ +3I(A;B|D)+I(A;D|B)+I(B;D|A)\\
      &\ \ +3I(C;D)\\
      &\ \ + 6I(R;CD|AB)                                                       & \text{(C2)} \\
      &\ \ - (H(R)-H(C))                                                       & \text{(C1)} \\
      &\ \ + 4(H(RA)-H(CA))                                                   & \text{(C1)} \\
      &\ \ + 4(H(RB)-H(CB))                                                    & \text{(C1)} \\
      &\ \ - 7(H(RAB)-H(CAB))                                                 & \text{(C1)} \\
\end{align*}
The expression in brackets ``[]" marked (L) is nonnegative by Lemma 11.1.
Each of the terms in the lines marked (S) is a conditional mutual
information and is therefore nonnegative.  Thus, if the terms marked (S)
are erased and the "=" is replaced by "$\le$", then we obtain a 5-variable
Shannon-type inequality.  By the Copy Lemma, we may choose R such
that R is a D-copy of C over AB.  Then the term marked (C2) is zero
by condition (C2), and each of the terms marked (C1) equals zero
by condition (C1).    This finishes the proof of \eqref{ineq3.1}.

{\bf Proof of (\ref{ineq3.2}):}
The proof of each of the remaining inequalities has the same form as the proof
of \eqref{ineq3.1}.  Therefore,
we will gradually take more and more shortcuts.  For each inequality there are two equations.
The first equation proves a lemma, and the second proves the inequality.
Each line marked with (S), (L), (C1), or (C2) can be ignored because it is a
Shannon inequality, an instance of the lemma, or it follows from (C1) or (C2) resp.
The inequalities are then formed by ignoring these lines and turning the
``=" sign into a ``$\le$" sign.

We now give an outline of the proof of \eqref{ineq3.2}:

\begin{align*}
& 2I(A;B)&\\
      &\ \ + I(A;B|CT)
      + I(A;B|DT)
      + I(A;B|RT)                                                          & \text{(S)} \\
      &\ \ + I(A;R|ST)
      + I(A;T|CD)
      + I(ACDR;T|BS)                                                       & \text{(S)} \\
      &\ \ + I(C;D|T)
      + I(C;T|B)
      + I(CRS;T|ABD)                                                       & \text{(S)} \\
      &\ \ + I(CDS;T|ABR)
      + I(C;T|ARS)
      + I(B;S|T)                                                           & \text{(S)} \\
      &\ \ + I(B;T|ACDRS)
      + I(D;RS|ACT)
      + I(D;T|A)                                                           & \text{(S)} \\
      &\ \ + I(D;T|B)
      + I(DRS;T|ABC)
      + I(R;T|A)                                                           & \text{(S)} \\
      &\ \ + I(R;T|B)
      + I(S;T|A)
      + I(S;T|R)                                                           & \text{(S)} \\
   & = I(A;B|C)+I(A;B|D)+3I(A;B|R) &\\
   &\ \ +2I(A;R|B)+I(A;R|S)+I(C;D) &\\
   &\ \ +I(B;R|A)+I(B;S) +I(D;RS|AC) &\\
       &\ \ + 5I(T;CDRS|AB)                                              & \text{(C2)} \\
       &\ \ - 2(H(T)-H(R))
        + 3(H(TA)-H(RA))    & \text{(C1)} \\
        &\ \ + 4(H(TB)-H(RB))
        - 5(H(TAB)-H(RAB))                                                 & \text{(C1)} \\
\end{align*}
where $T$ is a $CDS$-copy of $R$ over $AB$.  Then,
\begin{align*}
    &  3I(A;B) &\\
       &\ \ + 6I(A;C|BR)
      + I(A;B|DR)
      + 6I(B;R|C)                                                          & \text{(S)}\\
       &\ \ + I(D;R|A)
      + I(C;D|R)
      + I(AB;R|CD)                                                         & \text{(S)}\\
       &\ \ + I(D;R|B)
      + I(C;R|ABD)
      + I(B;RS|ACD)                                                        & \text{(S)}\\
       &\ \ + 6I(D;R|ABC)
      + I(C;S|B)
      + I(AR;B|CS)                                                         & \text{(S)}\\
       &\ \ + 8I(BD;S|ACR)
      + I(D;RS|ABC)                                                        & \text{(S)}\\
       &\ \ + [-2I(A;B)+I(A;B|C)+I(A;B|D) & \text{(L)}\\
       &\ \ +3I(A;B|R)+2I(A;R|B)+I(A;R|S) & \text{(L)}\\
       &\ \ +I(C;D)+I(B;R|A)+I(B;S)       & \text{(L)}\\
       &\ \ +I(D;RS|AC)]                                   & \text{(L)}\\
   &= 9I(A;B|C)+6I(A;C|B)+I(B;C|A)\\
     &\ \ +3I(A;B|D)+3I(C;D) \\
      &\ \ + 10I(RS;BD|AC)                                                    & \text{(C2)} \\
      &\ \ - 4(H(R)-H(B))
      + (H(RS)-H(BD))               & \text{(C1)}\\
       &\ \ + 5(H(RA)-H(BA))
       + (H(SA)-H(DA))               & \text{(C1)}\\
      &\ \ - (H(RSA)-H(BDA))
      + 7(H(RC)-H(BC))               & \text{(C1)}\\
      &\ \ - (H(SC)-H(DC))
      - 8(H(RAC)-H(BAC))                                                  & \text{(C1)}\\
\end{align*}
where $RS$ is a copy of $BD$ over $AC$.
This finishes the proof of \eqref{ineq3.2}.

{\bf Shortened Proof of (\ref{ineq3.2}):}
In the proof of the above lemma there are some lines marked (S).
These lines make up a Shannon inequality.
On the other side of the equation, there is a line marked (C2).  Except for the coefficient,
this comes directly from the Copy Lemma (Lemma~\ref{copyLemma}).  Finally, there are some lines
marked (C1).  These lines are zero by the Copy Lemma, and are easily deduced from the rest of the
equation.  We will now rewrite the proof in a more abbreviated form:

Let $T$ be a $CDS$-copy of $R$ over $AB$.  Let (L) be non-negative expression
\begin{align*}
       &\ \ -2I(A;B)+I(A;B|C)+I(A;B|D) +3I(A;B|R)\\
       &\ \ +3I(A;B|R)+2I(A;R|B)+I(A;R|S)+I(C;D) \\
       &\ \ +I(B;R|A) +I(B;S)  +I(D;RS|AC) \; ,
\end{align*}
which is proved using the Shannon inequality $0\leq $
\begin{align*}   &\ \ + I(A;B|CT)+ I(A;B|DT)+ I(A;B|RT) \\
      &\ \ + I(A;R|ST)+ I(A;T|CD)+ I(ACDR;T|BS)  \\
      &\ \ + I(C;D|T)+ I(C;T|B)+ I(CRS;T|ABD) \\
      &\ \ + I(CDS;T|ABR)+ I(C;T|ARS)+ I(B;S|T)\\
      &\ \ + I(B;T|ACDRS)+ I(D;RS|ACT)+ I(D;T|A) \\
      &\ \ + I(D;T|B)+ I(DRS;T|ABC)+ I(R;T|A) \\
      &\ \ + I(R;T|B)+ I(S;T|A)+ I(S;T|R)  \; ,
\end{align*}
and a (C2) term $5I(T;CDRS|AB)$.

Similarly, the main proof of inequality \eqref{ineq3.2} will be abbreviated as follows.
Let RS be a copy of BD over AC.
The proof uses the Shannon inequality $0\leq $
\begin{align*}
        &\ \  3I(A;B) + 6I(A;C|BR)+ I(A;B|DR)+ 6I(B;R|C)\\
        &\ \ + I(D;R|A)+ I(C;D|R)+ I(AB;R|CD) \\
       &\ \ + I(D;R|B)+ I(C;R|ABD)+ I(B;RS|ACD) \\
       &\ \ + 6I(D;R|ABC)+ I(C;S|B)+ I(AR;B|CS) \\
       &\ \ + 8I(BD;S|ACR)+ I(D;RS|ABC)
\end{align*}
and a (C2) term $10I(RS;BD|AC)$.

{\bf Proof of (\ref{ineq3.3}):}
Let T be a BDS-copy of R over AC.  Let (L) be the non-negative expression
\begin{align*}
&\ \  -2I(A;C)+3I(A;C|R)+I(C;D)+2I(C;R|A)\\
&\ \  +I(C;S)+I(A;B|D)+I(A;R|C)+I(A;R|S) \\
&\ \            +I(A;CS|B) \; ,
\end{align*}
which is proved using the Shannon inequality $0\leq $
\begin{align*}
&\ \  I(B;T|A)  + I(CS;A|BT) + I(C;D|T)+ I(D;T|B) \\
&\ \  + I(ABRS;T|CD) + I(D;T|A)+ I(A;B|DT)\\
&\ \  + I(CRS;T|ABD) + I(R;T|A) + I(R;T|C)\\
&\ \  + I(A;C|RT) + I(BDS;T|ACR) + I(C;S|T)\\
&\ \  + I(B;T|CS)+ I(S;T|A) + I(DR;T|ABCS)\\
&\ \  + I(S;T|R)  + I(BCD;T|ARS) \; ,
\end{align*}
and a (C2) term $5I(T;BDRS|AC)$ .  Then let RS be a copy of AD over BC.  The proof of \eqref{ineq3.3} uses the
Shannon inequality $0\leq $
\begin{align*}
&\ \  6I(A;R|B) + 6I(B;C|AR)+ I(C;D|R) + I(D;R|A)\\
&\ \    + I(D;R|B)+ I(AB;R|CD)+ I(A;B|DR)\\
&\ \    + 8I(D;R|ABC) + I(C;R|ABD) + I(A;S|B)\\
&\ \   + I(B;R|AS)+ 8I(A;S|BCR)  + I(A;R|BCS)\\
&\ \  + I(A;C|BRS)   + 10I(D;S|ABCR)
\end{align*}
and a (C2) term $10I(RS;AD|BC)$.

{\bf Proof of (\ref{ineq3.4}):} Let T be a BCS-copy of R over AD.  Let (L) be the non-negative expression
\begin{align*}
&\ \  -I(A;D)+I(A;B|CR)+I(A;D|S)+I(A;D|R)\\
&\ \  +I(A;D|BR)+2I(A;R|D)+I(C;D|R)+I(D;R|A)\\
&\ \  +I(C;S|ABR)+I(R;S) \; ,
\end{align*}
which is proved using the Shannon inequality $0\leq $
\begin{align*}
&\ \  I(BR;T|D) + I(CR;T|A)+ I(S;T|A)+ I(S;T|D)\\
&\ \        + I(A;D|ST)+ I(BCR;T|ADS)+ I(R;T|D)\\
&\ \  + I(A;D|BRT)     + I(C;T|BR) + I(A;B|CRT)\\
&\ \      + I(C;D|RT)+ I(ABS;T|CDR)    + I(CS;T|ABDR)\\
&\ \  + I(R;S|T)  + I(AB;T|RS)+ I(C;S|ABRT)\\
&\ \ + I(D;T|ABCRS) \; ,
\end{align*}
and a (C2) term $4I(T;BCRS|AD)$.  Then let RS be a copy of CD over AB.  The proof of \eqref{ineq3.4} uses the
Shannon inequality $0\leq $
\begin{align*}
&\ \ I(D;S|B)+ I(A;B|DS)+ 6I(C;R|ABD)\\
&\ \ + I(D;RS|ABC)       + 5I(D;R|B)+ 4I(A;B|DR) \\
&\ \   + I(C;R|A)+ I(C;R|B)   + I(AB;R|CD) \\
&\ \  + I(D;R|ABS) + 6I(D;S|ABR)+ 7I(C;S|ABDR)
\end{align*}
and a (C2) term $8I(RS;CD|AB)$.

{\bf Proof of (\ref{ineq3.5}):}  Let T be a CDRS-copy of (DR) over AB.  Let (L) be the non-negative expression
\begin{align*}
&\ \ -I(A;B)+2I(B;DR|A)+2I(A;DR|B)\\
&\ \ +2I(A;B|DR)+I(C;D)+I(A;B|C)\\
&\ \ +I(A;B|DS)+I(R;S|D) \; ,
\end{align*}
which is proved using the Shannon inequality $0\leq $
\begin{align*}
&\ \ I(C;T|A)+ I(C;T|B)+ I(A;B|CT)+ I(C;D|T)\\
&\ \ + I(ABRS;T|CD) + I(DRS;T|ABC)+ I(DS;T|A)\\
&\ \ + I(DS;T|B)+ I(A;B|DST)+ I(CR;T|ABDS)\\
&\ \ + I(DR;T|A)+ I(DR;T|B)+ I(A;B|DRT)\\
&\ \ + I(CS;T|ABDR)+ I(R;S|DT)+ I(ABC;T|DRS) \; ,
\end{align*}
and a (C2) term $5I(T;CDRS|AB)$.  Then let RS be a copy of CD over AB.  The proof of \eqref{ineq3.5} uses the
Shannon inequality $0\leq $
\begin{align*}
&\ \  I(D;S|A)+ I(D;S|B)+ 8I(C;R|ABD)\\
&\ \ + I(D;RS|ABC)+ 4I(D;R|A)+ 4I(D;R|B)\\
&\ \ + I(C;R|A)+ I(C;R|B) + I(A;B|CR)+ I(C;D|R)\\
&\ \  + I(AB;R|CD)+ I(AB;D|RS)+ I(D;R|ABS)\\
&\ \   + 7I(D;S|ABR) + I(C;S|ABR)+ 9I(C;S|ABDR)
\end{align*}
and a (C2) term $10I(RS;CD|AB)$.

{\bf Abbreviated Proof of (\ref{ineq3.5}):}
>From here on the proofs will be further abbreviated.  We begin with an abbreviation of the above proof.\newline
T: CDRS-copy of (DR) over AB.\newline
L: -a.b. +2b.dr.a +2a.dr.b +2a.b.dr +c.d. +a.b.c +a.b.ds +r.s.d\newline
SL: c.t.a +c.t.b +a.b.ct +c.d.t +abrs.t.cd +drs.t.abc +ds.t.a +ds.t.b +a.b.dst +cr.t.abds +dr.t.a +dr.t.b +a.b.drt +cs.t.abdr +r.s.dt +abc.t.drs
C2L: 5t.cdrs.ab\newline
RS:
copy of CD over AB.\newline
S: d.s.a +d.s.b +8c.r.abd +d.rs.abc +4d.r.a +4d.r.b +c.r.a +c.r.b +a.b.cr +c.d.r +ab.r.cd +ab.d.rs +d.r.abs +7d.s.abr +c.s.abr +9c.s.abdr
C2: 10rs.cd.ab\newline

{\bf Abbreviated Proof of (\ref{ineq3.6}):}
T: BCS-copy of R over AD.\newline
L: -2a.d. +a.b.c +ar.d.b +a.d.s +3a.d.r +2a.r.d +c.d. +d.r.a +c.s.abr +r.s.\newline
SL: b.t.d +ar.d.bt +c.t.a +c.t.b +a.b.ct +c.d.t +abrs.t.cd +s.t.a +s.t.d +a.d.st +bcr.t.ads +r.t.a +r.t.d +bcs.t.adr +a.d.rt +r.t.abc +cs.t.abdr +r.s.t +ab.t.rs +d.t.abcrs +c.s.abrt
C2L: 5t.bcrs.ad\newline
RS:
copy of CD over AB.\newline
S: d.s.b +a.b.ds +8c.r.abd +d.rs.abc +6d.r.b +6a.b.dr +c.r.a +c.r.b +a.b.cr +c.d.r +ab.r.cd +d.r.abs +8d.s.abr +9c.s.abdr
C2: 10rs.cd.ab\newline

{\bf Abbreviated Proof of (\ref{ineq3.7}):}
T: CDRS-copy of (CR) over AB.\newline
L: -a.b. +3a.cr.b +3b.cr.a +2a.b.cr +a.b.dr +a.b.cs +a.b.ds +c.d.r +c.d.s +r.s.\newline
SL: cr.t.a +cr.t.b +a.b.crt +dr.t.a +dr.t.b +a.b.drt +c.d.rt +ds.t.abcr +cs.t.abdr +abs.t.cdr +cs.t.a +cs.t.b +a.b.cst +ds.t.a +ds.t.b +a.b.dst +c.d.st +cr.t.abds +dr.t.abcs +abr.t.cds +r.s.t +abcd.t.rs
C2L: 7t.cdrs.ab\newline
RS:
copy of CD over AB.\newline
S: 5c.r.a +5c.r.b +d.r.a +d.r.b +ab.r.cd +11d.r.abc +2c.r.abd +2c.s.a +2c.s.b +a.b.cs +d.s.a +d.s.b +ab.s.cd +r.s.c +ab.c.rs +10c.s.abr +3c.r.abs +d.r.abs +13d.s.abcr +d.r.abcs +c.s.abdr
C2: 15rs.cd.ab\newline

{\bf Abbreviated Proof of (\ref{ineq3.8}):}
T: BS-copy of C over ADR.\newline
L: -2a.c. +a.b.d +3ad.c.r +a.r.b +a.r.c +a.r.s +bd.s.ac +c.d. +c.r.ad +c.s.\newline
SL: b.t.a +d.t.b +a.b.dt +c.t.abd +b.t.r +a.r.bt +c.t.r +a.r.ct +cds.t.abr +bds.t.acr +s.t.a +c.s.t +a.t.cs +bd.s.act +s.t.r +a.r.st +bcd.t.ars +r.t.abcds
C2L: 4t.bcs.adr\newline
RS:
copy of BD over AC.\newline
S: b.r.c +a.c.br +4d.r.a +d.r.abc +4b.r.acd +5bd.s.acr +bd.r.acs
C2: 6rs.bd.ac\newline

{\bf Abbreviated Proof of (\ref{ineq3.9}):}
T: CDRS-copy of (DR) over AB.\newline
L: -a.b. +3b.dr.a +3a.dr.b +a.b.cr +2a.b.dr +a.b.cs +a.b.ds +c.d.r +c.d.s +r.s.\newline
SL: cr.t.a +cr.t.b +dr.t.a +dr.t.b +cs.t.a +cs.t.b +ds.t.a +ds.t.b +a.b.crt +a.b.drt +a.b.cst +a.b.dst +c.d.rt +c.d.st +r.s.t +abcd.t.rs +abr.t.cds +abs.t.cdr +cr.t.abds +dr.t.abcs +cs.t.abdr +ds.t.abcr
C2L: 7t.cdrs.ab\newline
RS:
copy of CD over AB.\newline
S: 2c.r.a +2c.r.b +a.b.cr +5d.r.a +5d.r.b +c.d.r +2ab.r.cd +2d.r.abc +8c.r.abd +c.s.a +c.s.b +d.s.a +d.s.b +ab.s.cd +12cd.s.abr +c.r.abs +d.r.abs +c.rs.abd +cd.r.abs +d.rs.abc
C2: 15rs.cd.ab\newline

{\bf Abbreviated Proof of (\ref{ineq3.10}):}
T: DS-copy of B over ACR.\newline
L: -a.b. +a.br.d +3a.b.cr +a.b.cs +a.cr.b +b.c.rs +b.r.ac +c.d. +c.s.a +d.s.abc +r.s.c\newline
SL: a.br.dt +a.bc.st +a.r.bct +ab.t.cd +ad.t.bcrs +2b.t.cr +b.t.rs +c.d.t +c.s.brt +cs.t.abdr +d.s.abct +d.t.a +d.t.br +ds.t.abcr +r.s.t +r.t.abcds +s.t.a +s.t.bc
C2L: 4t.bds.acr\newline
RS:
copy of CD over AB.\newline
S: a.b.dr +a.b.cs +ab.c.rs +a.c.brs +c.d.r +2cd.r.abs +7d.rs.abc +cd.s.abr +2c.rs.abd +5c.s.abr +c.s.a +2c.s.b +ab.r.cd +2d.s.abr +d.r.a +2d.r.b +4c.r.a +4c.r.b +d.s.abc +r.s.c
C2: 12rs.cd.ab\newline

{\bf Abbreviated Proof of (\ref{ineq3.11}):}
T: DS-copy of A over BCR.\newline
L: -a.b. +2a.b.c +c.d. +2a.bc.r +c.r.a +a.r.bc +d.r.bc +a.b.ds +r.s.d\newline
SL: a.t.c +ds.t.a +ds.t.b +c.d.t +b.t.cd +a.t.r +c.r.at +a.b.crt +d.t.r +d.r.bct +ds.t.abcr +as.t.bcdr +a.b.dst +cr.t.abds +r.s.dt +abc.t.drs
C2L: 4t.ads.bcr\newline
RS:
copy of CD over AB.\newline
S: 5c.r.b +a.b.cr +2d.r.a +2d.r.b +2a.b.dr +c.d.r +ab.r.cd +4d.r.abc +4c.r.abd +a.r.bcd +d.s.a +d.s.b +8cd.s.abr +ab.d.rs +d.r.abs +2c.s.abdr
C2: 10rs.cd.ab\newline

{\bf Abbreviated Proof of (\ref{ineq3.12}):}
T: CS-copy of B over ADR.\newline
L: -a.b. +a.br.c +cd.r.b +a.d.b +c.d. +b.r.ad +d.r.b +3a.b.dr +d.s.a +a.b.ds +r.s.d +b.d.rs\newline
SL: c.t.a +a.br.ct +c.d.t +b.t.cd +c.t.br +2d.r.t +2b.t.dr +c.r.bdt +ds.t.abcr +as.t.bcdr +s.t.a +a.bd.st +s.t.bd +cr.t.abds +r.s.t +b.t.rs +d.s.brt +ac.t.bdrs
C2L: 4t.bcs.adr\newline
RS:
copy of CD over AB.\newline
S: 2c.r.a +3c.r.b +2a.b.cr +4d.r.a +2d.r.b +2c.d.r +2ab.r.cd +3d.r.abc +4c.r.abd +a.r.bcd +d.s.b +10cd.s.abr +d.r.abs +a.d.brs +2c.rs.abd
C2: 12rs.cd.ab\newline

{\bf Abbreviated Proof of (\ref{ineq3.13}):}
T: DS-copy of C over ABR.\newline
L: -b.c. +a.b.c +a.b.d +cr.d. +a.r.c +c.r.ab +abd.r.c +3b.c.ar +a.s.b +c.s.a +a.s.r +b.r.as +a.c.brs\newline
SL: d.t.a +d.t.b +a.b.dt +cr.d.t +c.t.abd +2a.r.t +2c.t.ar +bd.r.act +a.bs.crt +abs.t.cdr +ac.s.t +s.t.b +bdr.t.acs +s.t.r +b.r.st +c.t.brs +s.t.abcdr +d.t.abcrs
C2L: 4t.cds.abr\newline
RS:
copy of AD over BC.\newline
S: 4a.r.b +2a.r.c +d.r.a +d.r.b +a.b.dr +c.d.r +2ab.r.cd +8d.rs.abc +c.r.abd +a.s.b +c.s.a +2a.br.cs +a.s.r +b.r.as +4a.s.bcr +a.c.brs +4ad.s.bcr
C2: 12rs.ad.bc\newline

{\bf Abbreviated Proof of (\ref{ineq3.14}):}
T: DS-copy of B over ACR.\newline
L: -b.c. +a.c.b +a.b.d +c.d. +a.r.b +b.r.ac +3b.c.ar +ad.r.b +c.r.bd +c.s.a +a.b.cs +a.s.r +b.r.as\newline
SL: d.t.a +a.b.dt +b.t.cd +c.d.t +2a.r.t +2b.t.ar +d.t.br +d.r.abt +c.r.bdt +cs.t.abdr +as.t.bcdr +s.t.ab +c.s.t +b.t.cs +a.c.bst +dr.t.abcs +s.t.r +b.r.st +a.s.brt +cd.t.abrs
C2L: 4t.bds.acr\newline
RS:
copy of AD over BC.\newline
S: 2a.r.b +4a.r.c +d.r.a +2d.r.b +a.b.dr +c.d.r +ab.r.cd +8d.rs.abc +2c.r.abd +a.r.bcd +2a.s.b +c.s.a +a.br.cs +a.s.r +b.r.as +4a.s.bcr +a.r.bcs +2a.c.brs +4ad.s.bcr
C2: 12rs.ad.bc\newline

{\bf Abbreviated Proof of (\ref{ineq3.15}):}
T: DR-copy of A over BCS.\newline
L: -a.b. +2a.bs.c +b.c.a +a.b.d +c.d. +c.r.b +c.s.a +a.s.bc +a.b.cs +r.s.c +a.c.rs\newline
SL: a.t.c +d.t.a +d.t.b +a.b.dt +c.d.t +abrs.t.cd +crs.t.abd +r.t.b +r.t.c +c.rs.at +b.c.rt +ads.t.bcr +r.s.t +a.t.rs +bd.t.acrs
C2L: 4t.adr.bcs\newline
RS:
copy of CD over AB.\newline
S: c.r.a +a.b.cr +6d.rs.abc +4c.s.b +c.s.abr +4c.r.abs +b.c.ars
C2: 6rs.cd.ab\newline

{\bf Abbreviated Proof of (\ref{ineq3.16}):}
T: CS-copy of B over ADR.\newline
L: -a.b. +a.b.c +a.d.b +c.d. +r.ac.b +d.r.b +3a.b.dr +b.r.ad +d.r.bc +d.s.a +a.b.ds +r.s.d +b.d.rs\newline
SL: c.t.a +c.d.t +b.t.cd +c.t.br +a.br.ct +2d.r.t +2b.t.dr +c.r.bdt +s.t.a +s.t.bd +a.bd.st +r.s.t +b.t.rs +d.s.brt +t.ac.bdrs +t.cr.abds +t.ds.abcr +t.as.bcdr
C2L: 4t.bcs.adr\newline
RS:
copy of CD over AB.\newline
S: c.r.b +2a.d.r +2d.r.a +2d.r.b +a.r.bcd +4c.r.abd +d.r.abc +d.s.b +a.d.brs +d.r.abs +2c.rs.abd +6s.cd.abr
C2: 8rs.cd.ab\newline
\end{proof}


\section{Inequalities Using Three Copy Variables and Three Copy Steps}\label{ThreeCopiesB}

In this section we will continue the list of inequalities that involve three copy variables.
The inequalities in this section will also involve three copy steps, and hence two lemmas.
We will give the first one in full detail and then, as before, we will quickly turn to
abbreviated versions of the proofs.

\begin{theorem}
The following is an information inequality
\begin{eqnarray*}
aI(A;B) &\leq& bI(A;B|C) + cI(A;C|B) + dI(B;C|A) \\
      &+&  eI(A;B|D) + fI(A;D|B)+ gI(B;D|A)\\
      &+& hI(C;D) + iI(C;D|A)
\end{eqnarray*}
for each of the following values of $(a,b,c,d,e,f,g,h,i)$.
\begin{align}
&&(3&,&4&,&9&,&3&,&6&,&3&,&0&,&3&,&0)&&&& \label{ineq3.17}\\
&&(3&,&7&,&4&,&1&,&4&,&1&,&0&,&3&,&0)&&&& \label{ineq3.18}\\
&&(3&,&4&,&6&,&4&,&4&,&1&,&0&,&3&,&0)&&&& \label{ineq3.19}\\
&&(4&,&5&,&17&,&6&,&6&,&7&,&0&,&4&,&0)&&&& \label{ineq3.20}\\
&&(4&,&5&,&17&,&13&,&6&,&2&,&0&,&4&,&0)&&&& \label{ineq3.21}\\
&&(3&,&4&,&7&,&5&,&3&,&1&,&0&,&3&,&0)&&&& \label{ineq3.22}\\
&&(6&,&8&,&9&,&9&,&6&,&10&,&1&,&6&,&0)&&&& \label{ineq3.23}\\
&&(6&,&13&,&20&,&2&,&9&,&3&,&0&,&6&,&0)&&&& \label{ineq3.24}\\
&&(4&,&10&,&15&,&1&,&4&,&2&,&2&,&4&,&0)&&&& \label{ineq3.25}\\
&&(4&,&6&,&11&,&3&,&6&,&2&,&0&,&4&,&0)&&&& \label{ineq3.26}\\
&&(3&,&6&,&6&,&1&,&5&,&4&,&0&,&3&,&0)&&&& \label{ineq3.27}\\
&&(3&,&6&,&8&,&1&,&3&,&2&,&2&,&3&,&0)&&&& \label{ineq3.28}\\
&&(4&,&5&,&6&,&6&,&4&,&2&,&2&,&4&,&0)&&&& \label{ineq3.29}\\
&&(3&,&8&,&6&,&1&,&3&,&1&,&0&,&3&,&0)&&&& \label{ineq3.30}\\
&&(4&,&14&,&10&,&1&,&6&,&2&,&0&,&4&,&0)&&&& \label{ineq3.31}\\
&&(3&,&4&,&4&,&3&,&3&,&4&,&2&,&3&,&0)&&&& \label{ineq3.32}\\
&&(4&,&13&,&9&,&1&,&7&,&3&,&0&,&4&,&0)&&&& \label{ineq3.33}\\
&&(6&,&8&,&16&,&7&,&6&,&3&,&3&,&6&,&0)&&&& \label{ineq3.34}
\end{align}
\end{theorem}
\begin{proof}
{\bf Proof of (\ref{ineq3.17}):}
We first show the following 6-variable information inequality:
\begin{eqnarray}\label{lem4.1}
   I(A;B) &\le& 2I(A;BRS|D)+I(A;S) \nonumber\\
      && +I(A;BR|CS)+I(C;D|S)+I(B;C|RS)\nonumber\\
      && +I(D;R|AS) +I(A;D|BRS).
\end{eqnarray}
To see this, verify the following 7-variable identity by expanding into entropies
and canceling:
\begin{align*}
&I(A;B)&\\
      &\ \  + I(ARS;T|B)                                                        & \text{(S)} \\
      &\ \  + I(AS;T|D)                                                         & \text{(S)} \\
      &\ \  + I(C;T|AS)                                                          & \text{(S)} \\
      &\ \  + I(S;T|D)                                                        & \text{(S)} \\
      &\ \  + I(C;D|ST)                                                      & \text{(S)} \\
      &\ \  + I(ABR;T|CDS)                                                        & \text{(S)} \\
      &\ \  + I(C;T|RS)                                                        & \text{(S)} \\
      &\ \  + I(A;R|CST)                                                       & \text{(S)} \\
      &\ \  + I(B;C|ARST)                                                         & \text{(S)} \\
      &\ \  + I(D;R|AST)                                                       & \text{(S)} \\
      &\ \  + I(BC;T|ADRS)                                                       & \text{(S)} \\
      &\ \  + I(D;T|ABCRS)                                                         & \text{(S)} \\
      & = 2I(A;BRS|D)+I(A;S)+I(A;BR|CS)&\\
      & +I(C;D|S)+I(B;C|RS)+I(D;R|AS)&\\
      & +I(A;D|BRS) &\\
      &\ \   + 3I(T;AC|BDRS)                                                     & \text{(C2)} \\
      &\ \   + (H(TB)-H(AB))                                                      & \text{(C1)} \\
      &\ \   + 2(H(TD)-H(AD))                                                    & \text{(C1)} \\
      &\ \   - (H(TS)-H(AS))                                                   & \text{(C1)} \\
      &\ \   + (H(TRS)-H(ARS))                                                    & \text{(C1)} \\
      &\ \   - 3(H(TBDRS)-H(ABDRS))                                                     & \text{(C1)} \\
\end{align*}

Each of the terms in the lines marked (S) is a conditional mutual
information and is therefore nonnegative.  Thus, if the terms marked (S)
are erased and the "=" is replaced by "$\leq$", then we obtain a 7-variable
Shannon-type inequality.  By the Copy Lemma, we may choose T such
that T is a C-copy of A over BDRS.  Then the term marked (C2) is zero
by condition (C2), and each of the terms marked (C1) equals zero
by condition (C1).

With \eqref{lem4.1} verified, we next show the following 5-variable information inequality:
\begin{eqnarray}\label{lem4.2}
   2I(A;B) &\le& 2I(A;BR|C)+5I(A;BR|D) \nonumber\\
      && +I(A;DR|B)+2I(C;D)+2I(B;C|R)\nonumber\\
      && +I(D;R|A)+2I(A;D|BR).
\end{eqnarray}
To see this, verify the following 6-variable identity by expanding into entropies
and canceling:
\begin{align*}
&2I(A;B)&\\
      &\ \  + 2I(AR;S|B)                                                       & \text{(S)} \\
      &\ \  + 2I(C;S|A)                                                         & \text{(S)} \\
      &\ \  + I(A;S|D)                                                          & \text{(S)} \\
      &\ \  + I(C;D|S)                                                       & \text{(S)} \\
      &\ \  + 2I(ABR;S|CD)                                                    & \text{(S)} \\
      &\ \  + 2I(C;S|R)                                                        & \text{(S)} \\
      &\ \  + I(A;R|CS)                                                       & \text{(S)} \\
      &\ \  + I(B;C|ARS)                                                       & \text{(S)} \\
      &\ \  + I(BC;S|ADR)                                                         & \text{(S)} \\
      &\ \  + 3I(C;S|ABDR)                                                       & \text{(S)} \\
      &\ \  2I(D;S|ABCR)                                                      & \text{(S)} \\
      &\ \  [-I(A;B)+2I(A;BRS|D)+I(A;S) \\
      &\ \  +I(A;BR|CS)+I(C;D|S)+I(B;C|RS) \\
      &\ \  +I(D;R|AS)+I(A;D|BRS)]                                              & \text{(L1)} \\
      & = 2I(A;BR|C)+5I(A;BR|D)+I(A;DR|B)&\\
      & +2I(C;D)+2I(B;C|R)+I(D;R|A) &\\
      & +2I(A;D|BR) + 8I(S;AC|BDR) &\\
      &\ \   + 8I(S;AC|BDR)                                                     & \text{(C2)} \\
      &\ \   - (H(S)-H(A))                                                      & \text{(C1)} \\
      &\ \   + 2(H(SB)-H(AB))                                                     & \text{(C1)} \\
      &\ \   + 3(H(SD)-H(AD))                                                    & \text{(C1)} \\
      &\ \   + (H(SR)-H(AR))                                                     & \text{(C1)} \\
      &\ \   - 5(H(SBDR)-H(ABDR))                                                      & \text{(C1)} \\
\end{align*}

The expression in brackets marked with (L1) is nonnegative by \eqref{lem4.1}.
Each of the terms in the lines marked (S) is a conditional mutual
information and is therefore nonnegative.  Thus, if the terms marked (S)
are erased and the "=" is replaced by "$\leq$", then we obtain a 6-variable
Shannon-type inequality.  By the Copy Lemma, we may choose S such
that S is a C-copy of A over BDR.  Then the term marked (C2) is zero
by condition (C2), and each of the terms marked (C1) equals zero
by condition (C1).

With \eqref{lem4.2} verified, we now continue with the proof of \eqref{ineq3.17}.
By expanding mutual informations into entropies and cancelling
terms, one can verify the following 5-variable identity:
\begin{align*}
&3I(A;B)&\\
      &\ \ + 3I(A;R|B)                                                         & \text{(S)} \\
      &\ \ + I(A;R|C)                                                          & \text{(S)} \\
      &\ \ + I(B;C|AR)                                                         & \text{(S)} \\
      &\ \ + 9I(D;R|B)                                                        & \text{(S)} \\
      &\ \ + I(A;B|DR)                                                       & \text{(S)} \\
      &\ \ + I(C;D|R)                                                        & \text{(S)} \\
      &\ \ + I(AB;R|CD)                                                        & \text{(S)} \\
      &\ \ + 9I(C;R|ABD)                                                        & \text{(S)} \\
      &\ \ + 3I(D;R|ABC)                                                        & \text{(S)} \\
      &\ \ + [-2I(A;B)+2I(A;BR|C)+5I(A;BR|D)\\
      &\ \ +I(A;DR|B)+2I(C;D)+2I(B;C|R)\\
      &\ \ +I(D;R|A)+2I(A;D|BR)]                                                & \text{(L2)} \\
      & = 4I(A;B|C)+9I(A;C|B)+3I(B;C|A)\\
      &\ \ +6I(A;B|D)+3I(A;D|B)+3I(C;D)\\
      &\ \ + 13I(R;AD|BC)                                                       & \text{(C2)} \\
      &\ \ - 3(H(R)-H(A))                                                       & \text{(C1)} \\
      &\ \ + 12(H(RB)-H(AB))                                                   & \text{(C1)} \\
      &\ \ + 4(H(RC)-H(AC))                                                    & \text{(C1)} \\
      &\ \ - 13(H(RBC)-H(ABC))                                                 & \text{(C1)} \\
\end{align*}
The expression in brackets ``[]" marked (L) is nonnegative by Lemma 11.1.
Each of the terms in the lines marked (S) is a conditional mutual
information and is therefore nonnegative.  Thus, if the terms marked (S)
are erased and the "=" is replaced by "$\le$", then we obtain a 5-variable
Shannon-type inequality.  By the Copy Lemma, we may choose R such
that R is a D-copy of C over AB.  Then the term marked (C2) is zero
by condition (C2), and each of the terms marked (C1) equals zero
by condition (C1).    This finishes the proof of \eqref{ineq3.1}.

The expression in brackets marked with (L2) is nonnegative by \eqref{lem4.2}.
Each of the terms in the lines marked (S) is a conditional mutual
information and is therefore nonnegative.  Thus, if the terms marked (S)
are erased and the "=" is replaced by "$\leq$", then we obtain a 5-variable
Shannon-type inequality.  By the Copy Lemma, we may choose R such
that R is a D-copy of A over BC.  Then the term marked (C2) is zero
by condition (C2), and each of the terms marked (C1) equals zero
by condition (C1).

{\bf Abbreviated Proof of (\ref{ineq3.17}):}

We immediately give the abbreviated proof of the same inequality, using the
same abbreviations as in the last section, the main difference being that this time
there are two lemmas used in the proof.

T: C-copy of A over BDRS.\newline
L1: -a.b. +2a.brs.d +a.s. +a.br.cs +c.d.s +b.c.rs +d.r.as +a.d.brs\newline
SL1: ars.t.b +as.t.d +c.t.as +s.t.d +c.d.st +abr.t.cds +c.t.rs +a.r.cst +b.c.arst +d.r.ast +bc.t.adrs +d.t.abcrs
C2L1: 3t.ac.bdrs\newline
S: C-copy of A over BDR.\newline
L2: -2a.b. +2a.br.c +5a.br.d +a.dr.b +2c.d. +2b.c.r +d.r.a +2a.d.br\newline
SL2: 2ar.s.b +2c.s.a +a.s.d +c.d.s +2abr.s.cd +2c.s.r +a.r.cs +b.c.ars +bc.s.adr +3c.s.abdr +2d.s.abcr\newline
C2L2: 8s.ac.bdr\newline
R: D-copy of A over BC.\newline
S: 3a.r.b +a.r.c +b.c.ar +9d.r.b +a.b.dr +c.d.r +ab.r.cd +9c.r.abd +3d.r.abc\newline
C2: 13r.ad.bc\newline

{\bf Abbreviated Proof of (\ref{ineq3.18}):}
T: BCS-copy of R over AD.\newline
L1: -a.b. +a.b.c +a.b.d +c.d. +a.d.r +a.r.d +d.r.a\newline
SL1: b.t.a +c.t.a +c.t.b +a.b.ct +b.t.d +a.d.bt +c.d.t +abrs.t.cd +crs.t.abd +drs.t.abc
C2L1: 3t.bcrs.ad\newline
S: BD-copy of R over AC.\newline
L2: -2a.b. +2a.b.c +2a.b.d +2c.d. +a.c.r +a.r.c +c.r.a +a.d.r +a.r.d +d.r.a\newline
SL2: b.s.a +b.s.c +a.c.bs +d.s.a +d.s.b +a.b.ds +c.d.s +abr.s.cd +cr.s.abd +dr.s.abc\newline
C2L2: 3s.bdr.ac\newline
R: D-copy of C over AB.\newline
S: 3c.r.b +3a.b.cr +3d.r.b +3a.b.dr +c.d.r +ab.r.cd +3c.r.abd +3d.r.abc\newline
C2: 7r.cd.ab\newline

{\bf Abbreviated Proof of (\ref{ineq3.19}):}
T: CS-copy of A over BDR.\newline
L1: -a.b. +a.br.c +2a.br.d +c.d. +b.c.r +d.r.a +a.d.br\newline
SL1: ar.t.b +c.t.a +a.t.d +c.d.t +abrs.t.cd +c.t.r +a.r.ct +b.c.art +d.r.at +bcs.t.adr +ds.t.abcr
C2L1: 3t.acs.bdr\newline
S: BD-copy of R over AC.\newline
L2: -2a.b. +a.b.c +a.b.d +2c.d. +a.c.r +a.r.c +c.r.a +b.c.r +a.br.c +d.r.a +a.d.br +2a.br.d\newline
SL2: b.s.a +b.s.c +a.c.bs +d.s.a +d.s.b +a.b.ds +c.d.s +abr.s.cd +cr.s.abd +dr.s.abc\newline
C2L2: 3s.bdr.ac\newline
R: D-copy of A over BC.\newline
S: 4a.r.b +3b.c.ar +4d.r.b +a.b.dr +c.d.r +ab.r.cd +4c.r.abd +4d.r.abc\newline
C2: 9r.ad.bc\newline

{\bf Abbreviated Proof of (\ref{ineq3.20}):}
T: D-copy of C over ABRS.\newline
L1: -2a.c. +2brs.c.a +2a.b.d +c.d. +2abs.c.r +a.r.c +a.c.br +2c.r.bd +c.ds. +a.s.bcd +a.r.cs +as.c.br +d.s.abr +d.r.abcs +2d.s.abcr\newline
SL1: c.t.a +2d.t.a +2a.b.dt +c.d.t +b.t.cd +c.t.abd +c.t.r +a.r.ct +b.t.acr +2d.t.br +r.t.abd +2c.r.bdt +2a.t.bcdr +cs.t.a +s.t.d +c.d.st +b.t.cds +a.s.bcdt +cs.t.r +a.r.cst +b.c.arst +cd.r.abst +c.s.abrt +d.s.abrt +c.t.abdrs +2cd.s.abrt +2d.t.abcrs
C2L1: 6t.cd.abrs\newline
S: C-copy of A over BDR.\newline
L2: -3a.b. +3a.br.c +2a.cr.b +2br.c.a +5a.br.d +2a.d.b +3c.d. +a.r.c +3b.c.r +d.r.a +5a.d.br +2c.r.bd +d.r.abc\newline
SL2: 7ar.s.b +c.s.a +a.s.d +c.d.s +abr.s.cd +b.s.cd +dr.s.abc +c.s.r +b.c.ars +d.r.as +bc.s.adr +c.s.abdr +3d.s.abcr\newline
C2L2: 10s.ac.bdr\newline
R: D-copy of A over BC.\newline
S: 8a.r.b +b.c.ar +13d.r.b +a.b.dr +c.d.r +ab.r.cd +2a.r.bcd +11c.r.abd +7d.r.abc\newline
C2: 22r.ad.bc\newline

{\bf Abbreviated Proof of (\ref{ineq3.21}):}
T: C-copy of A over BDRS.\newline
L1: -2a.b. +a.bs.c +2a.br.d +2c.d. +b.c.r +a.br.c +d.r.a +a.ds.br +b.c.s +d.s.a +a.dr.bs +2a.brs.d +2c.s.bdr +c.r.abds +c.s.abdr\newline
SL1: as.t.b +2c.t.a +2a.t.d +2c.d.t +2br.t.cd +ar.t.b +c.t.r +a.r.ct +b.c.art +d.r.at +b.t.adr +d.t.abcr +c.t.s +a.s.ct +b.c.ast +d.s.at +b.t.ads +d.t.abcs +a.r.bdst +a.s.bdrt +2c.s.bdrt +2a.t.bcdrs +ac.r.bdst +ac.s.bdrt +2c.t.abdrs
C2L1: 6t.ac.bdrs\newline
S: BD-copy of A over CR.\newline
L2: -3a.c. +2a.cr.b +4a.br.d +3c.d. +8a.c.r +a.r.c +2c.r.a +2b.c.r +a.r.d +d.r.a +2a.d.br +d.r.abc\newline
SL2: a.s.c +c.d.s +abr.s.cd +d.s.abc +a.s.r +6bd.s.r +ab.s.r +c.r.as +bd.s.acr +d.s.r +a.r.ds +bc.s.adr +d.r.abs +2a.s.bcdr +3c.s.abdr\newline
C2L2: 10s.abd.cr\newline
R: D-copy of A over BC.\newline
S: 13a.r.b +11b.c.ar +8d.r.b +2a.b.dr +c.d.r +ab.r.cd +8c.r.abd +12d.r.abc\newline
C2: 22r.ad.bc\newline

{\bf Abbreviated Proof of (\ref{ineq3.22}):}
T: C-copy of A over BDRS.\newline
L1: -a.c. +a.brs.d +c.d. +cd.r.a +a.bs.dr +ab.c.s +d.s.a +a.dr.bs +c.r.abds +c.s.abdr\newline
SL1: c.t.a +c.d.t +a.t.cd +d.r.t +a.t.dr +c.r.adt +b.t.acdr +a.t.bs +c.t.s +a.s.ct +b.c.ast +d.s.at +b.t.ads +d.t.abcs +a.r.bdst +ac.r.bdst +ac.s.bdrt +c.t.abdrs
C2L1: 3t.ac.bdrs\newline
S: D-copy of A over BCR.\newline
L2: -2a.c. +a.cr.b +2br.c.a +a.b.d +2c.d. +2cd.r.a +a.b.cr +5a.c.br +a.b.dr +a.d.br +a.br.d +d.r.ab +d.r.abc\newline
SL2: d.s.b +a.b.ds +c.d.s +a.s.cd +d.s.abc +3a.s.br +c.r.s +a.s.cr +c.br.as +3d.s.br +2d.r.abs +d.r.acs +b.s.acdr +3c.s.abdr +d.s.abcr\newline
C2L2: 7s.ad.bcr\newline
R: D-copy of C over AB.\newline
S: c.r.a +7c.r.b +d.r.a +3d.r.b +c.d.r +ab.r.cd +2b.r.acd +4c.r.abd +8d.r.abc\newline
C2: 16r.cd.ab\newline

{\bf Abbreviated Proof of (\ref{ineq3.23}):}
T: D-copy of A over BCRS.\newline
L1: -3a.c. +a.crs.b +bs.c.a +a.b.d +c.d. +2a.bs.cr +5a.c.br +2c.br.a +2a.b.dr +d.r.ab +2c.dr. +a.br.cs +c.s.ad +3ad.s.bcr +2d.s.bcr +d.r.abcs\newline
SL1: d.t.b +a.b.dt +c.d.t +a.t.cd +3a.t.br +2c.r.t +2a.t.cr +2b.c.art +2d.t.ar +2d.t.br +2a.b.drt +d.r.abt +2c.dr.t +2b.t.cdr +3c.t.abdr +s.t.ab +c.s.t +a.t.cs +b.c.ast +d.t.as +c.s.adt +b.t.acds +a.r.bcst +2a.s.bcrt +2d.s.bcrt +2a.t.bcdrs +r.ad.bcst +3s.ad.bcrt +3d.t.abcrs
C2L1: 9t.ad.bcrs\newline
S: C-copy of A over BDR.\newline
L2: -4a.c. +a.cr.b +3br.c.a +a.b.d +2c.d. +2a.b.cr +6a.c.br +4a.b.dr +10a.d.br +a.br.d +b.d.ar +d.r.ab +2c.dr. +c.r.ad +d.r.ac +d.r.abc\newline
SL2: c.d.s +a.s.cd +d.s.abc +2a.s.br +c.s.ar +9c.s.br +c.r.abs +2a.s.dr +b.d.ars +c.r.ads +d.r.acs +2a.s.bcdr +2b.s.acdr +c.s.abdr +4d.s.abcr\newline
C2L2: 13s.ac.bdr\newline
R: D-copy of C over AB.\newline
S: 3c.r.a +9c.r.b +6d.r.a +15d.r.b +2c.d.r +2b.r.acd +4r.ab.cd +17c.r.abd +11d.r.abc\newline
C2: 35r.cd.ab\newline

{\bf Abbreviated Proof of (\ref{ineq3.24}):}
T: C-copy of A over BDRS.\newline
L1: -3a.d. +2a.b.c +3a.brs.d +3c.d. +a.brs.c +d.r.a +c.ds.r +2a.bdr.s +2d.s.a +3a.dr.bs +2d.s.bc +a.bd.rs +a.b.drs\newline
SL1: 3c.t.a +2a.b.ct +3a.t.d +3c.d.t +2b.t.cd +c.t.r +a.r.ct +d.r.at +2a.t.s +2c.t.bs +2d.s.at +2d.s.bct +a.t.rs +c.s.art +t.cr.bs +a.b.crst +r.t.ads +a.b.drst +c.d.rst +t.ab.cdrs +c.t.abdrs +d.t.abcrs +t.bcr.ads +2t.ar.bcds +t.abrs.cd +2t.drs.abc +t.bcs.adr
C2L1: 10t.ac.bdrs\newline
S: D-copy of C over ABR.\newline
L2: -4a.c. +4a.br.c +a.cr.b +a.b.d +3c.d. +a.r.c +11a.c.br +2b.c.ar +2d.r.a +a.b.dr +3a.d.br +6a.br.d +d.r.ab +c.d.r +d.cr. +2d.r.bc\newline
SL2: d.s.b +a.b.ds +2a.r.s +c.s.ar +4c.s.br +a.b.crs +r.ad.s +10d.s.br +d.r.abs +d.cr.s +2d.r.bcs +2a.s.bcdr +2s.ab.cdr +11c.s.abdr +2d.s.abcr\newline
C2L2: 17s.cd.abr\newline
R: D-copy of A over BC.\newline
S: 17a.r.b +a.r.c +11d.r.b +a.b.dr +c.d.r +2a.r.bcd +3r.ab.cd +12c.r.abd +18d.r.abc\newline
C2: 35r.ad.bc\newline

{\bf Abbreviated Proof of (\ref{ineq3.25}):}
T: C-copy of A over BDRS.\newline
L1: -2a.b. +2a.b.c +2b.d.a +2c.d. +d.ar.b +d.r.ac +a.b.drs +2a.brs.d +2a.rs.bd +a.bs.dr +d.ars.b +d.rs.ac\newline
SL1: 2c.t.b +2a.b.ct +2c.d.t +2a.t.cd +t.ar.b +c.t.ar +d.r.t +a.t.dr +b.d.art +d.r.act +t.ars.b +c.t.ars +d.rs.t +a.t.drs +b.d.arst +d.rs.act +b.t.acdrs +c.t.abdrs +2t.drs.abc +t.cs.abdr +t.bs.acdr
C2L1: 6t.ac.bdrs\newline
S: D-copy of A over BCR.\newline
L2: -3a.b. +2a.c.b +2c.d. +5a.cr.b +5a.br.c +b.c.ar +a.b.dr +2a.dr.b +3a.br.d +b.d.ar +d.br.a +d.cr. +2d.r.ac\newline
SL2: 2s.ar.b +r.s.c +2s.ar.c +b.c.ars +d.s.ar +6s.dr.b +c.d.rs +b.s.acdr +s.ab.cdr +7c.s.abdr +d.s.abcr\newline
C2L2: 10s.ad.bcr\newline
R: D-copy of A over BC.\newline
S: 3a.r.b +3a.r.c +2d.r.a +6d.r.b +c.d.r +2b.r.acd +2r.ab.cd +8c.r.abd +11d.r.abc\newline
C2: 23r.ad.bc\newline

{\bf Abbreviated Proof of (\ref{ineq3.26}):}
T: C-copy of A over BDRS.\newline
L1: -2a.c. +2c.d. +a.c.br +2d.r.a +a.c.bs +r.s.abc +s.dr.a +a.b.drs +2a.bd.rs +4a.brs.d +c.ds.r\newline
SL1: 2c.t.a +a.b.ct +2a.t.d +2c.d.t +c.t.r +a.r.ct +c.t.br +d.r.at +c.t.bs +r.s.t +2a.t.rs +c.s.art +r.s.bct +a.b.crst +d.rs.at +a.b.drst +c.d.rst +t.ab.cdrs +c.t.abdrs +d.t.abcrs +2t.abrs.cd +t.drs.abc +t.bcs.adr
C2L1: 7t.ac.bdrs\newline
S: D-copy of C over ABR.\newline
L2: -2a.b. +2a.b.c +2c.d. +c.r.a +c.r.b +2a.b.cr +2c.r.ab +c.br.a +5c.ar.b +2d.r.a +a.b.dr +2a.d.br +4a.br.d +c.d.r\newline
SL2: c.s.b +r.s.a +7s.dr.b +s.ab.cdr +7c.s.abdr +d.s.abcr\newline
C2L2: 9s.cd.abr\newline
R: D-copy of C over AB.\newline
S: c.r.a +c.r.b +8d.r.b +a.b.dr +c.d.r +2r.ab.cd +8c.r.abd +10d.r.abc\newline
C2: 20r.cd.ab\newline

{\bf Abbreviated Proof of (\ref{ineq3.27}):}
T: D-copy of A over BCRS.\newline
L1: -a.b. +c.d. +a.b.cs +b.as.d +c.rs.a +a.brs.c +a.bcs.r +d.r.bcs +d.r.abcs\newline
SL1: d.t.b +c.d.t +c.s.t +a.t.cs +d.t.as +b.as.dt +t.bs.cd +c.t.abds +t.as.r +c.r.ast +a.b.crst +d.r.bcst +a.t.bcdrs +r.ad.bcst +d.t.abcrs
C2L1: 3t.ad.bcrs\newline
S: C-copy of A over BDR.\newline
L2: -2a.c. +a.c.b +a.b.d +2c.d. +c.r.a +a.b.cr +2a.bc.r +d.r.a +a.b.dr +4a.d.br +3a.br.d +c.d.r +d.r.bc +d.r.abc\newline
SL2: c.s.a +2a.s.d +c.d.s +b.s.cd +c.s.abd +2a.s.r +4c.s.br +a.b.crs +d.r.as +a.b.drs +c.d.rs +a.s.bcdr +s.ab.cdr +c.s.abdr +2d.s.abcr +2s.dr.abc\newline
C2L2: 8s.ac.bdr\newline
R: D-copy of C over AB.\newline
S: 4c.r.b +8d.r.b +a.r.bcd +r.ab.cd +8c.r.abd +2d.r.abc\newline
C2: 13r.cd.ab\newline

{\bf Abbreviated Proof of (\ref{ineq3.28}):}
T: D-copy of B over ACRS.\newline
L1: -a.brs. +a.rs. +c.s.br +2a.b.crs +b.c.ars +c.as.br +a.b.drs +c.d.rs\newline
SL1: b.t.ars +c.s.rt +b.t.crs +a.c.brst +d.t.ars +d.t.brs +a.b.drst +c.ds.rt +t.ab.cdrs +c.t.abdrs +d.t.abcrs
C2L1: 3t.bd.acrs\newline
S: C-copy of D over ABR.\newline
L2: -2a.d. +2a.b.cr +a.c.br +b.c.ar +a.b.dr +4a.d.br +b.d.ar +d.br.a +2d.cr.\newline
SL2: 3s.cr.a +c.s.br +2d.s.ar +2d.s.br +a.b.drs +c.d.rs +2s.ab.cdr +2c.s.abdr +4d.s.abcr\newline
C2L2: 8s.cd.abr\newline
R: D-copy of A over BC.\newline
S: 8a.r.b +3a.r.c +2d.r.a +5d.r.b +c.d.r +3r.ab.cd +7c.r.abd +4d.r.abc\newline
C2: 14r.ad.bc\newline

{\bf Abbreviated Proof of (\ref{ineq3.29}):}
T: D-copy of C over ABRS.\newline
L1: -2a.b. +2a.b.c +a.b.cr +a.b.dr +c.dr. +a.b.crs +2c.rs.ab +2c.brs.a +2c.ars.b +a.b.drs +c.drs.\newline
SL1: t.cr.a +t.cr.b +a.b.crt +t.dr.a +t.dr.b +a.b.drt +c.d.rt +t.crs.a +t.crs.b +a.b.crst +t.drs.a +t.drs.b +a.b.drst +c.d.rst +t.ab.cdrs +c.t.abdrs +d.t.abcrs +t.ds.abcr +t.cs.abdr +t.abs.cdr
C2L1: 6t.cd.abrs\newline
S: C-copy of D over ABR.\newline
L2: -3a.b. +3a.b.c +a.b.cr +a.c.br +b.c.ar +3c.r.ab +c.br.a +c.ar.b +d.r.a +d.r.b +4a.b.dr +2d.br.a +2d.ar.b +3c.dr.\newline
SL2: s.cr.a +s.cr.b +2s.dr.a +2s.dr.b +a.b.drs +c.d.rs +2s.ab.cdr +2c.s.abdr +7d.s.abcr\newline
C2L2: 11s.cd.abr\newline
R: D-copy of C over AB.\newline
S: 2c.r.a +2c.r.b +3d.r.a +3d.r.b +c.d.r +4r.ab.cd +8c.r.abd +8d.r.abc\newline
C2: 20r.cd.ab\newline

{\bf Abbreviated Proof of (\ref{ineq3.30}):}
T: BC-copy of S over ADR.\newline
L1: -a.s. +a.r. +a.b.cr +a.d.br +c.d.r +a.r.s +2a.s.dr +d.s.ar\newline
SL1: 2a.r.t +b.t.ar +c.t.ar +c.t.br +a.b.crt +b.t.dr +a.d.brt +c.d.rt +t.ds.abcr +t.cs.abdr +t.abs.cdr
C2L1: 3t.bcs.adr\newline
S: D-copy of C over ABR.\newline
L2: -2a.c. +2a.r. +a.b.cr +4a.c.br +a.br.c +b.c.ar +3a.b.dr +a.d.br +2c.d.r\newline
SL2: a.r.s +c.s.ar +c.s.br +a.b.crs +3d.s.br +3a.b.drs +c.d.rs +s.ab.cdr +3c.s.abdr +d.s.abcr\newline
C2L2: 5s.cd.abr\newline
R: D-copy of A over BC.\newline
S: 6a.r.b +2a.r.c +3d.r.a +4d.r.b +c.d.r +3r.ab.cd +4c.r.abd +7d.r.abc\newline
C2: 14r.ad.bc\newline

{\bf Abbreviated Proof of (\ref{ineq3.31}):}
T: BCR-copy of S over AD.\newline
L1: 2(-a.b. +a.b.c +a.b.d +c.d. +a.d.s +a.s.d +d.s.a)\newline
SL1: b.t.a +c.t.a +c.t.b +a.b.ct +b.t.d +a.d.bt +c.d.t +t.drs.abc +t.crs.abd +t.abrs.cd
C2L1: 3t.bcrs.ad\newline
S: BD-copy of C over AR.\newline
L2: -3a.c. +2a.b.c +2a.d.b +3c.d. +6a.c.r +a.r.c +c.r.a +5a.r.d +d.r.a\newline
SL2: c.s.a +c.d.s +c.s.r +a.r.cs +5d.s.r +5a.r.ds +d.r.as +6s.bc.adr +s.bd.acr +s.abr.cd\newline
C2L2: 8s.bcd.ar\newline
R: D-copy of C over AB.\newline
S: 8c.r.b +8a.b.cr +6d.r.b +6a.b.dr +c.d.r +r.ab.cd +6c.r.abd +8d.r.abc\newline
C2: 15r.cd.ab\newline

{\bf Abbreviated Proof of (\ref{ineq3.32}):}
T: D-copy of A over BCRS.\newline
L1: -a.c. +c.d. +r.cd.a +c.s.a +b.d.s +a.bs.d +2a.cr.bs +a.bs.cr +d.r.abcs +d.s.abcr\newline
SL1: d.t.a +c.d.t +a.t.cd +c.r.t +a.t.cr +d.r.act +b.t.acdr +a.t.bs +c.s.at +b.t.acs +d.t.s +a.s.dt +b.d.ast +c.t.abds +a.r.bcst +r.ad.bcst +s.ad.bcrt +d.t.abcrs
C2L1: 3t.ad.bcrs\newline
S: C-copy of A over BDR.\newline
L2: -2a.c. +2c.d. +a.b.cr +2a.c.br +c.ar.b +a.b.dr +4a.d.br +2a.br.d +2d.br.a +2r.cd.a +d.r.abc\newline
SL2: c.s.b +a.b.cs +c.d.s +a.s.cd +d.s.abc +3a.s.br +3c.s.br +c.r.abs +d.r.s +a.s.dr +d.r.abs +d.br.as +c.r.ads +b.s.acdr +2c.s.abdr +2d.s.abcr\newline
C2L2: 7s.ac.bdr\newline
R: D-copy of C over AB.\newline
S: c.r.a +3c.r.b +d.r.a +7d.r.b +c.d.r +2b.r.acd +r.ab.cd +9c.r.abd +3d.r.abc\newline
C2: 16r.cd.ab\newline

{\bf Abbreviated Proof of (\ref{ineq3.33}):}
T: BDR-copy of S over AC.\newline
L1: 2(-a.b. +a.b.c +a.b.d +c.d. +a.c.s +a.s.c +c.s.a)\newline
SL1: b.t.a +b.t.c +a.c.bt +d.t.a +d.t.b +a.b.dt +c.d.t +t.drs.abc +t.crs.abd +t.abrs.cd
C2L1: 3t.bdrs.ac\newline
S: BC-copy of D over AR.\newline
L2: -3a.c. +2a.b.c +2a.d.b +3c.d. +5a.c.r +c.r.a +a.d.r +6a.r.d +d.r.a\newline
SL2: d.s.a +c.d.s +5c.s.r +5a.r.cs +c.r.as +d.s.r +a.r.ds +s.bc.adr +6s.bd.acr +s.abr.cd\newline
C2L2: 8s.bcd.ar\newline
R: D-copy of C over AB.\newline
S: 6c.r.b +6a.b.cr +8d.r.b +8a.b.dr +c.d.r +r.ab.cd +8c.r.abd +6d.r.abc\newline
C2: 15r.cd.ab\newline

{\bf Abbreviated Proof of (\ref{ineq3.34}):}
T: C-copy of A over BDRS.\newline
L1: -3a.d. +a.b.c +c.d. +2a.b.cr +c.r.ab +a.d.br +a.dr.b +2d.br.a +2d.cr. +d.bs.a +d.s.ac +a.br.ds +4a.ds.br +2a.bs.dr +2c.s.bdr +c.r.abds +3c.s.abdr\newline
SL1: c.t.b +a.b.ct +c.d.t +a.t.cd +3a.t.br +2c.t.ar +2c.t.br +2a.b.crt +c.r.abt +2d.r.t +2a.t.dr +2b.d.art +2d.cr.t +2b.t.cdr +3d.t.abcr +s.t.ab +c.t.as +d.s.t +a.t.ds +b.d.ast +d.s.act +b.t.acds +a.r.bdst +2a.s.bdrt +2c.s.bdrt +2a.t.bcdrs +r.ac.bdst +3s.ac.bdrt +3c.t.abdrs
C2L1: 9t.ac.bdrs\newline
S: D-copy of A over BCR.\newline
L2: -4a.c. +b.d.a +2c.d. +2a.b.cr +13a.c.br +c.br.a +c.ar.b +d.r.a +4a.b.dr +3a.d.br +2a.br.d +2d.br.a +2c.dr. +2d.r.ac +r.cd.ab\newline
SL2: c.d.s +a.s.cd +c.s.abd +2a.s.br +2c.r.s +2a.s.cr +b.c.ars +d.s.ar +9d.s.br +d.r.abs +c.r.ads +d.r.acs +2a.s.bcdr +2b.s.acdr +4c.s.abdr +d.s.abcr\newline
C2L2: 13s.ad.bcr\newline
R: D-copy of C over AB.\newline
S: 4c.r.a +15c.r.b +3d.r.a +9d.r.b +2c.d.r +2b.r.acd +4r.ab.cd +11c.r.abd +17d.r.abc\newline
C2: 35r.cd.ab\newline

\end{proof}
\section{Inequalities Using Four Copy Variables with at most Three Copy Steps}\label{FourCopies}

In this section we list the inequalities that can be derived from this method using
at most four auxiliary variables and at most three instances of the Copy Lemma.
Therefore, for the sake of completeness, this list includes those found in previous sections.
Each entry in the following list begins with a list of nine integers.  These integers represent the
coefficients $a,b,c,d,e,f,g,h,i$ for the non-Shannon inequality of the form
\begin{eqnarray*}
aI(A;B) &\leq& bI(A;B|C) + cI(A;C|B) + dI(B;C|A) \\
      &+&  eI(A;B|D) + fI(A;D|B)+ gI(B;D|A)\\
      &+& hI(C;D) + iI(C;D|A)
\end{eqnarray*}
Following the list of integers, is a list of copy steps used in the proof of the inequality.
Each copy step is ended by a period.

\begin{enumerate}
\item 2 4 2 1 3 1 0 2 0 r c ab.s r ad.
\item 2 3 3 1 5 2 0 2 0 rs cd ab.t r ad.
\item 3 6 3 1 6 3 0 3 0 rs cd ab.t r ad.
\item 2 4 2 1 2 0 0 2 3 rs bd ac.t c adr.
\item 2 3 3 2 2 0 0 2 0 rs cd ab.t b acr.
\item 4 6 4 3 4 2 1 4 0 rs cd ab.t b adr.
\item 2 5 2 1 2 0 0 2 0 rs ad bc.t c abr.
\item 2 4 3 1 2 0 0 2 0 rs ad bc.t b acr.
\item 2 4 1 2 2 3 0 2 0 rs cd ab.t a bcs.
\item 3 7 4 1 4 1 0 3 0 r c ab.s r ac.t r ad.
\item 4 6 11 3 6 2 0 4 0 r c ab.s c abr.t a bdrs.
\item 3 6 3 1 4 1 0 3 5 rs bd ac.tu cs adr.
\item 7 8 12 12 7 5 5 7 0 rs cd ab.tu (cr)(cs) ab.
\item 5 14 9 1 7 2 0 5 0 rs cd ab.tu cd ar.
\item 6 7 11 11 6 3 3 6 0 rs cd ab.tu (cr)(dr) ab.
\item 3 4 6 3 6 2 0 3 0 r c ab.st cd abr.u (cr) at.
\item 11 23 28 3 11 7 5 11 0 r a bc.st cd abr.u b acrt.
\item 5 6 8 7 5 3 2 5 0 r c ab.st cd abr.u b acrt.
\item 6 12 6 3 6 4 3 6 0 r b ac.st cd ab.u b act.
\item 4 5 16 4 10 6 0 4 0 r a bc.st ca bdr.u (at) bdr.
\item 3 6 5 1 5 3 0 3 0 r c ab.st ac bdr.u r ads.
\item 4 13 7 1 4 2 2 4 0 r b ac.st cd ab.u a bcrs.
\item 4 5 7 6 4 1 1 4 0 r c ab.st cd ab.u b acrs.
\item 4 8 4 1 10 6 0 4 0 r c ab.st bc ad.u r as.
\item 5 16 13 1 5 1 1 5 0 r a bc.st cd abr.u s acr.
\item 5 6 11 11 5 1 1 5 0 r c ab.st cd ab.u (cs) abr.
\item 2 3 4 1 4 5 0 2 0 rs cd ab.t c abrs.u a dst.
\item 4 5 6 4 4 2 4 4 0 rs cd ab.t d abrs.u b acrt.
\item 4 7 4 2 4 1 1 4 7 rs bd ac.t c adr.u a rst.
\item 2 3 3 1 2 3 2 2 0 rs cd ab.t d abrs.u a bcrt.
\item 4 10 2 3 9 9 0 4 0 rs cd ab.t c as.u r bdst.
\item 4 7 3 4 5 5 0 4 0 rs cd ab.t a bdr.u c abrst.
\item 3 5 4 2 3 0 0 3 4 rs bd ac.t c abs.u b acrst.
\item 5 14 11 1 5 2 2 5 0 rs cd ab.t a bcr.u a bcrt.
\item 6 15 10 2 6 0 0 6 11 rs bd ac.t c adr.u a bcrt.
\item 11 31 18 3 13 4 0 11 0 rs cd ab.t c ar.u d art.
\item 18 38 31 6 18 6 6 18 0 rs cd ab.t a bcr.u b acrt.
\item 4 9 3 2 8 4 0 4 0 rs cd ab.t c as.u r bst.
\item 5 12 3 3 10 9 0 5 0 rs cd ab.t c as.u r bdst.
\item 8 19 6 4 9 14 0 8 0 rs cd ab.t a bdr.u a bdrt.
\item 3 5 4 2 4 1 0 3 0 rs cd ab.t d ar.u b acrt.
\item 7 19 8 1 9 8 2 7 0 rs cd ab.t d ar.u c art.
\item 6 16 2 9 6 11 0 6 0 rs cd ab.t a bdr.u b acrt.
\item 7 8 11 11 7 7 7 7 0 rs cd ab.t c abrs.u (dr) abt.
\item 5 8 10 3 5 1 1 5 0 rs cd ab.t c ar.u b acrt.
\item 4 10 10 1 4 1 4 4 0 rs cd ab.t a bcr.u b adrt.
\item 8 9 14 14 8 8 8 8 0 rs cd ab.t c abrs.u (dr) abt.
\item 3 5 4 1 8 5 0 3 0 rs cd ab.t c as.u a bdst.
\item 6 11 10 2 6 6 9 6 0 rs cd ab.t d abr.u a bcrt.
\item 7 19 11 2 7 5 10 7 0 rs cd ab.t b adr.u a bcrt.
\item 6 13 10 2 6 2 2 6 0 rs cd ab.t a bcr.u b acrt.
\item 9 12 16 7 9 13 11 9 0 rs cd ab.t (dr) ab.u a bcrt.
\item 7 8 16 16 7 3 3 7 0 rs cd ab.t (cr) ab.u c abrt.
\item 5 9 4 4 5 3 1 5 0 rs cd ab.t a bcs.u (cs) abrt.
\item 8 17 7 3 16 8 0 8 0 rs cd ab.t c as.u c bst.
\item 3 9 2 2 3 0 0 3 0 rs bd ac.t b ar.u c abrt.
\item 9 17 15 3 9 15 5 9 0 rs cd ab.t d abrs.u a bcrt.
\item 3 4 5 4 4 1 0 3 0 rs cd ab.t a bcr.u a bdst.
\item 7 16 12 2 7 4 3 7 0 rs cd ab.t a bcr.u b adst.
\item 3 5 6 2 3 0 0 3 0 rs ad bc.t a cr.u b acrt.
\item 6 11 10 2 6 9 4 6 0 rs cd ab.t d abrs.u a bcrt.
\item 6 7 13 13 6 2 2 6 0 rs cd ab.t (cr) ab.u c abrst.
\item 10 23 16 3 10 5 5 10 0 rs cd ab.t a bcr.u d abst.
\item 4 5 9 6 6 3 0 4 0 rs cd ab.t a bcr.u a bdst.
\item 9 19 8 4 15 7 0 9 0 rs cd ab.t c as.u s act.
\item 10 11 22 22 10 9 9 10 0 rs cd ab.t (cr) ab.u (dr) abt.
\item 5 20 15 1 5 0 0 5 0 rs ad bc.t (ar) bc.u (ar) bct.
\item 4 8 11 1 4 4 4 4 0 rs cd ab.t a bcr.u s abcdt.
\item 7 16 12 2 7 3 4 7 0 rs cd ab.t a bcr.u a bdst.
\item 4 9 8 1 4 3 2 4 0 rs cd ab.t a bcr.u b adst.
\item 6 12 5 5 6 3 0 6 0 rs cd ab.t a bcs.u c abrt.
\item 7 12 10 3 14 6 0 7 0 rs cd ab.t c as.u r act.
\item 13 24 16 5 26 10 0 13 0 rs cd ab.t c as.u s act.
\item 6 14 4 3 12 11 0 6 0 rs cd ab.t c as.u c bdt.
\item 10 18 11 7 13 3 0 10 0 rs cd ab.t a bdr.u r abcdt.
\item 12 22 14 5 23 9 0 12 0 rs cd ab.t r ad.u a bdrt.
\item 4 14 9 1 4 0 0 4 0 rs ad bc.t (ar) bc.u c abrt.
\item 6 11 10 2 6 8 5 6 0 rs cd ab.t d abrs.u a bcrt.
\item 5 11 4 2 10 8 0 5 0 rs cd ab.t c as.u c bdt.
\item 8 14 13 4 12 8 0 8 0 rs cd ab.t a bcs.u a bdrt.
\item 10 20 9 6 12 7 0 10 0 rs cd ab.t a bcs.u c abrt.
\item 8 17 9 5 8 0 0 8 0 rs ad bc.t c abr.u c abrt.
\item 3 8 1 4 6 7 0 3 0 rs cd ab.t c as.u t cs.
\item 8 9 15 15 8 6 6 8 0 rs cd ab.t (cr) ab.u (dr) abt.
\item 7 8 21 21 7 2 2 7 0 rs cd ab.t (cr) ab.u (cr) abt.
\item 5 6 15 15 5 0 0 5 0 rs cd ab.t (cr) ab.u (cr) abt.
\item 5 6 5 5 5 10 10 5 0 rs cd ab.t (cs) ab.u (ds) abt.
\item 9 10 19 19 9 7 7 9 0 rs cd ab.t (cr) ab.u (dr) abt.
\item 6 7 17 17 6 1 1 6 0 rs cd ab.t (cr) ab.u (cr) abt.
\item 10 17 11 11 10 0 0 10 0 rs cd ab.t b acr.u b acrt.
\item 4 9 3 2 5 7 0 4 0 rs cd ab.t a bcs.u (at) bcs.
\item 4 6 5 4 4 4 0 4 0 rs cd ab.t a bcs.u (bt) acr.
\item 4 5 10 9 4 0 0 4 0 rs cd ab.t (cr) ab.u b acrt.
\item 3 4 5 5 3 0 0 3 0 rs cd ab.t b acr.u t bcr.
\item 6 13 5 3 8 0 2 6 10 rs bd ac.t c adr.u t rs.
\item 3 4 4 5 4 1 0 3 0 rs cd ab.t a bcr.u t acs.
\item 7 25 18 1 13 6 0 7 0 rs cd ab.t c ar.u t ad.
\item 10 17 8 9 10 4 2 10 0 rs cd ab.t a bcs.u a bcst.
\item 9 16 14 4 15 9 0 9 0 rs cd ab.t a bcs.u (rt) ad.
\item 9 10 18 18 9 8 8 9 0 rs cd ab.t (cr) ab.u (cs) abt.
\item 7 15 6 3 12 7 0 7 0 rs cd ab.t c as.u t cs.
\item 3 6 5 1 3 1 3 3 0 rs cd ab.t c abrs.u t adr.
\item 5 9 6 2 10 5 0 5 0 rs cd ab.t c as.u t abc.
\item 7 35 28 1 11 4 0 7 0 rs cd ab.t r ac.u s at.
\item 4 7 5 1 13 8 0 4 0 rs cd ab.t c as.u t ad.
\item 7 19 2 11 7 12 1 7 0 rs cd ab.t a bcs.u b acrt.
\item 4 13 9 1 4 0 0 4 7 rs bd ac.t c abs.u t ar.
\item 4 13 10 1 4 0 0 4 0 rs ad bc.t (ar) bc.u b acrt.
\item 3 8 5 1 3 0 0 3 0 rs ad bc.t b acr.u t abr.
\item 3 6 2 2 4 0 1 3 5 rs bd ac.t c adr.u t rs.
\item 4 9 8 1 4 2 3 4 0 rs cd ab.t a bcr.u a bdst.
\item 5 7 11 6 5 2 0 5 0 rs cd ab.t (cr) ab.u t adr.
\item 5 10 5 1 15 10 0 5 0 rs cd ab.t c as.u t ad.
\item 4 11 7 1 4 4 5 4 0 rs cd ab.t b acs.u t ar.
\item 8 17 7 4 10 0 2 8 12 rs bd ac.t c adr.u t rs.
\item 7 8 9 9 7 14 14 7 0 rs cd ab.t (cs) ab.u (ds) abt.
\item 4 5 10 6 7 2 0 4 0 rs cd ab.t (cr) ab.u (rt) ad.
\item 10 26 3 16 10 18 1 10 0 rs cd ab.t a bdr.u b acrt.
\item 9 12 22 10 13 10 0 9 0 rs cd ab.t (cr) ab.u a bdst.
\item 3 7 2 2 3 0 0 3 2 rs bd ac.t c abs.u t bcr.
\item 7 9 16 10 10 2 0 7 0 rs cd ab.t (cr) ab.u (rt) ad.
\item 6 9 6 4 6 3 4 6 0 rs cd ab.t b acs.u t bcr.
\item 4 8 2 5 4 0 0 4 4 rs bd ac.t c abs.u (at) bcr.
\item 13 37 22 3 17 6 0 13 0 rs cd ab.t c ar.u d art.
\item 8 10 23 12 8 1 4 8 0 rs cd ab.t (cr) ab.u a bcst.
\item 6 7 7 7 6 11 11 6 0 rs cd ab.t (cs) ab.u (ds) abt.
\item 11 29 3 19 11 21 2 11 0 rs cd ab.t a bcs.u b acrt.
\item 9 12 8 13 9 9 2 9 0 rs cd ab.t a bcs.u a bcrt.
\item 4 7 5 2 8 2 0 4 0 rs cd ab.t r ad.u r adt.
\item 4 11 7 1 4 1 7 4 0 rs cd ab.t b acs.u a bcrt.
\item 8 14 6 7 8 4 2 8 0 rs cd ab.t a bcs.u c abrst.
\item 6 8 6 5 6 4 5 6 0 rs cd ab.t b adr.u t bdr.
\item 7 13 5 6 7 5 1 7 0 rs cd ab.t a bdr.u r abcdt.
\item 12 16 25 18 12 3 0 12 0 rs cd ab.t (cr) ab.u t acs.
\item 8 12 9 5 8 4 8 8 0 rs cd ab.t b acs.u t bcr.
\item 5 6 14 9 5 1 2 5 0 rs cd ab.t (cr) ab.u a bdrt.
\item 15 27 45 5 15 27 11 15 0 rs cd ab.t (cs) ab.u a bcrt.
\item 10 18 12 4 21 7 0 10 0 rs cd ab.t r ad.u r adt.
\item 9 17 9 6 11 4 0 9 0 rs cd ab.t a bcs.u c abrst.
\item 11 29 16 3 15 6 0 11 0 rs cd ab.t s ac.u r act.
\item 18 22 51 30 18 3 6 18 0 rs cd ab.t (cr) ab.u a bcst.
\item 3 6 6 1 3 1 1 3 0 rs ad bc.t b ars.u b acr.
\item 18 32 13 16 18 13 2 18 0 rs cd ab.t a bcs.u (bt) acs.
\item 6 14 9 2 6 2 3 6 0 rs cd ab.t a bcr.u (ds) abt.
\item 14 24 13 14 14 11 0 14 0 rs cd ab.t a bcs.u (bt) acr.
\item 5 9 5 4 6 2 0 5 0 rs cd ab.t a bcs.u c abrt.
\item 6 17 2 9 9 3 0 6 0 rs cd ab.t c br.u t ad.
\item 4 11 7 1 4 2 6 4 0 rs cd ab.t b acs.u t ar.
\item 18 30 20 12 37 15 0 18 0 rs cd ab.t c as.u b acrt.
\item 8 10 22 13 8 1 2 8 0 rs cd ab.t (cr) ab.u a bdrt.
\item 5 8 3 5 5 5 2 5 0 rs cd ab.t a bcs.u t acr.
\item 9 19 7 5 12 8 0 9 0 rs cd ab.t a bcs.u t as.
\item 4 5 4 4 4 6 5 4 0 rs cd ab.t (dr) ab.u b adst.
\item 14 28 10 13 14 0 0 14 11 rs bd ac.t c abs.u (at) bcr.
\item 7 15 6 4 8 6 0 7 0 rs cd ab.t a bdr.u r abcdt.
\item 7 18 26 2 7 2 1 7 0 rs da bc.t (as) bc.u a bdst.
\item 11 19 34 4 11 20 8 11 0 rs cd ab.t (cs) ab.u a bcrt.
\item 4 5 5 4 4 3 3 4 0 rs cd ab.t b acr.u (dr) ab.
\item 5 13 8 1 8 3 0 5 0 rs cd ab.t c ar.u r ad.
\item 6 7 15 12 6 2 3 6 0 rs cd ab.t (cr) ab.u a bcst.
\item 14 21 17 9 14 6 12 14 0 rs cd ab.t b acs.u t bcr.
\item 3 5 2 4 3 2 0 3 0 rs cd ab.t a bdr.u t ar.
\item 4 7 4 3 6 2 0 4 0 rs cd ab.t d ar.u a bcr.
\item 19 32 21 12 40 17 0 19 0 rs cd ab.t c as.u b acrt.
\item 10 20 9 7 11 5 0 10 0 rs cd ab.t a bcs.u c abrt.
\item 8 21 34 2 8 3 2 8 0 rs ad bc.t (ar) bc.u r abst.
\item 6 8 6 7 6 2 1 6 0 rs cd ab.t b adr.u a bcr.
\item 4 10 2 3 5 1 0 4 4 rs bd ac.t b adr.u c abr.
\item 4 5 9 7 5 1 0 4 0 rs ad bc.t a bdr.u a cr.
\item 5 8 5 3 5 4 4 5 0 rs cd ab.t b acs.u t bcr.
\item 4 13 1 9 4 7 0 4 0 rs cd ab.t a bdr.u t bc.
\item 9 25 12 3 11 4 0 9 0 rs cd ab.t d ar.u c art.
\item 5 6 7 6 5 4 4 5 0 rs cd ab.t b acr.u (dr) ab.
\item 6 9 9 5 8 2 0 6 0 rs cd ab.t b acr.u r ad.
\item 7 11 6 6 7 3 2 7 0 rs cd ab.t a bdr.u b adr.
\item 4 9 5 1 7 3 0 4 0 rs cd ab.t c ar.u r ad.
\item 9 18 7 6 12 9 0 9 0 rs cd ab.t a bdr.u t ad.
\item 3 7 2 2 3 4 0 3 0 rs cd ab.t a bdr.u r abcdt.
\item 7 14 6 4 10 7 0 7 0 rs cd ab.t a bcs.u t as.
\item 6 8 8 5 6 2 1 6 0 rs cd ab.t b acr.u b acs.
\item 9 24 31 2 9 12 3 9 0 rs ad bc.t a bdr.u (ar) bct.
\item 9 24 43 2 9 4 3 9 0 rs ad bc.t (ar) bc.u r abst.
\item 5 9 7 2 5 3 2 5 0 rs cd ab.t a bcr.u b acs.
\item 13 20 8 14 13 12 6 13 0 rs cd ab.t a bdr.u t acr.
\item 5 7 7 4 5 1 2 5 0 rs cd ab.t b acr.u b adr.
\item 21 32 24 13 21 10 18 21 0 rs cd ab.t b acs.u t bcr.
\item 5 10 3 6 5 0 0 5 2 rs bd ac.t c adr.u a bcr.
\item 6 11 6 3 12 6 0 6 0 rs cd ab.t a bcr.u r ad.
\item 20 34 17 18 20 5 4 20 0 rs cd ab.t a bcs.u a bcst.
\item 24 40 17 24 24 13 6 24 0 rs cd ab.t a bdr.u d abst.
\item 5 11 14 1 5 6 6 5 0 rs cd ab.t a bcr.u (ds) abt.
\item 7 18 2 21 7 1 7 7 3 rs bd ac.t b adr.u (br) act.
\item 8 10 14 13 8 2 1 8 0 rs cd ab.t b acs.u (cr) ab.
\item 5 8 6 5 5 2 0 5 0 rs cd ab.t b acr.u a bdr.
\item 19 29 22 12 19 8 16 19 0 rs cd ab.t b acs.u t bcr.
\item 4 6 5 4 6 2 0 4 0 rs cd ab.t a bcr.u r ad.
\item 26 43 25 21 26 6 7 26 0 rs cd ab.t b acs.u b acst.
\item 11 19 28 4 11 20 12 11 0 rs cd ab.t (cs) ab.u a bcrt.
\item 5 10 19 1 5 8 8 5 0 rs cd ab.t a bcr.u (ds) abt.
\item 7 13 6 7 7 3 0 7 0 rs cd ab.t a bdr.u r abcdt.
\item 36 61 35 29 36 6 11 36 0 rs cd ab.t b acs.u b acst.
\item 6 9 8 4 6 2 5 6 0 rs cd ab.t b adr.u t bcr.
\item 10 12 18 17 10 4 3 10 0 rs cd ab.t b adr.u (cr) ab.
\item 14 36 45 4 14 17 2 14 0 rs ad bc.t r abs.u (ar) bct.
\item 14 22 10 5 16 20 7 14 0 rs cd ab.t s ac.u b acst.
\item 3 6 5 1 6 2 0 3 0 rs cd ab.t c ar.u r ad.
\item 11 19 30 4 11 20 10 11 0 rs cd ab.t (cs) ab.u a bcrt.
\item 20 34 20 16 20 3 6 20 0 rs cd ab.t b acs.u b acst.
\item 16 42 51 4 16 19 4 16 0 rs ad bc.t a bdr.u (ar) bct.
\item 15 25 15 12 15 3 4 15 0 rs cd ab.t b acs.u b acst.
\item 8 21 2 25 8 2 9 8 1 rs bd ac.t b adr.u (br) act.
\item 47 79 37 41 50 15 9 47 0 rs cd ab.t a bcs.u a bcst.
\item 25 66 81 6 25 30 7 25 0 rs ad bc.t r abs.u (ar) bct.
\item 20 30 14 15 24 8 7 20 0 rs cd ab.t b adr.u b adrt.
\end{enumerate}


\section{Extension Rules and Infinite Families of Inequalities}\label{infinite}

The first examples of infinite families of non-Shannon inequalities
were found by Matus \cite{Matus}.  For example, Matus showed (\cite{Matus} Corollary 3, formula (4))
that for any nonnegative integer s, the following is an
information inequality:
\begin{eqnarray} \label{Matus}
sI(A;B)&<=&(s(s+3)/2)I(A;B|C)\\
\nonumber &+&(s(s+1)/2)I(A;C|B)+I(B;C|A)\\
\nonumber &+&sI(A;B|D)+sI(C;D).
\end{eqnarray}
Here we have interpreted Matus' variables $\xi_1,\xi_2,\xi_3,\xi_4$ as $D,C,B,A$ respectively.
Note that in the case $s=0$ this is just the Shannon inequality $0 \leq I(B;C|A)$.
If $s=1$ we get the Zhang-Yeung inequality. When $s=2$ this is Inequality~(\ref{Matus2}).
But when $s=3$ the resulting inequality does not appear previously in this paper.
Each inequality in this list can be proved using the Copy Lemma method,
but the number of copies goes to $\infty$ as $s\to \infty$.
Matus used this list in a clever way to show that $\overline{\Gamma}^*_4$ is not polyhedral.

In this section we will introduce another method of quickly generating inequalities.  Rather than
listing them using an index $s$, we will use a more general approach which we call Extension Rules.
An Extension Rule turns one inequality of a certain type into another of the same type.
As an example, consider our first Extension Rule.

{\bf Rule [1]:} If
\begin{eqnarray*}
   aI(A;B)&\leq&bI(A;B|C)+cI(A;C|B)+dI(B;C|A)\\
   &+&aI(A;B|D)+fI(A;D|B)+gI(B;D|A)\\
   &+&hI(C;D)+iI(C;D|A)+jI(C;D|B)
\end{eqnarray*}
is an information inequality with nonnegative coefficients, then so is
\begin{eqnarray*}
    (a+d)I(A;B)&\leq&(b+2d+h+i)I(A;B|C)\\
    &+&(a+c+d+f+g)I(A;C|B)\\
    &+&dI(B;C|A)+(a+d)I(A;B|D)\\
    &+&fI(A;D|B)+gI(B;D|A)\\
    &+&(d+h)I(C;D)+iI(C;D|A)\\
    &+&jI(C;D|B).
\end{eqnarray*}

Notice that the premise inequality differs from the general form given in Section~\ref{FourCopies} in two ways.
First, an additional term $jI(C;D|B)$ has been added.  This of course, is not necessary since the coefficient
$j$ could be set to zero.  In other words, this term is only included because it can be and it is not
known whether it is ever useful.  Secondly, notice that the coefficient of $I(A;B|D)$ on the right matches the
coefficient of $I(A;B)$ on the left.  This is the only restriction (assuming it is of the general form)
needed to apply this rule.  Notice also that the conclusion inequality also satisfies this same restriction,
allowing us to iterate this rule indefinitely.

To see why this rule generalizes the Matus list~\eqref{Matus}, start with the Shannon inequality
$0 \leq I(B;C|A)$ and iterate the rule indefinitely.

Notice also that there are many other inequalities in Section~\ref{FourCopies} that satisfy
the special condition of the coefficients and can therefore be iterated using this rule.

We will now give an outline of the proof of Rule [1] and then give several other similar rules.
We observe that None of the rules discovered so far generalize the second infinite list of inequalities given by Matus:
\begin{eqnarray}
sI(A;B)&<=&(s(s+1)/2)I(A;B|C)\\
\nonumber &+&(s(s-1)/2)I(A;C|B)+2sI(A;B|D)\\
\nonumber &+&sI(A;D|B)+I(B;D|A)+sI(C;D).
\end{eqnarray}

\begin{proof}[Proof Outline for Rule 1]
The procedure for proving Rule [1] is similar to that in preceding sections.
Given a collection of already-known inequalities, one can perform
a polytope computation to produce a set of extreme rays.  Each of
these extreme rays can be tested against one or more copy
specifications; if such a test results in a contradiction,
then a new inequality can be deduced.  The difference here is
that, when testing an extreme ray against a copy specification,
one can use instances of previously-obtained inequalities,
involving both the given random variables and the copy
variables, in an attempt to reach a contradiction.

For example, suppose we take the Zhang-Yeung inequality
as known, and produce the list of extreme rays (see \eqref{newList}).
We now consider the copy specification "$R$ is a
$D$-copy of $A$ over $BC$".  Using the equations represented by
this copy specification, the Shannon inequalities on
variables $A,B,C,D,R,$ and the instances of the Zhang-Yeung
inequality obtained by substituting in combinations of
these five variables, one finds that some of the
listed extreme rays are not attainable.  This leads to
additional inequalities, such as:
\begin{eqnarray}\label{*}
    2I(A;B) &\leq& 5I(A;B|C)+3I(A;C|B)\\
\nonumber    && +I(B;C|A)+2I(A;B|D)+2I(C;D).
\end{eqnarray}

This inequality could be obtained purely using copy variables
(and Shannon inequalities); in fact, it is one of the six
inequalities we obtained in Section~\ref{TwoCopies} using two copy variables.
This is not surprising, because further investigation shows
that the process of obtaining this inequality here used only
one instance of Zhang-Yeung, namely:
\begin{eqnarray}\label{**}
    I(AR;BR) &\leq& 2I(AR;BR|CR)\\
\nonumber     &+& I(AR;CR|BR) + I(BR;CR|AR)\\
\nonumber     &+& I(AR;BR|D)+I(CR;D).
\end{eqnarray}
It only takes one additional copy variable to prove \eqref{**}
directly, so two copy variables in all suffice to prove \eqref{*}.

But further information can be extracted from the proof of \eqref{*}
which will allow us to get a substantially more general result.
First, we find that the proof of \eqref{*} from \eqref{**} actually shows
that the inequality
\begin{eqnarray}\label{***}
    2I(A;B) &\leq& 5I(A;B|C)+3I(A;C|B)+I(B;C|A)\\
\nonumber     &+& 2I(A;B|D)+2I(C;D)+I(AR;BR)\\
\nonumber     &-&2I(AR;BR|CR)-I(AR;CR|BR)\\
\nonumber     &-&I(BR;CR|AR)-I(AR;BR|D)-I(CR;D)
 \end{eqnarray}
is provable from "$R$ is a $D$-copy of $A$ over $BC$" and the
Shannon inequalities.  It turns out that \eqref{***} can be decomposed
into several simpler inequalities: each of the inequalities
\begin{eqnarray*}
       I(A;B) &\leq& I(A;C|B)+I(A;B|D)
                  +I(AR;BR)\\
                  &-&I(AR;BR|D)\\
            0 &\leq& I(A;B|C)
                  -I(AR;BR|CR)\\
            0 &\leq& I(A;C|B)
                  -I(AR;CR|BR)\\
       I(A;B) &\leq& 2I(A;B|C)+I(A;C|B)+I(B;C|A)\\
       &+&I(A;B|D)+I(C;D)
                  -I(BR;CR|AR)\\
            0 &\leq& I(A;B|C)+I(C;D)
                  -I(CR;D)
\end{eqnarray*}
is provable from "$R$ is a $D$-copy of $A$ over $BC$" and the
Shannon inequalities.  (This decomposition was found by a
trial-and-error approach.)

Now that we have divided \eqref{***} into these pieces, we
can reassemble the pieces with different coefficients
to get a more general rule for deducing inequalities from
old ones.  By putting arbitrary nonnegative coefficients $a,b,c,d,h$ on the
five inequalities above, instead of the coefficients $1,2,1,1,1$
used to get \eqref{***}, we get: the inequality
\begin{eqnarray*}\label{****}
     (a+d)I(A;B) &\leq& (b+2d+h)I(A;B|C)\\
\nonumber        &+& (a+c+d)I(A;C|B)+dI(B;C|A)\\
\nonumber        &+&(a+d)I(A;B|D)+(d+h)I(C;D)\\
\nonumber        &+&aI(AR;BR)-bI(AR;BR|CR)\\
\nonumber        &-&cI(AR;CR|BR)-dI(BR;CR|AR)\\
\nonumber        &-&aI(AR;BR|D)-hI(CR;D)
\end{eqnarray*}

is provable from "$R$ is a $D$-copy of $A$ over $BC$" and the
Shannon inequalities.  Therefore, if
\begin{eqnarray*}
    aI(A;B) &\leq& I(A;B|C)+cI(A;C|B)+dI(B;C|A)\\
    &+&aI(A;B|D)+hI(C;D)
\end{eqnarray*}
is an information inequality (i.e., always true, so there is no
problem substituting in $AR,BR,CR,D$ for $A,B,C,D$), then
\begin{eqnarray*}
    (a+d)I(A;B) &\leq& (b+2d+h)I(A;B|C)\\
    &+&(a+c+d)I(A;C|B)+dI(B;C|A)\\
    &+&(a+d)I(A;B|D)+(d+h)I(C;D)
\end{eqnarray*}
is also an information inequality.

It turns out that we can increase its applicability
by adding a few more pieces to the list needed for \eqref{***}.  Namely, we
use the fact that the inequalities
\begin{eqnarray*}
            0 &\leq& I(A;C|B)+I(A;D|B)
                  -I(AR;D|BR)\\
            0 &\leq& I(A;C|B)+I(B;D|A)
                  -I(BR;D|AR)\\
            0 &\leq& I(A;B|C)+I(C;D|A)
                  -I(CR;D|AR)\\
            0 &\leq& I(C;D|B)
                  -I(CR;D|BR)
\end{eqnarray*}
are provable from "$R$ is a $D$-copy of $A$ over $BC$" and the
Shannon inequalities.  Using these additional pieces, we can extend
the above rule to the form Rule [1].

Now the premise inequality (assuming it is of the general form
we have seen earlier) needs to meet only one restriction:
the coefficient of I(A;B|D) on the right must match the
coefficient of I(A;B) on the left.  (We tried to find
another piece which would remove this restriction, but
did not find a suitable one.)
\end{proof}

We now list the other extension rules that have been found,
where much of the above process has been summarized succinctly.  Each of the
lower case letters in these rules is assumed to be a nonnegative real number.

\input{Rules.txt}


\section{Connection with inequalities of Xu, Wang, and Sun}\label{XWS}

Another, very similar, method that automatically generates non-Shannon inequalities was
presented by Xu, Wang, and Sun \cite{XuWangSun}.  As an example of their method,
they give four new inequalities and also a separate infinite list of inequalities.

The new inequalities given \cite{XuWangSun}, Section V, are summarized in the
following two theorems.  The first theorem has four inequalities that were presented as
examples of their method in Sections V.A and V.B of \cite{XuWangSun}.
In order to connect these with the present work,
we have relabeled their variables $\xi_1, \xi_2, \xi_3, \xi_4$ as
$B,C,A,D$ in the first three of these inequalities and as $A,B,C,D$ in the fourth.
The inequalities are then rewritten in the form of Theorem~\ref{thm:many-inequalities}.
It should be emphasized that these inequalities were intended only as examples of their method
and should not be considered complete or exhaustive.  In the second theorem below we present their
infinite list of inequalities from Section V.C of of \cite{XuWangSun}.
The variables $\xi_1, \xi_2, \xi_3, \xi_4$
are again relabeled as $A,B,C,D$ and the inequalities are rewritten to match Theorem~\ref{thm:many-inequalities}.

\begin{theorem}[Xu, Wang, Sun]\label{thXWS1}
The following is an information inequality
\begin{eqnarray*}
aI(A;B) &\leq& bI(A;B|C) + cI(A;C|B) + dI(B;C|A) \\
      &+&  eI(A;B|D) + fI(A;D|B)+ gI(B;D|A)\\
      &+& hI(C;D) +iI(C;D|A)
\end{eqnarray*}
for each of the following values of $(a,b,c,d,e,f,g,h,i)$.
\begin{align}
&&(8&,&12&,&33&,&10&,&8&,&15&,&1&,&8&,&0&)&&&& \label{XWS1}\\
&&(5&,&7&,&20&,&5&,&6&,&9&,&1&,&5&,&0&)&&&& \label{XWS2}\\
&&(4&,&5&,&17&,&6&,&6&,&7&,&0&,&4&,&0&)&&&& \label{XWS3}\\
&&(3&,&4&,&6&,&6&,&3&,&1&,&1&,&3&,&0&)&&&& \label{XWS4}
\end{align}
\end{theorem}

\begin{theorem}[Xu, Wang, Sun]\label{thXWS2}
For each positive integer $s$, the following is an information inequality
\begin{eqnarray*}
(2^{s-1}-1)I(A;B) &\leq& 2^{s-1}I(A;B|C) \\
      &+& (c_1-2^{s-1})I(A;C|B)\\
      &+& (c_1-2^{s-1})I(B;C|A) \\
      &+&  (2^{s-1}-1)I(A;B|D) \\
      &+& (c_2-2^{s-1}+1)I(A;D|B)\\
      &+& (c_2-2^{s-1}+1)I(B;D|A)\\
      &+& (2^{s-1}-1)I(C;D),
\end{eqnarray*} where
\begin{eqnarray*}
c_1 &=& (S_++S_-)/4\\
c_2 &=& ((\sqrt2-1)S_+ - (\sqrt2+1)S_-)/4\\
S_+ &=& (2+\sqrt{2})^s\\
S_- &=& (2-\sqrt{2})^s\\
\end{eqnarray*}
\end{theorem}

In this section we will show how to derive the inequalities of Xu, Wang, and Sun by
adding together linear combinations of inequalities from the present work.  This will usually
result in stronger versions of these inequalities.

To begin, first note that Inequality \eqref{XWS3} is identical to \eqref{ineq3.20}.
It does not appear in the list given in Section \ref{FourCopies} because it has been
superseded.
We also note that Inequality \eqref{XWS4} is a weaker form of
\begin{align*}
&&(3&,&4&,&4&,&4&,&3&,&1&,&1&,&3&,&0)&&&& \eqref{ineq3.1},
\end{align*}
which has also been superseded.

To get Inequality \eqref{XWS1}, we add together the previous inequalities
\begin{align*}
&&(3&,&4&,&6&,&6&,&3&,&0&,&0&,&3&,&0)&&&& \eqref{ineq3.3}\\
&&(4&,&6&,&4&,&3&,&4&,&2&,&1&,&4&,&0)&&&& \eqref{ineq3.12}\\
&&(2&,&4&,&1&,&2&,&2&,&3&,&0&,&2&,&0)&(1/2)\hspace{2pt}&&& \eqref{ineq3.15}.
\end{align*}
This gives us the inequality
\begin{align*}
&&(8&,&12&,&21/2&,&10&,&8&,&7/2&,&1&,&8&,&0)&&&& \eqref{ineq3.12}
\end{align*}
which is an improvement of \eqref{XWS1}.

To get Inequality \eqref{XWS2}, we first switch the variables $C$ and $D$ in the inequality
\begin{align*}
&&(3&,&7&,&5&,&1&,&3&,&1&,&1&,&3&,&0)&&&&&\;\; \eqref{ineq3.11}
\end{align*}
to obtain the inequality
\begin{align}
&&(3&,&3&,&1&,&1&,&7&,&5&,&1&,&3&,&0)&&&&& \hspace{-10pt}\label{ineq3.11Alt}
\end{align}

Then combining the inequalities
\begin{align*}
&&(3&,&4&,&6&,&6&,&3&,&0&,&0&,&3&,&0)&(20)&&& \eqref{ineq3.3}\\
&&(4&,&5&,&8&,&8&,&4&,&1&,&1&,&4&,&0)&(2)&&& \eqref{ineq3.7}\\
&&(3&,&3&,&1&,&1&,&7&,&5&,&1&,&3&,&0)&(14)&&&  \eqref{ineq3.11Alt}\\
&&(4&,&6&,&4&,&3&,&4&,&2&,&1&,&4&,&0)&(40)&&& \eqref{ineq3.12}\\
&&(2&,&4&,&1&,&2&,&2&,&3&,&0&,&2&,&0)&(5)&&& \eqref{ineq3.15}
\end{align*}
and dividing by 56, we get the inequality
\begin{align*}
&&(5&,&7&,&45/8&,&5&,&6&,&167/56&,&1&,&5&,&0)&&&& \hspace{-10pt}\eqref{ineq3.3}
\end{align*}
which is an improvement of \eqref{XWS2}.

To derive the infinite list of inequalities given in Theorem~\ref{thXWS2}, we will actually
derive a stronger list.
\begin{theorem}\label{thXWS3}
For each positive integer $s$, the following is an information inequality
\begin{eqnarray*}
(2^{s-1}-1)I(A;B) &\leq& 2^{s-1}I(A;B|C)\\
      &+& (s-1)2^{s-2}I(A;C|B)\\
      &+& (s-1)2^{s-2}I(B;C|A)) \\
      &+&  (2^{s-1}-1)I(A;B|D) \\
      &+& ((s-3)2^{s-2}+1)I(A;D|B)\\
      &+& ((s-3)2^{s-2}+1)I(B;D|A))\\
      &+& (2^{s-1}-1)I(C;D)
\end{eqnarray*}
\end{theorem}
\begin{proof}
Start with the Shannon inequality $0\leq I(A;B|C)$ to get the case $s=1$.
Once the case $s$ has been established, permute the variables $C,D$ to get
\begin{eqnarray*}
(2^{s-1}-1)I(A;B) &\leq& (2^{s-1}-1)I(A;B|C)\\
      &+& ((s-3)2^{s-2}+1)I(A;C|B)\\
      &+& ((s-3)2^{s-2}+1)I(B;C|A)) \\
      &+&  2^{s-1}I(A;B|D) \\
      &+& (s-1)2^{s-2}I(A;D|B)\\
      &+& (s-1)2^{s-2}I(B;D|A))\\
      &+& (2^{s-1}-1)I(C;D).
\end{eqnarray*}
Then apply Rule [6] from Section~\ref{infinite}  with the substitutions
\begin{eqnarray*}
a=h=a'=h' &=& 2^{s-1}-1\\
d=c=g'=f' &=& (s-3)2^{s-2}+1\\
g=f=c'=d' &=& (s-1)2^{s-2}\\
z&=&1\\
i=j=i'=j'&=& 0
\end{eqnarray*}
to get the inequality for the case $s+1$.
\end{proof}

To see why this improves Theorem~\ref{thXWS2} note that in that theorem, the coefficients
$c_1$ and $c_2$ follow the pattern:
\begin{align*}
&&c_1&=&1&,&3&,&10&,&34&,&116&,&\ldots&&&\\
&&c_2&=&0&,&1&,&4&,&14&,&48&,&\ldots&&&
\end{align*}
and both of these satisfy the recursion
\begin{equation*}
c_i(s+2) = 4*c_i(s+1) - 2*c_i(s)
\end{equation*}
and grow at the rate of $O((2+\sqrt2)^s)$
Since these majorize the sequences
\begin{align*}
(s-1)2^{s-2} + 2^{s-1}&=&1&,&3&,&8&,&20&,&48&,&\ldots&\\
(s-3)2^{s-2}+1 + 2^{s-1}-1&=&0&,&1&,&4&,&12&,&32&,&\ldots&
\end{align*}
respectively, the right side of the inequality in Theorem~\ref{thXWS2}
is strictly larger (for $s>1$) than the right side of the inequality from Theorem~\ref{thXWS3}.
We note that for $s=1$ both lists yield the Shannon inequality $0\leq I(A;B|C)$.
For $s=2$ both lists yield the Zhang-Yeung inequality.  For $s=3$ Theorem~\ref{thXWS2}
yields Inequality~\eqref{XWS4} while Theorem~\ref{thXWS3} yields the stronger Inequality \eqref{ineq3.1}.


\section{Structure of $\bar\Gamma^*_4$}\label{structure}

Given any four random variables, $A,B,C,D$ we can form a vector of the fifteen joint entropies:
\begin{eqnarray}\label{R15}
\nonumber &&\langle H(A), H(B), H(AB), H(C), H(AC), H(BC),\\
\nonumber && H(ABC),H(D), H(AD), H(BD), H(ABD), H(CD),\\
&& H(ACD), H(BCD), H(ABCD) \rangle.
\end{eqnarray}

The space of all such fifteen dimensional vectors is called $\Gamma^*_4$.
This notation as well as some of the basic properties of this space are given in \cite{Yeung-book}.
For example, this space is not closed, but it's closure, $\bar\Gamma^*_4$, is convex and forms a cone with vertex at the origin.
The spaces $\Gamma^*_4$ and $\bar\Gamma^*_4$ share the same interior, but on the boundary of these regions,
$\Gamma^*_4$ is very complicated.  Here we will concentrate on the closed entropy space, $\bar\Gamma^*_4$.

Each information inequality forms an outer bound of this region.  The superset of $\bar\Gamma^*_4$
defined by the Shannon inequalities is denoted by $\Gamma_4$, which is also a convex cone.
Note that although expressions such as $I(A;B|C)$ and $H(ABCD)$ are defined for probability distributions,
we can easily extend these definitions in the natural way to any element of $\R^{15}$ using \eqref{R15}.  Thus, if
$(x_1, \ldots, x_{15})$ is an element of $\R^{15}$, then  $H(ABCD)= x_{15}$
while $I(A;B|C) = H(AC)+H(BC)-H(C)-H(ABC)=x_{5}+x_{6}-x_{4}-x_{7}$.

We begin with a list of four Shannon inequalities:
\begin{eqnarray*}
I(A;B) &\geq& 0 \\
I(A;B|C) &\geq& 0 \\
I(A;B|CD) &\geq& 0 \\
H(A|BCD) &\geq& 0
\end{eqnarray*}
These four inequalities together with the ones that can be formed from these by permuting the variables, form
the list of twenty-eight {\em elemental} Shannon Inequalities. This list is complete in the sense that
the space they define is exactly $\Gamma_4$ (see \cite{Yeung-book}).

This space $\Gamma_4$ can be reduced by combining the Shannon inequalities with six Ingleton inequalities, represented by the six permutations of:
\begin{equation}\label{ingleton}
I(A;B) \leq I(A;B|C)+I(A;B|D)+I(C;D) \, .
\end{equation}
We will refer to this reduced space as the {\em Ingleton region} and denote it by $I_4$.

The Ingleton inequalities are not information inequalities as defined in this paper; they are not valid for all entropy vectors.
Their intended application, which is the study of linear ranks, does not concern us here.
Nevertheless as we will see below, the Ingleton region is very important in the structure of entropy space.
For one thing, every vertex of $I_4$ has been shown to be
an entropy vector for an actual probability distribution.  It follows by convexity that
$I_4 \subseteq \bar\Gamma^*_4$.  Therefore, it is outside the Ingleton region where
the non-Shannon inequalities are relevant.

Since the spaces $\bar\Gamma^*_4$, $\Gamma_4$, and $I_4$ are infinite cones, it is often convenient to study just a cross-section.
>From here on, we will normalize each of these spaces by using the equation $H(ABCD)=1$.  To reduce notation, we will keep the same
names for the normalized spaces, but we will consider them as subsets of $\R^{14}$, since the last entry
in \eqref{R15} is held constant.  These normalized spaces carry all the information of the original spaces, but are
easier to describe.
For example, whereas $\Gamma_4$ was originally a polytopal cone in $\R^{15}$, it is now a polytope in $\R^{14}$.

Sticking out of the normalized Ingleton region are six ``pyramid" shaped figures.  There is one such pyramid for
each Ingleton inequality.  Each pyramid is defined by the points that satisfy all Shannon inequalities, but fail the
respective Ingleton inequality.  These pyramids are simplices in $\R^{14}$;
each is bounded by an Ingleton equality and fourteen Shannon faces.

For our purposes, we will concentrate on just one of the pyramids, the pyramid associated with \eqref{ingleton},
and we will henceforth refer to this as \textit{the pyramid}. The \textit{base} of this pyramid is formed by the
Ingleton equation \eqref{ingleton} with ``$\leq$" replaced by ``$=$".  In terms of the fourteen coordinates, this equation becomes
\begin{equation}\label{eq:Ingleton}
x_{1}+x_{2}+x_{7}+x_{11}+x_{12}=x_{3}+x_{5}+x_{6}+x_{9}+x_{10}
\end{equation}

It is also bounded by the fourteen Shannon faces:

\begin{tabular}{|rcc|rcl|}
\multicolumn{6}{c}{}\\
\hline
\multicolumn{3}{|c|}{Shannon equation}& \multicolumn{3}{c|}{coordinate equation}\\
\hline
$H(A|BCD)$&=&$0$ & $x_{14}$&=&$1$\\
$H(B|ACD)$&=&$0$ & $x_{13}$&=&$1$\\
$H(C|ABD)$&=&$0$ & $x_{11}$&=&$1$\\
$H(D|ABC)$&=&$0$ & $x_{7}$&=&$1$\\
$I(C;D|A)$&=&$0$ & $x_{5}+x_{9}$&=&$x_{1}+x_{13}$\\
$I(C;D|B)$&=&$0$ & $x_{6}+x_{10}$&=&$x_{2}+x_{14}$\\
$I(A;C|B)$&=&$0$ & $x_{3}+x_{6}$&=&$x_{2}+x_{7}$\\
$I(B;C|A)$&=&$0$ & $x_{3}+x_{5}$&=&$x_{1}+x_{7}$\\
$I(A;D|B)$&=&$0$ & $x_{3}+x_{10}$&=&$x_{2}+x_{11}$\\
$I(B;D|A)$&=&$0$ & $x_{3}+x_{9}$&=&$x_{1}+x_{11}$\\
$I(A;B|C)$&=&$0$ & $x_{5}+x_{6}$&=&$x_{4}+x_{7}$\\
$I(A;B|D)$&=&$0$ & $x_{9}+x_{10}$&=&$x_{8}+x_{11}$\\
$I(A;B|CD)$&=&$0$ & $x_{13}+x_{14}$&=&$x_{12}+1$\\
$I(C;D)$&=&$0$ & $x_4+x_8$&=&$x_{12}$\\
\hline
\multicolumn{6}{c}{}\\
\end{tabular}

There are fourteen corresponding base points of the pyramid.  Each base point satisfies
the Ingleton equation~\eqref{eq:Ingleton} and all but one of the fourteen Shannon equations.
The base points are:
\begin{eqnarray*}
&&(1,0,1,0,1,0,1,0,1,0,1,0,1,0)\\
&&(0,1,1,0,0,1,1,0,0,1,1,0,0,1)\\
&&(0,0,0,1,1,1,1,0,0,0,0,1,1,1)\\
&&(0,0,0,0,0,0,0,1,1,1,1,1,1,1)\\
&&(\frac{1}{2},0,\frac{1}{2},\frac{1}{2},1,\frac{1}{2},1,\frac{1}{2},1,\frac{1}{2},1,1,1,1)\\
&&(0,\frac{1}{2},\frac{1}{2},\frac{1}{2},\frac{1}{2},1,1,\frac{1}{2},\frac{1}{2},1,1,1,1,1)\\
&&(1,0,1,1,1,1,1,0,1,0,1,1,1,1)\\
&&(0,1,1,1,1,1,1,0,0,1,1,1,1,1)\\
&&(1,0,1,0,1,0,1,1,1,1,1,1,1,1)\\
&&(0,1,1,0,0,1,1,1,1,1,1,1,1,1)\\
&&(1,1,1,0,1,1,1,1,1,1,1,1,1,1)\\
&&(1,1,1,1,1,1,1,0,1,1,1,1,1,1)\\
&&(\frac{1}{3},\frac{1}{3},\frac{2}{3},\frac{1}{3},\frac{2}{3},\frac{2}{3},1,\frac{1}{3},\frac{2}{3},\frac{2}{3},1,\frac{2}{3},1,1)\\
&&(1,1,1,1,1,1,1,1,1,1,1,1,1,1)\; .\\
\end{eqnarray*}

The remaining vertex of the pyramid will be called the \textit{top of the pyramid}.  This point satisfies all
fourteen Shannon equations, but fails the Ingleton inequality.  It is the point

\begin{equation}\label{top}
(\frac{1}{2},\frac{1}{2},\frac{3}{4},\frac{1}{2},\frac{3}{4},\frac{3}{4},1, \frac{1}{2},\frac{3}{4},\frac{3}{4},1,1,1,1).
\end{equation}

No point in $\Gamma_4$ can simultaneously fail two Ingleton inequalities and so
the six pyramids have disjoint interiors.  However, the same Shannon equation can define a face on more than one of these pyramids.
For example, each of the six pyramids shares the first four Shannon faces in the list above, since
this set is closed under permutations of the variables.  Also,
$I(A;B|C) = 0$ has twelve forms obtained by permuting the variables, and each pyramid contains eight of them.
However, the equations
of the form $I(A;B|CD) = 0$ and $I(C;D) = 0$ each have six permuted forms, and each pyramid
contains exactly one such pair.

The non-Shannon inequalities give us additional outer bounds for $\bar\Gamma^*_4$.
Together, they gradually chip away each of the pyramids.
We will often consider these inequalities in groups according to how many auxiliary variables are used.
Thus, the region formed by the Shannon inequalities will be referred to as ``0 copies".  When we add the
Zhang-Yeung inequality the corresponding outer-bound region will be referred to as ``1 copy".
Similarly, when the inequalities from Section~\ref{TwoCopies} are added we get the region of ``2 copies",
when the inequalities from Section~\ref{ThreeCopies} are added  we get the region of ``3 copies",
and adding the inequalities from Section~\ref{ThreeCopiesB} give us the region of ``3.5 copies".
The more
inequalities we add, the more complicated the outer-bound region seems to become.  Here we show how the number of inequalities (faces)
and the number of vertices in the normalized outer bound grows with the number of copy steps.

\begin{tabular}{|c|c|c|}
\multicolumn{3}{c}{}\\
   \hline
  copies  & faces & vertices \\
   \hline
  0 [Shannon]& 28 & 41 \\
  1 [Zhang-Yeung]& 40 & 89 \\
  2 [Section ~\ref{TwoCopies}]& 160 & 299 \\
  3 [Section ~\ref{ThreeCopies}]& 796 & 10361 \\
  3.5 [Section ~\ref{ThreeCopiesB}]& 4924 & 224801 \\
  \hline
  \multicolumn{3}{c}{}
\end{tabular}

One way of estimating progress toward characterizing $\Gamma^*_4$ is by the decrease in volume
of the outer bounds.
Thus, the Shannon region, $\Gamma^N_4$, has total volume $$1.9787036156085\cdot 10^{-10}.$$
The Ingleton region has volume $$1.964365183611\cdot 10^{-10}$$ leaving a volume of
$$ 2.38973866625 \cdot  10^{-13}$$ for each of the six pyramids.  The following table shows how
the volume decreases within each pyramid as more non-Shannon inequalities are found.

\begin{tabular}{|c|c|}
\multicolumn{2}{c}{}\\
  \hline
  copies  & percent of pyramid left\\
   \hline
  0 [Shannon]& 100\\
  1 [Zhang-Yeung]& 98.4568\\
  2 [Section ~\ref{TwoCopies}]& 97.7040\\
  3 [Section ~\ref{ThreeCopies}]& 96.7214\\
  3.5 [Section ~\ref{ThreeCopiesB}]& 96.4682\\
  \hline
  \multicolumn{2}{c}{}
\end{tabular}

Of course, we do not know how much progress we are making without knowing the goal, which is
the proportion of each pyramid taken up by $\bar\Gamma^{*N}_4$.
Although this volume is unknown, we can get lower bounds.
All we have to do is find many probability distributions of four random variables, compute their joint entropies,
obtain the corresponding elements of $\R^{14}$ and
then compute the volume of the convex hull of these points.
The problem with this technique is that computationally, we can only handle relatively few points.

Nevertheless, using this technique with eight points (not necessarily the best eight points)
 we determine that the inner volume of
$\bar\Gamma^*_4$ is at least $53.4815 $ percent of each pyramid.  It is clear that according to this measure, there is still quite a gap
to fill, with room for improvement in both inner and outer bounds .

Perhaps one of the most interesting entropy vectors is the one that yields the worst violation of Ingleton's inequality~\eqref{ingleton}.
\begin{definition}
Given a probability distribution, we define the Ingleton score of the distribution to be the value of
$$
\frac{I(A;B|C)+I(A;B|D)+I(C;D)-I(A;B)}{H(ABCD)}
$$
\end{definition}

\begin{conjecture}[Four-atom conjecture]
The lowest possible Ingleton score is approximately -.089373.  It is attained by a four-variable
distribution with alphabet size two, given by
\begin{eqnarray*}
P(0,0,0,0)=P(1,1,1,1)&=&\alpha\\
P(0,1,0,1)=P(0,1,1,0)&=&.5-\alpha
\end{eqnarray*}
where $\alpha$ takes on a value which is approximately 0.350457.
\end{conjecture}

Evidence for this conjecture comes from hill-climbs.  Using Newton's method, we climb from random starting distributions
using alphabet sizes from 2 to 10.  We gradually make changes to the distribution in order to
decrease the Ingleton score.  The most successful climbs, i.e. the climbs that result
 in the lowest score, always tend to the four-atom distribution given in
the conjecture, some permutation of it, or a direct product of independent copies of such distributions.
In other words, we easily find the score $-.089373$ over and over again, but never beat it.

This suggests another way of measuring progress by how close we come to proving the Four-Atom
conjecture.
Given any outer-bound region, we can consider the lowest Ingleton score allowed in the region.
Progress toward the Four-Atom Conjecture is given in the following table:

\begin{tabular}{|c|c|}
\multicolumn{2}{c}{}\\
  \hline
  copies  & minimum Ingleton score\\
   \hline
  0 [Shannon]& -1/4 $\thickapprox$ -.25000\\
  1 [Zhang-Yeung]& -1/6$\thickapprox$ -.16667\\
  2 [Section ~\ref{TwoCopies}]& -1/6$\thickapprox$ -.16667\\
  3 [Section ~\ref{ThreeCopies}]& -7/44 $\thickapprox$ -.15909\\
  3.5 [Section ~\ref{ThreeCopiesB}]& -3/19 $\thickapprox$ -.15789 \\
  \hline
  \multicolumn{2}{c}{}
\end{tabular}

Again, progress toward the conjectured goal of -.089373 seems slow.  In fact,
even if we combine all the non-Shannon inequalities in this paper, it is still not enough to
prove the four-atom conjecture and it is likely that some new techniques will be necessary to settle it.

We now present a third way of measuring progress of the non-Shannon inequalities.
Unlike the previous two methods, this method makes explicit use of the infinite lists of inequalities
from Section~\ref{infinite}.  We begin with some definitions.

\begin{definition}
Let $P$ be a polytope in $\R^{n}$, let $L$ be an affine function of $n$-variables such that
for each point $x$ in $P$, $L(x)\geq 0$.  Then the set of points in $P$ satisfying $L(x)=0$ is called
an {\em extreme} segment of $P$.  Thus, the vertices of $P$ are the extreme points of $P$, the
edges of $P$ are the extreme line segments of $P$, etc.
\end{definition}

\begin{definition}
Let $P$ be a polytope in $\R^{n}$ and let $0\leq m \leq n$.  Then the $m$-skeleton of $P$ is the
union of all of the extreme $m$-dimensional segments of $P$.
\end{definition}

\begin{definition}
Let $P$ be a polytope in $\R^{n}$, let $0\leq m \leq n$, and let $S$ be a subset of $P$.
Then the $m$-core of $S$ in $P$ is the convex hull of the intersection of $S$ with the $m$-skeleton of $P$.
\end{definition}

It is easy to see from the Zhang-Yeung inequality that $I_4$ is the 0-core of
$\bar\Gamma^*_4$ in $\Gamma_4$.  First of all,
the polytope $I_4$ is a subset of both $\bar\Gamma^*_4$ and $\Gamma_4$, so we only need to show that
$\bar\Gamma^*_4$ does not contain any vertices of $\Gamma_4$ that lie outside of $I_4$.
Outside of $I_4$ the only vertices
of $\Gamma_4$ are the tops of the pyramids.  But these tops fail
one of the permuted forms of the Zhang-Yeung inequality.  For example,
the top of the pyramid give by \eqref{top} fails the inequality of Theorem~\ref{th:ZY2}.

We will now use the Zhang-Yeung inequality to show that  $I_4$ is also the 1-core of $\bar\Gamma^*_4$ in $\Gamma_4$.
Then we will use the infinite families of inequalities to prove the stronger result
that it is also the 3-core.
So the Ingleton region, aside from its applications to linear ranks, is a natural part of the geometry of entropy space.

\begin{theorem}
The 1-core of $\bar\Gamma^*_4$ in $\Gamma_4$ is $I_4$.
\end{theorem}

\begin{proof}
Since the polytope $I_4$ is a subset of both $\bar\Gamma^*_4$ and $\Gamma_4$,
we need to show that the intersection of $\bar\Gamma^*_4$ with the 1-skeleton of
$\Gamma_4$ lies entirely in $I_4$.

We only need the Zhang-Yeung inequality,
\begin{eqnarray*}
I(A;B) &\leq& 2I(A;B|C) + I(A;C|B) + I(B;C|A)\\
&+&I(A;B|D)+I(C;D)
\end{eqnarray*}
and the Ingleton inequality~\eqref{ingleton}.
\begin{equation}
I(A;B) \leq I(A;B|C)+I(A;B|D)+I(C;D) \, .
\end{equation}
Subtracting the two, we get the inequality
$$
0 \leq I(A;B|C) + I(A;C|B) +I(B;C|A)
$$
which must be satisfied by any point on the base of the pyramid, since the Ingleton equation~\eqref{eq:Ingleton}
is satisfied by such points.
Also, if any base point has a zero value for all of the terms
$$
I(A;B|C), I(A;C|B), I(B;C|A)
$$
Then it satisfies the Zhang-Yeung inequality non-strictly (as an equation).

We now consider an extreme line segment connecting
a two vertices of $\Gamma_4$, one of them being the top of the pyramid.  Since this
is the only vertex of $\Gamma_4$ that fails the Ingleton inequality, the other vertex satisfies the
Ingleton inequality.  Therefore, at some point on this segment, the Ingleton equation is satisfied.
We will call this the base point of the segment.

There are thirteen Shannon equations satisfied by all of the points on this segment.
Since the top of the pyramid is defined by the fourteen Shannon faces of the pyramid,
the thirteen Shannon equations satisfied by the line segment must be included in this group of fourteen.
Therefore, all points on the line segment are either zero on all three of the terms above, or on all three
of the terms with $C$, $D$ interchanged,
$$
I(A;B|D), I(A;D|B), I(B;D|A)\, ,
$$
since these represent six of the fourteen Shannon faces.  Therefore, the base point of this line segment
either satisfies the Zhang-Yeung equation
or the Zhang-Yeung equation with $C$, $D$ interchanged.
On the other hand, it is easy to check that the top of the pyramid~\eqref{top} fails both the Zhang-Yeung inequality,
and the same inequality with $C$, $D$ interchanged.
By linearity, one of the two Zhang-Yeung inequalities must fail everywhere on the
part of the line segment where the Ingleton inequality fails.
Therefore the intersection of $\bar\Gamma^*_4$ with the 0-skeleton of $\Gamma_4$ lies entirely in $I_4$.
\end{proof}

We will now prove a stronger result using a similar technique, and making use of the infinite families of
inequalities.

\begin{theorem}
The 3-core of $\bar\Gamma^*_4$ in $\Gamma_4$ is $I_4$.
\end{theorem}

\begin{proof}
As in the previous proof, we need to show that no point of
$\bar\Gamma^*_4\setminus I_4$ is in the 3-skeleton of $\Gamma_4$.
We suppose there is such a point and call it $x$.  Then $x$ lies on a three dimensional extreme segment that includes
at least four vertices of $\Gamma_4$, one of which we may assume is the top of the pyramid and is the only one that
fails the Ingleton inequality.
The segment is defined by eleven Shannon faces, and since the segment includes the top of the
pyramid, these eleven faces must be among the fourteen Shannon faces of the pyramid.

Suppose first that the eleven faces includes the two faces
\begin{eqnarray*}
I(A;B|C)&=&0\\
I(A;C|B)&=&0\, .
\end{eqnarray*}
Then consider the first Matus family of inequalities~\eqref{Matus}
\begin{eqnarray*}
sI(A;B)&<=&(s(s+3)/2)I(A;B|C)\\
&+&(s(s+1)/2)I(A;C|B)+I(B;C|A)\\
&+&sI(A;B|D)+sI(C;D),
\end{eqnarray*} which on our segment reduces to
\begin{eqnarray*}
0 &<=& -I(A;B) + (1/s)I(B;C|A)\\
&+& I(A;B|D)+I(C;D).
\end{eqnarray*}

At the top of the pyramid, all terms of this reduced inequality are zero except $I(A;B)$.
The inequality becomes $0\leq -1/4$, which is false.
But if the Ingleton equation is satisfied, then this
becomes
\begin{eqnarray*}
0 &<=& (1/s)I(B;C|A).
\end{eqnarray*}
and the value of the right side tends to zero with large values of $s$.
By linearity, the value of the right side of the reduced inequality at $x$
must become negative for sufficiently large $s$, so the information inequality
is false.  Therefore, the eleven defining faces of the segment cannot include both
$I(A;B|C)$ and $I(A;C|B)$.

A similar argument with $A$ and $B$ interchanged tells us that
the eleven faces cannot include both $I(A;B|C)$ and $I(B;C|A)$.

Next, suppose that the eleven faces of our segment includes the two faces
\begin{eqnarray*}
I(A;C|B)&=&0\\
I(B;C|A)&=&0\, .
\end{eqnarray*}
Iterating Rule 2, and starting with the Shannon Inequality
$0\leq I(A;B|C)$, we generate the infinite list of inequalities:
\begin{eqnarray*}
s I(A;B) &\leq& (s+1)I(A;B|C) +(s(s+1)/2)I(A;C|B))\\
&+&  (s(s+1)/2)I(B;C|A)) + s I(A;B|D)\\
 &+& s I(C;D)\, ,
\end{eqnarray*}
where $s$ is any any non-negative integer.
On our segment this reduces to
\begin{eqnarray*}
0  &\leq& -I(A;B) + (s+1)/s I(A;B|C)+I(A;B|D)\\
 &+& I(C;D)\, .
\end{eqnarray*}
At the top of the pyramid, this again becomes $0 \leq -1/4$, which is false.
If the Ingleton equation is satisfied, then this becomes
\begin{eqnarray*}
0  &\leq& (1/s) I(A;B|C)\, ,
\end{eqnarray*}
We conclude, as before, that the inequality becomes false at $x$ for sufficiently
large values of $s$.  Therefore, the eleven defining faces of our segment cannot include
both $I(A;C|B)$ and $I(B;C|A)$.

Summarizing the work so far, we see that the eleven faces cannot include any two of the three
faces
\begin{eqnarray*}
I(A;B|C)&=&0\\
I(A;C|B)&=&0\\
I(B;C|A)&=&0\, .
\end{eqnarray*}

A similar argument with $C$ and $D$ interchanged shows that the eleven faces cannot
include any two of the three faces
\begin{eqnarray*}
I(A;B|D)&=&0\\
I(A;D|B)&=&0\\
I(B;D|A)&=&0\, .
\end{eqnarray*}

But only three of the fourteen Shannon faces are excluded, so this is a contradiction.
\end{proof}

Future inequalities may be able to improve this further, perhaps showing that $I_4$ is the 4-core
of $\bar\Gamma^*_4$ in $\Gamma_4$, etc.  But it can't go beyond the 5-core
since the four-atom point lies on eight of the fourteen Shannon faces.


%
%



\section{Application of Non-Shannon Information Inequalities to Network Coding}\label{Vamos}

One application of non-Shannon information inequalities is in the field of network coding.
In network coding, a subset of network nodes called sources generate messages
and a subset of network nodes called receivers need to acquire the source messages.
The messages can propagate through the network but need not travel as in packet-switched routing.
Rather, network ``coding'' allows every out-edge of a node to be
an arbitrary function of the information carried on the node's in-edges.

In this section,
a \textit{network} is a finite, directed, acyclic multigraph
together with a finite set called the \textit{message set}.
An \textit{alphabet} is a finite set $\alphabet$ with at least two elements.
Two special subsets of nodes of the network are called \textit{sources}
and \textit{receiver}, respectively.
The sources generate messages and the receivers need to obtain
certain source messages
(namely those messages that the receiver \textit{demands}).

Let $k$ and $n$ be positive integers,
called the \textit{source dimension} and the \textit{edge capacity}, respectively.
For every network edge $e=(x,y)$
an \textit{edge function} $f_e$
puts an $n$-dimensional vector of alphabet symbols on $e$;
the vector's value is a function of the $n$-dimensional vectors carried on the in-edges of node $x$
and the $k$-dimensional source vectors produced at $x$.
Similarly,
for every receiver node $x$ and every network message $m$ demanded by $x$,
a \textit{decoding function}
$f_{x,m}$
produces a $k$-dimensional vector
as a function of the in-edge and source values at $x$.

Given an alphabet $\alphabet$,
a \textit{$(k,n)$ code}
for a network
is an assignment of edge functions and decoding functions
to the network's edges and receivers, respectively.

A $(k,n)$ network code is said to be a \textit{$(k,n)$ solution} if every receiver can recover
each of the messages it demands
(i.e. via it's decoding functions).

Special codes of interest include \textit{linear codes},
where the edge functions and decoding functions are linear,
and \textit{routing codes},
where the edge functions and decoding functions simply copy input components to output components.

If a network has a $(k,n)$ solution over some alphabet,
then we say the ratio $k/n$ is an \textit{achievable coding rate} for the network.

The
\textit{coding capacity} of a network
is
$$\sup \{ k/n : \exists\ (k,n) \mbox{ coding solution over } \mathcal{A} \}.$$
If $(k,n)$ coding solutions are restricted to linear codes
(over some finite field, for example) or routing codes,
then the capacity is called the \textit{linear capacity}
or \textit{routing capacity}, respectively.

Computing the network coding capacity of networks,
or at least accurately approximating or bounding the capacity
is a fundamental problem in network coding.
No general method for computing network capacity is presently known,
so attention is focused on bounding techniques.
Lower bounds are generally obtained by exhibiting specific solutions for a network.
Upper bounds can sometimes be obtained by assuming the source messages are random variables
and then using standard information inequalities.

Specifically, in \cite{Dougherty-Freiling-Zeger06-Matroids},
the Zhang-Yeung inequality
(our Theorem~\ref{thm:Zhang-Yeung})
was used to
derive an upper bound on the coding capacity of the \Vamos{} network
(see Figure~\ref{fig:Vamos-network}).
The \Vamos{} network is a
network constructed from the well-known \Vamos{} matroid
and is of interest
as it is presently the only known network for which non-Shannon
information inequalities have improved capacity bound calculations.

The capacity upper bound derived in \cite{Dougherty-Freiling-Zeger06-Matroids} was $10/11$
and was shown to be strictly smaller than the smallest possible upper bound
(i.e. $1$)
obtainable directly from Shannon-type information inequalities.
The \Vamos{} network thus illustrates a potential for improvement in capacity
calculations using non-Shannon-type inequalities over Shannon-type ineqaulities.

It was also shown in \cite{Dougherty-Freiling-Zeger06-Matroids}
that the linear coding
capacity of the \Vamos{} network over every finite field
is exactly equal to $5/6$,
which is presently the best known lower bound on the
(possibly non-linear) coding capacity of the \Vamos{} network.

The capacity bound calculation in \cite{Dougherty-Freiling-Zeger06-Matroids} was specific
to the particular form of the Zhang-Yeung inequality.
A more general strategy for improving capacity upper bounds
using other non-Shannon inequalities is given in what follows.
We present a systematic method
for using non-Shannon information inequalities to
obtain capacity bounds.
In particular,
the best known upper bound on the \Vamos{} network capacity is
improved from $10/11$ to $19/21$.
The exact coding capacity of this network remains an open problem,
lying somewhere in the interval $[5/6, 19/21]$.

\begin{figure}[h]
\begin{center}
\leavevmode
\hbox{\epsfxsize=0.40\textwidth\epsffile{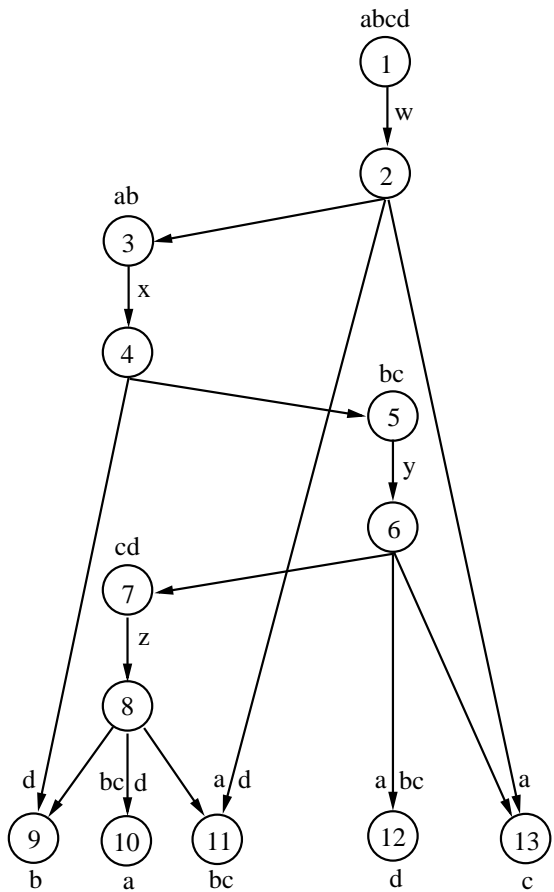}}
\end{center}
\caption{
The \Vamos{} network.
A message variable $a$, $b$, $c$, or $d$ labeled above a node
indicates an in-edge (not shown) from the source node (not shown) generating the  message.
Demand variables are labeled below the receivers $n_9$--$n_{13}$ demanding them.
The edges $e_{1,2}$, $e_{3,4}$, $e_{5,6}$, and $e_{7,8}$
are denoted by
$w$, $x$, $y$, and $z$,
respectively.
}
\label{fig:Vamos-network}
\end{figure}

We will now demonstrate one way to exploit the new non-Shannon inequalities given in
the present paper to obtain even tighter bounds on the coding
capacity of the \Vamos{} network.
The next theorem provides a tighter upper bound on the
\Vamos{} network coding capacity than previously known.

\begin{theorem}
The coding capacity of the \Vamos{} network is at most $19/21$.
\label{thm:Vamos-bound}
\end{theorem}

\begin{proof}
Consider a $(k,n)$ solution to the \Vamos{} network.
Assume that the network messages $a,b,c,d$ are
independent,
$k$-dimensional,
random vectors
with uniformly distributed components
and assume each edge in the network has capacity $n$.
Let $w,x,y,z$ denote the random variables carried by edges
$e_{1,2},e_{3,4},e_{5,6},e_{7,8}$, respectively.
Since the $n$-dimensional vector carried on any edge $e=(u,v)$ is a deterministic function
of the in-edges to node $u$ and sources generated at $u$,
the conditional entropy of the random vector on $e$,
given the vectors on the in-edges to $u$ and message vectors produced at $u$
must be zero.
Thus, we have
\begin{align}
H(w|a,b,c,d)           &= 0 & \Comment{$n_1$}     \label{eq:72} \\
H(x|a,b,w)             &= 0 & \Comment{$n_3$}     \label{eq:70} \\
H(y|b,c,x)             &= 0 & \Comment{$n_5$}     \label{eq:60} \\
H(z|c,d,y)             &= 0.& \Comment{$n_7$}     \label{eq:68}
\end{align}
Similarly,
since the $k$-dimensional vector decoded at a receiver $u$ is a deterministic function
of the in-edges to $u$ and sources generated at $u$,
the conditional entropy of the decoded vector at $u$,
given the vectors on the in-edges to $u$ and message vectors produced at $u$
must be zero.
Thus, we have
\begin{align}
H(b|d,x,z)             &= 0 & \Comment{$n_9$}     \label{eq:67}\\
H(a|b,c,d,z)           &= 0 & \Comment{$n_{10}$}  \label{eq:66} \\
H(b,c|a,d,w,z)         &= 0 & \Comment{$n_{11}$}  \label{eq:65} \\
H(d|a,b,c,y)           &= 0 & \Comment{$n_{12}$}  \label{eq:69} \\
H(c|a,w,y)             &= 0.& \Comment{$n_{13}$}  \label{eq:71}
\end{align}

In what follows we repeatedly make use of the following well known facts:
\begin{align}
I(R;S|T) &= H(R,T) + H(S,T)\notag\\
         &\ \  - H(R,S,T) - H(T)              \label{eq:id1}\\
H(R,S) &= H(R) + H(S|R)                       \label{eq:id2}\\
H(R,S) &\le H(R) + H(S)                       \label{eq:id3}\\
H(R)   &\le H(R,S).                           \label{eq:id4}
\end{align}
We thus obtain the following list of inequalities:
\begin{align}
&   I(c,y;b,x) \notag\\
             &\ \ = H(c,y) + H(b,x) - H(b,c,x,y)   & \Comment{\eqref{eq:id1}} \notag\\
             &\ \ = H(c,y) + H(b,x) - H(b,c,x)     & \Comment{\eqref{eq:60}} \notag\\
             &\ \ \ge H(c,y) - H(c)                & \Comment{\eqref{eq:id3}} \label{eq:101}
\end{align}
\begin{align}
&   I(c,y;b,x|d,z) \notag\\
             &\ \ = H(c,d,y,z) + H(b,d,x,z)\notag\\
             &\ \ \ \  - H(b,c,d,x,y,z) - H(d,z)   & \Comment{\eqref{eq:id1}} \notag\\
             &\ \ = H(c,d,y) + H(d,x,z)            & \Comment{\eqref{eq:67},\eqref{eq:68}} \notag\\
             &\ \ \ \  - H(a,b,c,d,x,y,z) - H(d,z) & \Comment{\eqref{eq:66},\eqref{eq:72},\eqref{eq:70}} \notag\\
             &\ \ \le H(c,y) + H(d) + H(x)         & \Comment{\eqref{eq:id3}} \notag\\
             &\ \ \ \  - H(a,b,c,d)                & \Comment{\eqref{eq:id4}}
\end{align}
\begin{align}
&   I(c,y;d,z|b,x) \notag\\
             &\ \ = H(b,c,x,y) + H(b,d,x,z)\notag\\
             &\ \ \ \  - H(b,c,d,x,y,z) - H(b,x)   & \Comment{\eqref{eq:id1}} \notag\\
             &\ \ = H(b,c,x) + H(d,x,z)            & \Comment{\eqref{eq:60},\eqref{eq:67}} \notag\\
             &\ \ \ \  - H(a,b,c,d,x,y,z) - H(b,x) & \Comment{\eqref{eq:66}} \notag\\
             &\ \ \le H(c) + H(d,z) + H(x)         & \Comment{\eqref{eq:id3}} \notag\\
             &\ \ \ \  - H(a,b,c,d)                & \Comment{\eqref{eq:id4}}
\end{align}
\begin{align}
&   I(b,x;d,z|c,y) \notag\\
             &\ \ = H(b,c,x,y) + H(c,d,y,z)\notag\\
             &\ \ \ \  - H(b,c,d,x,y,z) - H(c,y)   & \Comment{\eqref{eq:id1}} \notag\\
             &\ \ = H(b,c,x) + H(c,d,y)            & \Comment{\eqref{eq:60},\eqref{eq:68}} \notag\\
             &\ \ \ \  - H(a,b,c,d,x,y,z) - H(c,y) & \Comment{\eqref{eq:66}} \notag\\
             &\ \ \le H(b,x) + H(c) + H(d)         & \Comment{\eqref{eq:id3}} \notag\\
             &\ \ \ \  - H(a,b,c,d)                & \Comment{\eqref{eq:id4}}
\end{align}
\begin{align}
&   I(c,y;b,x|a,w) \notag\\
             &\ \ = H(a,c,w,y) + H(a,b,w,x)\notag\\
             &\ \ \ \  - H(a,b,c,w,x,y) - H(a,w)   & \Comment{\eqref{eq:id1}} \notag\\
             &\ \ = H(a,w,y) + H(a,b,w)            & \Comment{\eqref{eq:71},\eqref{eq:70}} \notag\\
             &\ \ \ \  - H(a,b,c,d,w,x,y) - H(a,w) & \Comment{\eqref{eq:69}} \notag\\
             &\ \ \le H(y) + H(a,w) + H(b)         & \Comment{\eqref{eq:id3}} \notag\\
             &\ \ \ \  - H(a,b,c,d)                & \Comment{\eqref{eq:id4}}
\end{align}
\begin{align}
&   I(c,y;a,w|b,x) \notag\\
             &\ \ = H(b,c,x,y) + H(a,w,b,x)\notag\\
             &\ \ \ \  - H(a,w,b,c,x,y) - H(b,x)   & \Comment{\eqref{eq:id1}} \notag\\
             &\ \ = H(b,c,x) + H(a,b,w)            & \Comment{\eqref{eq:60},\eqref{eq:70}}\notag\\
             &\ \ \ \  - H(a,b,c,d,w,x,y) - H(b,x) & \Comment{\eqref{eq:69}} \notag\\
             &\ \ \le H(c) + H(a,w) + H(b)         & \Comment{\eqref{eq:id3}} \notag\\
             &\ \ \ \  - H(a,b,c,d)                & \Comment{\eqref{eq:id4}}
\end{align}
\begin{align}
&   I(b,x;a,w|c,y) \notag\\
             &\ \ = H(b,c,x,y) + H(a,c,w,y)\notag\\
             &\ \ \ \  - H(a,b,c,w,x,y) - H(c,y)    & \Comment{\eqref{eq:id1}} \notag\\
             &\ \ = H(b,c,x) + H(a,w,y)             & \Comment{\eqref{eq:60},\eqref{eq:71}}\notag\\
             &\ \ \ \  - H(a,b,c,d,w,x,y) - H(c,y)  & \Comment{\eqref{eq:66}} \notag\\
             &\ \ \le H(b,x) + H(c) + H(a,w) + H(y) & \Comment{\eqref{eq:id3}} \notag\\
             &\ \ \ \  - H(a,b,c,d) - H(c,y)        & \Comment{\eqref{eq:id4}}
\end{align}
\begin{align}
&   I(d,z;a,w) \notag\\
             &\ \ = H(d,z) + H(a,w) - H(a,d,w,z)     & \Comment{\eqref{eq:id1}} \notag\\
             &\ \ = H(d,z) + H(a,w)\notag\\
             &\ \ \ \  - H(a,b,c,d,w,z) & \Comment{\eqref{eq:65}} \notag\\
             &\ \ = H(d,z) + H(a,w) - H(a,b,c,d)     & \Comment{\eqref{eq:72}--\eqref{eq:68}} \notag\\
\end{align}
\begin{align}
&   I(d,z;a,w|c,y) \notag\\
             &\ \ = H(c,d,y,z) + H(a,c,w,y)\notag\\
             &\ \ \ \  - H(a,c,d,w,y,z) - H(c,y)    & \Comment{\eqref{eq:id1}} \notag\\
             &\ \ = H(c,d,y) + H(a,w,y)             & \Comment{\eqref{eq:68},\eqref{eq:71}}\notag\\
             &\ \ \ \  - H(a,b,c,d,w,y,z) - H(c,y)  & \Comment{\eqref{eq:65}}\notag\\
             &\ \ \le H(d) + H(a,w) + H(y)          & \Comment{\eqref{eq:id3}} \notag\\
             &\ \ \ \  - H(a,b,c,d)                 & \Comment{\eqref{eq:id4}}
\end{align}
\begin{align}
&   I(d,z;a,w|b,x) \notag\\
             &\ \ = H(b,d,x,z) + H(a,b,w,x)\notag\\
             &\ \ \ \  - H(a,b,d,w,x,z) - H(b,x)    & \Comment{\eqref{eq:id1}} \notag\\
             &\ \ = H(d,x,z) + H(a,b,w)             & \Comment{\eqref{eq:67},\eqref{eq:70}}\notag\\
             &\ \ \ \  - H(a,b,c,d,w,x,z) - H(b,x)  & \Comment{\eqref{eq:65}}\notag\\
             &\ \ \le H(d,z) {+} H(x) {+} H(a,w) {+} H(b) & \Comment{\eqref{eq:id3}} \notag\\
             &\ \ \ \  - H(a,b,c,d) - H(b,x).       & \Comment{\eqref{eq:id4}}  \label{eq:102}
\end{align}

Now, suppose that $A$, $B$, $C$, and $D$ are  random variables and
we have an information inequality of the form
\begin{align}
&a_1 I(A;B) \notag\\
&\ \ \le a_2 I(A;B|C)
+ a_3 I(A;C|B)
+ a_4 I(B;C|A)\notag\\
&\ \ + a_5 I(A;B|D)
+ a_6 I(A;D|B)
+ a_7 I(B;D|A)\notag\\
&\ \ + a_8 I(C;D)
+ a_9 I(C;D|A)
+ a_{10} I(C;D|B)
\label{eq:105}
\end{align}
for some $a_i \ge 0$, for all $i$.
If we set
\begin{align*}
    A &= (c,y)\\
    B &= (b,x)\\
    C &= (d,z)\\
    D &= (a,w)
\end{align*}
then
\eqref{eq:101}--\eqref{eq:102}
give the inequality
\begin{align*}
&a_1 ( H(c,y) - H(c) ) \\
&\le
a_2 ( H(c,y) + H(d) + H(x) - H(a,b,c,d) ) \\
&\ \ + a_3 ( H(c) + H(d,z) + H(x) - H(a,b,c,d) ) \\
&\ \ + a_4 ( H(b,x) + H(c) + H(d) - H(a,b,c,d) ) \\
&\ \ + a_5 ( H(y) + H(a,w) + H(b) - H(a,b,c,d) ) \\
&\ \ + a_6 ( H(c) + H(a,w) + H(b) - H(a,b,c,d) ) \\
&\ \ + a_7 ( - H(a,b,c,d) - H(c,y) ) \\
&\ \ + a_8 ( H(d,z) + H(a,w) - H(a,b,c,d) ) \\
&\ \ + a_9 ( H(d) + H(a,w) + H(y) - H(a,b,c,d) ) \\
&\ \ + a_{10} ( H(b) - H(a,b,c,d) - H(b,x) ) \\
\end{align*}
or, equivalently, after collecting terms,
\begin{align}
& (a_2+a_3+a_4) H(a) \notag\\
&\ \ + (a_2+a_3+a_8+a_9+a_{10}) H(b) \notag\\
&\ \ + (a_5+a_7+a_8+a_9+a_{10}) H(c) \notag\\
&\ \ + (a_5+a_6+a_7) H(d) \notag\\
&\ \ + (a_2-a_1-a_7) I(c;y) \notag\\
&\ \ + (a_4+a_7-a_{10}) I(b;x)\notag\\
&\le
(a_5+a_6+a_7+a_8+a_9+a_{10}) H(w) \notag\\
&\ \ + (a_2+a_3+a_4+a_7) H(x) \notag\\
&\ \ + (-a_1+a_2+a_5+a_9) H(y) \notag\\
&\ \ + (a_3+a_8+a_{10}) H(z).
 \label{eq:103}
\end{align}

If the inequalities
\begin{align}
&a_2 \ge a_1 + a_7\notag\\
&a_4 + a_7 \ge a_{10}
\label{eq:106}
\end{align}
are satisfied,
then the inequality \eqref{eq:103} directly leads to a
\Vamos{} coding capacity region bound,
by neglecting the (nonnegative) terms involving $I(c;y)$ and $I(b;x)$.
Specifically, in this case,
by substituting
\begin{align*}
H(a)&=H(b)=H(c)=H(d)=k\\
H(w)&=H(x)=H(y)=H(z)=n
\end{align*}
into \eqref{eq:103},
we obtain
\begin{align}
&k (
2a_2
+ 2a_3
+ a_4
+ 2a_5
+ a_6\notag\\
&\ \
+ 2a_7
+ 2a_8
+ 2a_9
+ 2a_{10}
)\notag\\
&\le
n(
-a_1
+ 2a_2
+ 2a_3
+ a_4
+ 2a_5\notag\\
&\ \
+ a_6
+ 2a_7
+ 2a_8
+ 2a_9
+ 2a_{10}
).
 \label{eq:104}
\end{align}

One of our new non-Shannon inequalities is
\begin{align}
&2I(A;B) \le 3I(A;B|C) + 2I(A;C|B)\notag\\
&\ \ + 3I(B;C|A) + 2I(A;B|D) + 2I(C;D)
\label{eq:non-Shannon1}
\end{align}
which is equivalent to inequality \eqref{ineq3.10}
(by exchanging the roles of $A$ and $B$ in in Theorem~\ref{thm:many-inequalities}).
Using \eqref{eq:105} and \eqref{eq:non-Shannon1},
we see that both of the inequalities in \eqref{eq:106} hold
and therefore we substitute
\begin{align*}
a_1 &= a_3 = a_5 = a_8 = 2\\
a_2 &= a_4 = 3\\
a_6 &= a_7 = a_9 = a_{10} = 0
\end{align*}
into \eqref{eq:104} and obtain
\begin{align}
k/n \le 19/21.
\label{eq:107}
\end{align}
Thus, the coding capacity of the \Vamos{} network can be at most $19/21$.
\end{proof}

We note that if one or both of the inequalities in
\eqref{eq:106} are not satisfied,
we may be able to combine \eqref{eq:103}
with another such inequality
(having a positive coefficient for $I(c;y)$ and/or $I(b;x)$)
to eliminate the extra term(s).
Some specific inequalities that are useful for this are given in the following lemma.

\begin{lemma}
\begin{align*}
&(i)   &&H(a) + I(c;y)        \le H(y)\\
&(ii)  &&H(a) + I(b;x)        \le H(x)\\
&(iii) &&H(a) + H(b) + I(c;y) \le H(x) + H(y)\\
&(iv)  &&H(b) + H(c)          \le H(w) + H(y)\\
&(v)   &&H(c) + H(d)          \le H(w) + H(y)\\
&(vi)  &&H(d) + I(c;y)        \le H(y)\\
&(vii) &&H(d) + I(b;x)        \le H(x)
\end{align*}
\end{lemma}

\begin{proof} \ \\
\begin{itemize}

\item [(i)]
\begin{align}
& H(y) - H(a) - I(c;y)\notag\\
&\ \ = H(c,y) - H(a) - H(c)                    & \Comment{\eqref{eq:id1}} \notag\\
&\ \ \ge H(b,c,d,y) - H(a,b,c,d)               & \Comment{\eqref{eq:id3}} \notag\\
&\ \ \ge H(b,c,d,z) - H(a,b,c,d)               & \Comment{\eqref{eq:68}} \notag\\
&\ \ \ge H(a,b,c,d) - H(a,b,c,d)               & \Comment{\eqref{eq:66}} \notag\\
&\ \ = 0                                       \label{eq:108}
\end{align}

\item [(ii)]
\begin{align*}
& H(x) - H(a) - I(b;x)\notag\\
&\ \ = H(b,x) - H(a) - H(b)                    & \Comment{\eqref{eq:id1}} \notag\\
&\ \ \ge H(b,c,d,x) - H(a,b,c,d)               & \Comment{\eqref{eq:id3}} \notag\\
&\ \ \ge H(b,c,d,y) - H(a,b,c,d)               & \Comment{\eqref{eq:60}} \notag\\
&\ \ \ge 0                                     & \Comment{\eqref{eq:108}} \notag
\end{align*}

\item [(iii)]
\begin{align*}
& H(x) + H(y) - H(a) - H(b)\notag\\
&\ \ \ \  - I(c;y)\notag\\
&\ \ = H(c,y) + H(d) + H(x)\notag\\
&\ \ \ \  - H(a,b,c,d)                        & \Comment{\eqref{eq:id1}} \notag\\
&\ \ \ge H(c,d,x,y) - H(a,b,c,d)              & \Comment{\eqref{eq:id3}} \notag\\
&\ \ \ge H(c,d,x,z) - H(a,b,c,d)              & \Comment{\eqref{eq:68}} \notag\\
&\ \ \ge H(a,b,c,d) - H(a,b,c,d)              & \Comment{\eqref{eq:67},\eqref{eq:66}} \notag\\
&\ \ =  0                                     \notag
\end{align*}

\item [(iv)]
\begin{align*}
& H(w) - H(y) - H(b) - H(c)\notag\\
&\ \ \ge H(a,d,w,y) - H(a,b,c,d)              & \Comment{\eqref{eq:id3}} \notag\\
&\ \ \ge H(a,d,w,z) - H(a,b,c,d)              & \Comment{\eqref{eq:71},\eqref{eq:68}} \notag\\
&\ \ \ge H(a,b,c,d) - H(a,b,c,d)              & \Comment{\eqref{eq:65}} \notag\\
&\ \ =  0                                     \notag
\end{align*}

\item [(v)]
\begin{align*}
& H(w) - H(y) - H(c) - H(d)\notag\\
&\ \ \ge H(a,b,w,y) - H(a,b,c,d)              & \Comment{\eqref{eq:id3}} \notag\\
&\ \ \ge H(a,b,c,d) - H(a,b,c,d)              & \Comment{\eqref{eq:71},\eqref{eq:69}} \notag\\
&\ \ =  0                                     \notag
\end{align*}

\item [(vi)]
\begin{align}
& H(y) - H(d) - I(c;y)\notag\\
&\ \ = H(c,y) - H(c) - H(d)                  & \Comment{\eqref{eq:id1}} \notag\\
&\ \ \ge H(a,b,c,y) - H(a,b,c,d)             & \Comment{\eqref{eq:id3}} \notag\\
&\ \ \ge H(a,b,c,d) - H(a,b,c,d)             & \Comment{\eqref{eq:69}} \notag\\
&\ \ =  0                                    \label{eq:109}
\end{align}

\item [(vii)]
\begin{align*}
& H(x) - H(d) - I(b;x)\notag\\
&\ \ = H(b,x) - H(b) - H(d)                 & \Comment{\eqref{eq:id1}} \notag\\
&\ \ \ge H(a,b,c,x) - H(a,b,c,d)            & \Comment{\eqref{eq:id3}} \notag\\
&\ \ \ge H(a,b,c,y) - H(a,b,c,d)            & \Comment{\eqref{eq:60}} \notag\\
&\ \ \ge  0.                                & \Comment{\eqref{eq:109}} \notag
\end{align*}

\end{itemize}

\end{proof}

We have exhaustively computed network coding capacity bounds
for the \Vamos{} network using all of the new non-Shannon information
inequalities we have discovered and
the bound in \eqref{eq:107} is the tightest (i.e. smallest)
among those we checked.



\section{Conclusion}

In this article, we have presented a lengthy list of non-Shannon type inequalities.
This is a continuation of work started by Zhang and Yeung and roughly follows their approach.
We have seen that many of the inequalities, including the original Zhang-Yeung inequality have
been superseded by stronger ones.  Almost certainly, many more of the inequalities that have been presented
here will also be superseded by future work in this area.  To our knowledge, no one has ever given a non-Shannon
inequality and proved that it will not be superseded by others.  This would certainly be an interesting challenge for
future research.

Although each inequality found gives more information about the shape of entropy space, and each one is at least theoretically
useful in a network information flow problem, it seems unlikely that endlessly producing new inequalities will be fruitful in itself.
These inequalities are presented with the hope that it may be possible to study the list as a whole, in order to gain further insights
that will enable doing more than just endlessly extending the list.

Along these lines, perhaps the most striking feature of list of the inequalities listed is their special form.
Indeed every inequality produced has the form
\begin{eqnarray*}
aI(A;B) &\leq& bI(A;B|C) + cI(A;C|B) + dI(B;C|A) \\
      &+&  eI(A;B|D) + fI(A;D|B)+ gI(B;D|A)\\
      &+& hI(C;D) + iI(C;D|A)
\end{eqnarray*}
and vary from each other only by the choice of the coefficients $a,b,c,d,e,f,g,h,i$. There seems to be no obvious reason for this form.
Indeed, there were a few inequalities that were found and were not of this form.  But each of these non-conforming inequalities was
superseded by other inequalities or was reduced to the special form under tightening.
Explaining this special form might be a first step in gaining a deeper understanding of the inequalities.

Another direction to pursue is to find patterns among the inequalities and/or their proofs.
Identifying such patterns may allow us to continue the patterns indefinitely.  Along these lines,
we have identified several ``rules" that allow us to generate new inequalities from old ones.
Application of these rules quickly yields infinite families of inequalities similar to those
found by Matus \cite{Matus}.  Families of inequalities have already given valuable insight into
the structure of entropy space and were instrumental in Matus' proof that entropy space is not a polytope (or polytopal cone).
To our knowledge, no one has yet been able to show whether or not entropy space is curved.
This would certainly be a valuable contribution toward understanding this space.

Finally, the techniques used in this paper seem limited and do not seem like they will be strong enough
to settle the four atom conjecture.  We believe therefore, that some new techniques for
finding non-Shannon inequalities will be vital for gaining a more complete understanding of entropy space.


\renewcommand{\baselinestretch}{1.0}

\end{document}

%% file: Rules.txt
{\bf Rule [2]}
Given:
\begin{eqnarray*}&&     aI(A;B)   \\
	&\leq&   (b+h)I(A;B|C)+cI(A;C|B)+dI(B;C|A)  \\
	&+& aI(A;B|D)+fI(A;D|B)+gI(B;D|A)  \\
	&+& hI(C;D)+iI(C;D|A)+jI(C;D|B)
\end{eqnarray*}
Get:
\begin{eqnarray*}
&&    (a+b)I(A;B)   \\
	&\leq&   (2b+h)I(A;B|C)+(b+c+h+j)I(A;C|B)  \\
	&+& (b+d+h+i)I(B;C|A)+(a+b)I(A;B|D)  \\
	&+& fI(A;D|B)+gI(B;D|A)+(b+h)I(C;D)  \\
	&+& iI(C;D|A)+jI(C;D|B)
\end{eqnarray*}
 Using:    $R$ is copy of $C$ over $AB$    
\newline Substitution: $AR$ $BR$ $CR$ $D$

{\bf Rule [3]}
Given:
\begin{eqnarray*}&&     aI(A;B)   \\
	&\leq&   aI(A;B|C)+cI(A;C|B)  \\
	&+& (d+h+z)I(B;C|A)+(a+e+h+z)I(A;B|D)  \\
	&+& fI(A;D|B)+gI(B;D|A)+(a+h)I(C;D)  \\
	&+& iI(C;D|A)+jI(C;D|B)
\end{eqnarray*}
Get:
\begin{eqnarray*}
&&    (a+z)I(A;B)   \\
	&\leq&   (2a+d+h+i+2z)I(A;B|C)  \\
	&+& (a+c+e+f+g+2h+z)I(A;C|B)  \\
	&+& (d+h+z)I(B;C|A)+(a+e+h+z)I(A;B|D)  \\
	&+& fI(A;D|B)+gI(B;D|A)+(a+h+z)I(C;D)  \\
	&+& iI(C;D|A)+jI(C;D|B)
\end{eqnarray*}
 Using:    $R$ is copy of $A$ over $BC$    
\newline Substitution: $AR$ $BR$ $C$ $DR$

{\bf Rule [4]}
Given:
\begin{eqnarray*}&&     aI(A;B)   \\
	&\leq&   aI(A;B|C)+cI(A;C|B)+dI(B;C|A)  \\
	&+& (e+h)I(A;B|D)+fI(A;D|B)+gI(B;D|A)  \\
	&+& hI(C;D)+iI(C;D|A)+jI(C;D|B)
\end{eqnarray*}
Get:
\begin{eqnarray*}
&&    (a+e)I(A;B)   \\
	&\leq&   (a+2e)I(A;B|C)+(c+e+h+j)I(A;C|B)  \\
	&+& (d+e+h+i)I(B;C|A)+(e+h)I(A;B|D)  \\
	&+& fI(A;D|B)+gI(B;D|A)+(e+h)I(C;D)  \\
	&+& iI(C;D|A)+jI(C;D|B)
\end{eqnarray*}
 Using:    $R$ is copy of $C$ over $AB$    
\newline Substitution: $AR$ $BR$ $C$ $DR$

{\bf Rule [5]}
Given:
\begin{eqnarray*}&&     aI(A;B)   \\
	&\leq&   aI(A;B|C)+cI(A;C|B)  \\
	&+& (d+h+z)I(B;C|A)  \\
	&+& (a+h+z+z')I(A;B|D)+fI(A;D|B)  \\
	&+& gI(B;D|A)+(a+h)I(C;D)+iI(C;D|A)  \\
	&+& jI(C;D|B)
\end{eqnarray*}
and
\begin{eqnarray*}&&     a'I(A;B)   \\
	&\leq&   b'I(A;B|C)+c'I(A;C|B)+z'I(B;C|A)  \\
	&+& a'I(A;B|D)+f'I(A;D|B)+g'I(B;D|A)  \\
	&+& h'I(C;D)+i'I(C;D|A)+j'I(C;D|B)
\end{eqnarray*}
Get:
\begin{eqnarray*}
&&    (a+a'+z+z')I(A;B)   \\
	&\leq&  (2a+d+h+i+b'+h'+i'+2z+2z')  \\
	&& I(A;B|C)  \\
	&+& (a+c+f+g+2h+a'+c'+f'+g'+z+z')  \\
	&& I(A;C|B)+(d+h+z+z')I(B;C|A)  \\
	&+& (a+h+a'+z+z')I(A;B|D)  \\
	&+& (f+f')I(A;D|B)+(g+g')I(B;D|A)  \\
	&+& (a+h+h'+z+z')I(C;D)  \\
	&+& (i+i')I(C;D|A)+(j+j')I(C;D|B)
\end{eqnarray*}
 Using:    $R$ is copy of $A$ over $BC$    
\newline Substitutions: $AR$ $BR$ $C$ DR;                   $AR$ $BR$ $CR$ $D$

{\bf Rule [6]}
Given:
\begin{eqnarray*}&&     aI(A;B)   \\
	&\leq&   aI(A;B|C)+cI(A;C|B)+dI(B;C|A)  \\
	&+& (h+z)I(A;B|D)+fI(A;D|B)+gI(B;D|A)  \\
	&+& hI(C;D)+iI(C;D|A)+jI(C;D|B)
\end{eqnarray*}
and
\begin{eqnarray*}&&     a'I(A;B)   \\
	&\leq&   (h'+z)I(A;B|C)+c'I(A;C|B)  \\
	&+& d'I(B;C|A)+a'I(A;B|D)+f'I(A;D|B)  \\
	&+& g'I(B;D|A)+h'I(C;D)+i'I(C;D|A)  \\
	&+& j'I(C;D|B)
\end{eqnarray*}
Get:
\begin{eqnarray*}
&&    (a+a'+z)I(A;B)   \\
	&\leq&   (a+h'+2z)I(A;B|C)  \\
	&+& (c+h+j+c'+h'+j'+z)I(A;C|B)  \\
	&+& (d+h+i+d'+h'+i'+z)I(B;C|A)  \\
	&+& (h+a'+z)I(A;B|D)+(f+f')I(A;D|B)  \\
	&+& (g+g')I(B;D|A)+(h+h'+z)I(C;D)  \\
	&+& (i+i')I(C;D|A)+(j+j')I(C;D|B)
\end{eqnarray*}
 Using:    $R$ is copy of $C$ over $AB$    
\newline Substitutions: $AR$ $BR$ $C$ DR;                   $AR$ $BR$ $CR$ $D$

{\bf Rule [7]}
Given:
\begin{eqnarray*}&&     aI(A;B)   \\
	&\leq&   bI(A;B|C)+cI(A;C|B)  \\
	&+& (d+j+z)I(B;C|A)+eI(A;B|D)  \\
	&+& fI(A;D|B)+(g+x+z)I(B;D|A)+hI(C;D)  \\
	&+& iI(C;D|A)+(j+x+z)I(C;D|B)
\end{eqnarray*}
Get:
\begin{eqnarray*}
&&    (a+z)I(A;B)   \\
	&\leq&   (a+b+c+f+j+x+2z)I(A;B|C)  \\
	&+& (-a+b+c+e+j+x+z)I(A;C|B)  \\
	&+& (d+j+z)I(B;C|A)+(e+z)I(A;B|D)  \\
	&+& fI(A;D|B)+(g+x)I(B;D|A)  \\
	&+& (h+j+z)I(C;D)+iI(C;D|A)+xI(C;D|B)
\end{eqnarray*}
 Using:    $R$ is copy of $C$ over $AB$    
\newline Substitution: $A$ $R$ $C$ $D$

{\bf Rule [8]}
Given:
\begin{eqnarray*}&&     aI(A;B)   \\
	&\leq&   (h+a'+c'+z)I(A;B|C)+cI(A;C|B)  \\
	&+& dI(B;C|A)+aI(A;B|D)+fI(A;D|B)  \\
	&+& gI(B;D|A)+hI(C;D)+iI(C;D|A)  \\
	&+& jI(C;D|B)
\end{eqnarray*}
and
\begin{eqnarray*}&&     a'I(A;B)   \\
	&\leq&   (a'+b')I(A;B|C)+c'I(A;C|B)  \\
	&+& (d+h+d'+j'+z)I(B;C|A)+e'I(A;B|D)  \\
	&+& f'I(A;D|B)+(a+g+z)I(B;D|A)  \\
	&+& h'I(C;D)+i'I(C;D|A)  \\
	&+& (j'+z)I(C;D|B)
\end{eqnarray*}
Get:
\begin{eqnarray*}
&&    (a+a'+z)I(A;B)   \\
	&\leq&   (h+2a'+b'+c'+f'+j'+2z)I(A;B|C)  \\
	&+& (c+h+j+b'+c'+e'+j'+z)I(A;C|B)  \\
	&+& (d+h+i+d'+j'+z)I(B;C|A)  \\
	&+& (a+e'+z)I(A;B|D)+(f+f')I(A;D|B)  \\
	&+& (g)I(B;D|A)+(h+h'+j'+z)I(C;D)  \\
	&+& (i+i')I(C;D|A)+(j)I(C;D|B)
\end{eqnarray*}
 Using:    $R$ is copy of $C$ over $AB$    
\newline Substitutions: $AR$ $BR$ $CR$ D;                   $A$ $R$ $C$ $D$

{\bf Rule [9]}
Given:
\begin{eqnarray*}&&     aI(A;B)   \\
	&\leq&   (a+b+f+g+z)I(A;B|C)+cI(A;C|B)  \\
	&+& dI(B;C|A)+eI(A;B|D)+fI(A;D|B)  \\
	&+& (g+g_2+z)I(B;D|A)+hI(C;D)  \\
	&+& iI(C;D|A)+(j+z)I(C;D|B)
\end{eqnarray*}
Get:
\begin{eqnarray*}
&&    (a+z)I(A;B)   \\
	&\leq&   (a+b+f+g+2z)I(A;B|C)  \\
	&+& (c+d+e+f+g+g_2+z)I(A;C|B)  \\
	&+& (a+c+d+f+g+j+z)I(B;C|A)  \\
	&+& (e+z)I(A;B|D)+(f)I(A;D|B)  \\
	&+& (g_2)I(B;D|A)+(h+z)I(C;D)  \\
	&+& (i)I(C;D|A)+(j)I(C;D|B)
\end{eqnarray*}
 Using:    $R$ is copy of $A$ over $BC$    
\newline Substitution: $A$ $R$ $C$ $D$

{\bf Rule [10]}
Given:
\begin{eqnarray*}&&     aI(A;B)   \\
	&\leq&   bI(A;B|C)+(c+x+z)I(A;C|B)  \\
	&+& dI(B;C|A)+eI(A;B|D)+fI(A;D|B)  \\
	&+& gI(B;D|A)+hI(C;D)+(i+x+z)I(C;D|A)  \\
	&+& (j+z)I(C;D|B)
\end{eqnarray*}
Get:
\begin{eqnarray*}
&&    (a+z)I(A;B)   \\
	&\leq&   (b+d+h+i+x+2z)I(A;B|C)  \\
	&+& (c+j_1+x+z)I(A;C|B)  \\
	&+& (b+d+x+z)I(B;C|A)+(e+z)I(A;B|D)  \\
	&+& (f)I(A;D|B)+(g)I(B;D|A)  \\
	&+& (h+z)I(C;D)+(i)I(C;D|A)  \\
	&+& (j_2)I(C;D|B)
\end{eqnarray*}
 Using:    $R$ is copy of $A$ over $BC$    
\newline Substitution: $A$ $B$ $R$ $D$

{\bf Rule [11]}
Given:
\begin{eqnarray*}&&     aI(A;B)   \\
	&\leq&   bI(A;B|C)+cI(A;C|B)+(d+x)I(B;C|A)  \\
	&+& (a+d+e)I(A;B|D)+fI(A;D|B)  \\
	&+& gI(B;D|A)+aI(C;D)+iI(C;D|A)  \\
	&+& jI(C;D|B)
\end{eqnarray*}
Get:
\begin{eqnarray*}
&&    (a+d)I(A;B)   \\
	&\leq&   (a+b+2d+i+x)I(A;B|C)  \\
	&+& (a+c+d+e+f+g)I(A;C|B)  \\
	&+& (d+x)I(B;C|A)+(a+d+e)I(A;B|D)  \\
	&+& fI(A;D|B)+gI(B;D|A)+(a+d)I(C;D)  \\
	&+& iI(C;D|A)+jI(C;D|B)
\end{eqnarray*}
 Using:    $R$ is copy of $A$ over $BC$    
\newline Substitution: $AR$ $BR$ $CR$ $DR$

{\bf Rule [12]}
Given:
\begin{eqnarray*}&&     aI(A;B)   \\
	&\leq&   (a+z)I(A;B|C)+cI(A;C|B)+dI(B;C|A)  \\
	&+& (a+x+z)I(A;B|D)+fI(A;D|B)  \\
	&+& gI(B;D|A)+aI(C;D)+iI(C;D|A)  \\
	&+& jI(C;D|B)
\end{eqnarray*}
Get:
\begin{eqnarray*}
&&    (a+z)I(A;B)   \\
	&\leq&   (a+2z)I(A;B|C)+(a+c+j+z)I(A;C|B)  \\
	&+& (a+d+i+x+z)I(B;C|A)  \\
	&+& (a+x+z)I(A;B|D)+fI(A;D|B)  \\
	&+& gI(B;D|A)+(a+z)I(C;D)+iI(C;D|A)  \\
	&+& jI(C;D|B)
\end{eqnarray*}
 Using:    $R$ is copy of $C$ over $AB$    
\newline Substitution: $AR$ $BR$ $CR$ $DR$

{\bf Rule [13]}
Given:
\begin{eqnarray*}&&     aI(A;B)   \\
	&\leq&   (a+y+z)I(A;B|C)+cI(A;C|B)  \\
	&+& dI(B;C|A)+(a+z)I(A;B|D)+fI(A;D|B)  \\
	&+& gI(B;D|A)+aI(C;D)+iI(C;D|A)  \\
	&+& jI(C;D|B)
\end{eqnarray*}
and
\begin{eqnarray*}&&     a'I(A;B)   \\
	&\leq&   b'I(A;B|C)+c'I(A;C|B)+d'I(B;C|A)  \\
	&+& (h'+y)I(A;B|D)+f'I(A;D|B)  \\
	&+& g'I(B;D|A)+h'I(C;D)+i'I(C;D|A)  \\
	&+& j'I(C;D|B)
\end{eqnarray*}
Get:
\begin{eqnarray*}
&&    (a+a'+y+z)I(A;B)   \\
	&\leq&   (a+b'+i'+2y+2z)I(A;B|C)  \\
	&+& (a+c+j+c'+f'+h'+i'+j'+y+z)  \\
	&& I(A;C|B)  \\
	&+& (a+d+i+d'+g'+h'+i'+y+z)I(B;C|A)  \\
	&+& (a+h'+y+z)I(A;B|D)+(f+f')I(A;D|B)  \\
	&+& (g+g')I(B;D|A)+(a+h'+y+z)I(C;D)  \\
	&+& (i+i')I(C;D|A)+(j+j')I(C;D|B)
\end{eqnarray*}
 Using:    $R$ is copy of $C$ over $AB$    
\newline Substitutions: $AR$ $BR$ $CR$ DR;                   $A$ $B$ $C$ $DR$

{\bf Rule [14]}
Given:
\begin{eqnarray*}&&     aI(A;B)   \\
	&\leq&   (a+b'+c'+z)I(A;B|C)+cI(A;C|B)  \\
	&+& dI(B;C|A)+(a+e'+f'+z)I(A;B|D)  \\
	&+& fI(A;D|B)+gI(B;D|A)+aI(C;D)  \\
	&+& iI(C;D|A)+jI(C;D|B)
\end{eqnarray*}
and
\begin{eqnarray*}&&     a'I(A;B)   \\
	&\leq&   b'I(A;B|C)+c'I(A;C|B)  \\
	&+& (a+d+i+j'+z)I(B;C|A)+e'I(A;B|D)  \\
	&+& f'I(A;D|B)+(a+g+i+z)I(B;D|A)  \\
	&+& h'I(C;D)+i'I(C;D|A)  \\
	&+& (j'+z)I(C;D|B)
\end{eqnarray*}
Get:
\begin{eqnarray*}
&&    (a+a'+z)I(A;B)   \\
	&\leq&   (a+a'+b'+c'+f'+j'+2z)I(A;B|C)  \\
	&+& (a+c+j-a'+b'+c'+e'+j'+z)I(A;C|B)  \\
	&+& (a+d+i+j'+z)I(B;C|A)  \\
	&+& (a+e'+z)I(A;B|D)+(f+f')I(A;D|B)  \\
	&+& gI(B;D|A)+(a+h'+j'+z)I(C;D)  \\
	&+& (i+i')I(C;D|A)+jI(C;D|B)
\end{eqnarray*}
 Using:    $R$ is copy of $C$ over $AB$    
\newline Substitutions: $AR$ $BR$ $CR$ DR;                   $A$ $R$ $C$ $D$

{\bf Rule [15]}
Given:
\begin{eqnarray*}&&     aI(A;B)   \\
	&\leq&   aI(A;B|C)+cI(A;C|B)+dI(B;C|A)  \\
	&+& (h+a')I(A;B|D)+fI(A;D|B)  \\
	&+& gI(B;D|A)+hI(C;D)+iI(C;D|A)  \\
	&+& jI(C;D|B)
\end{eqnarray*}
and
\begin{eqnarray*}&&     a'I(A;B)   \\
	&\leq&   b'I(A;B|C)+c'I(A;C|B)  \\
	&+& (a+d')I(B;C|A)+(a'+e')I(A;B|D)  \\
	&+& f'I(A;D|B)+(g+h+g')I(B;D|A)  \\
	&+& h'I(C;D)+i'I(C;D|A)+j'I(C;D|B)
\end{eqnarray*}
Get:
\begin{eqnarray*}
&&    (a+a')I(A;B)   \\
	&\leq&   (a+a'+b'+c'+f'+j')I(A;B|C)  \\
	&+& (c+h+j+b'+c'+e'+j')I(A;C|B)  \\
	&+& (d+h+i+d'+g'_1+j')I(B;C|A)  \\
	&+& (h+a'+e')I(A;B|D)+(f+f')I(A;D|B)  \\
	&+& (g+g'_2)I(B;D|A)+(h+h'+j')I(C;D)  \\
	&+& (i+i')I(C;D|A)+jI(C;D|B)
\end{eqnarray*}
 Using:    $R$ is copy of $C$ over $AB$    
\newline Substitutions: $AR$ $BR$ $C$ DR;                   $A$ $R$ $C$ $D$

{\bf Rule [16]}
Given:
\begin{eqnarray*}&&     aI(A;B)   \\
	&\leq&   aI(A;B|C)+cI(A;C|B)  \\
	&+& (-a+a'+c'+d'+f'+x)I(B;C|A)  \\
	&+& (e'+f'+x)I(A;B|D)+fI(A;D|B)  \\
	&+& gI(B;D|A)+(x-w)I(C;D)+iI(C;D|A)  \\
	&+& jI(C;D|B)
\end{eqnarray*}
and
\begin{eqnarray*}&&     a'I(A;B)   \\
	&\leq&   (a'+b'+f'+x)I(A;B|C)+c'I(A;C|B)  \\
	&+& d'I(B;C|A)+e'I(A;B|D)+f'I(A;D|B)  \\
	&+& (g+x)I(B;D|A)+h'I(C;D)+i'I(C;D|A)  \\
	&+& (j'+w)I(C;D|B)
\end{eqnarray*}
Get:
\begin{eqnarray*}
&&    (a+a'+w)I(A;B)   \\
	&\leq&   (a+i+a'+b'+f'+w+x)I(A;B|C)  \\
	&+& (-a+c+f+g+c'+d'+e'+f'-w+2x)  \\
	&& I(A;C|B)  \\
	&+& (-a+a'+c'+d'+f'+j'+x)I(B;C|A)  \\
	&+& (e'+x)I(A;B|D)+(f+f')I(A;D|B)  \\
	&+& (g)I(B;D|A)+(h'+x)I(C;D)  \\
	&+& (i+i')I(C;D|A)+(j+j')I(C;D|B)
\end{eqnarray*}
 Using:    $R$ is copy of $A$ over $BC$    
\newline Substitutions: $AR$ $BR$ $C$ DR;                   $A$ $R$ $C$ $D$

{\bf Rule [17]}
Given:
\begin{eqnarray*}&&     aI(A;B)   \\
	&\leq&   (h+z)I(A;B|C)+cI(A;C|B)+dI(B;C|A)  \\
	&+& eI(A;B|D)+fI(A;D|B)+gI(B;D|A)  \\
	&+& hI(C;D)+iI(C;D|A)+jI(C;D|B)
\end{eqnarray*}
and
\begin{eqnarray*}&&     a'I(A;B)   \\
	&\leq&   b'I(A;B|C)+c'I(A;C|B)+d'I(B;C|A)  \\
	&+& (h'+z)I(A;B|D)+f'I(A;D|B)  \\
	&+& g'I(B;D|A)+h'I(C;D)+i'I(C;D|A)  \\
	&+& j'I(C;D|B)
\end{eqnarray*}
Get:
\begin{eqnarray*}
&&    (a+a'+z)I(A;B)   \\
	&\leq&   (h+b'+i'+2z)I(A;B|C)  \\
	&+& (2c+h+j+c'+f'+h'+i'+j'+z)I(A;C|B)  \\
	&+& (2d+h+i+d'+g'+h'+i'+z)I(B;C|A)  \\
	&+& (e+h'+z)I(A;B|D)+(f+f')I(A;D|B)  \\
	&+& (g+g')I(B;D|A)+(h+h'+z)I(C;D)  \\
	&+& (i+i')I(C;D|A)+(j+j')I(C;D|B)
\end{eqnarray*}
 Using:    $R$ is copy of $C$ over $AB$    
\newline Substitutions: $A$ $B$ $CR$ D;                   $A$ $B$ $C$ $DR$

{\bf Rule [18]}
Given:
\begin{eqnarray*}&&     aI(A;B)   \\
	&\leq&   bI(A;B|C)+(i+b'+c'+z)I(A;C|B)  \\
	&+& dI(B;C|A)+eI(A;B|D)  \\
	&+& (e'+z)I(A;D|B)+gI(B;D|A)+hI(C;D)  \\
	&+& (i+z_1)I(C;D|A)+jI(C;D|B)
\end{eqnarray*}
and
\begin{eqnarray*}&&     a'I(A;B)   \\
	&\leq&   b'I(A;B|C)+c'I(A;C|B)  \\
	&+& (b+j'+z)I(B;C|A)+e'I(A;B|D)  \\
	&+& f'I(A;D|B)+(e+g+z)I(B;D|A)  \\
	&+& h'I(C;D)+i'I(C;D|A)  \\
	&+& (j'+z_2)I(C;D|B)
\end{eqnarray*}
Get:
\begin{eqnarray*}
&&    (a+a'+z)I(A;B)   \\
	&\leq&  (a+b+d+g+i+a'+b'+c'+f'+j'+2z)  \\
	&& I(A;B|C)  \\
	&+& (i-a'+b'+c'+e'+j'+z)I(A;C|B)  \\
	&+& (-a+b+d+e+i+j'+z)I(B;C|A)  \\
	&+& (e+e'+z)I(A;B|D)+f'I(A;D|B)  \\
	&+& (g)I(B;D|A)+(h+i+h'+j'+z)I(C;D)  \\
	&+& i'I(C;D|A)+jI(C;D|B)
\end{eqnarray*}
 Using:    $R$ is copy of $C$ over $AB$    
\newline Substitutions: $R$ $B$ $C$ D;                   $A$ $R$ $C$ $D$

{\bf Rule [19]}
Given:
\begin{eqnarray*}&&     aI(A;B)   \\
	&\leq&   (a+z_1)I(A;B|C)+cI(A;C|B)  \\
	&+& dI(B;C|A)+(a+z)I(A;B|D)+fI(A;D|B)  \\
	&+& gI(B;D|A)+aI(C;D)+iI(C;D|A)  \\
	&+& jI(C;D|B)
\end{eqnarray*}
and
\begin{eqnarray*}&&     a'I(A;B)   \\
	&\leq&   (h'+z_2)I(A;B|C)+c'I(A;C|B)  \\
	&+& d'I(B;C|A)+e'I(A;B|D)+f'I(A;D|B)  \\
	&+& g'I(B;D|A)+h'I(C;D)+i'I(C;D|A)  \\
	&+& j'I(C;D|B)
\end{eqnarray*}
Get:
\begin{eqnarray*}
&&    (a+a'+z)I(A;B)   \\
	&\leq&   (a+h'+2z)I(A;B|C)  \\
	&+& (a+c+j+2c'+h'+j'+z)I(A;C|B)  \\
	&+& (a+d+i+2d'+h'+i'+z)I(B;C|A)  \\
	&+& (a+e'+z)I(A;B|D)+(f+f')I(A;D|B)  \\
	&+& (g+g')I(B;D|A)+(a+h'+z)I(C;D)  \\
	&+& (i+i')I(C;D|A)+(j+j')I(C;D|B)
\end{eqnarray*}
 Using:    $R$ is copy of $C$ over $AB$    
\newline Substitutions: $AR$ $BR$ $CR$ DR;                   $A$ $B$ $CR$ $D$

{\bf Rule [20]}
Given:
\begin{eqnarray*}&&     aI(A;B)   \\
	&\leq&   (a+b)I(A;B|C)+cI(A;C|B)  \\
	&+& (j+x+z)I(B;C|A)+eI(A;B|D)  \\
	&+& fI(A;D|B)+xI(B;D|A)+hI(C;D)  \\
	&+& iI(C;D|A)+(j+x)I(C;D|B)
\end{eqnarray*}
Get:
\begin{eqnarray*}
&&    (a+x+z)I(A;B)   \\
	&\leq&   (a+b+x+2z)I(A;B|C)+(c+z)I(A;C|B)  \\
	&+& (j+z)I(B;C|A)  \\
	&+& (a+c+e+f+j+2x+z)I(A;B|D)  \\
	&+& (b+e+f+j+x)I(A;D|B)+(x)I(B;D|A)  \\
	&+& (h+x+z)I(C;D)+(i)I(C;D|A)  \\
	&+& (j)I(C;D|B)
\end{eqnarray*}
 Using:    $RS$ is copy of $CD$ over $AB$    
\newline Substitution: $A$ $D$ $R$ $S$

{\bf Rule [21]}
Given:
\begin{eqnarray*}&&     aI(A;B)   \\
	&\leq&   (a+b)I(A;B|C)+cI(A;C|B)  \\
	&+& (j+w+x)I(B;C|A)+eI(A;B|D)  \\
	&+& fI(A;D|B)+(g+x)I(B;D|A)+hI(C;D)  \\
	&+& iI(C;D|A)+(j+x)I(C;D|B)
\end{eqnarray*}
and
\begin{eqnarray*}&&     a'I(A;B)   \\
	&\leq&   (a'+b')I(A;B|C)+c'I(A;C|B)  \\
	&+& zI(B;C|A)+e'I(A;B|D)+f'I(A;D|B)  \\
	&+& (g'+j'-w+z)I(B;D|A)+h'I(C;D)  \\
	&+& i'I(C;D|A)+(j'+z)I(C;D|B)
\end{eqnarray*}
Get:
\begin{eqnarray*}
&&    (a+a'+x+z)I(A;B)   \\
	&\leq&   (a+b+j+2a'+b'+c'+f'+x+2z)I(A;B|C)  \\
	&+& (c+b'+c'+e'+j'+z)I(A;C|B)  \\
	&+& (j'+z)I(B;C|A)  \\
	&+& (a+c+e+f+e'+j'+2x+z)I(A;B|D)  \\
	&+& (b+e+f+f'+j+x)I(A;D|B)  \\
	&+& (g+g'+j+x)I(B;D|A)  \\
	&+& (h+j+h'+j'+x+z)I(C;D)  \\
	&+& (i+i')I(C;D|A)
\end{eqnarray*}
 Using:    $RS$ is copy of $CD$ over $AB$    
\newline Substitutions: $A$ $D$ $R$ S;                   $A$ $C$ $R$ $S$

{\bf Rule [22]}
Given:
\begin{eqnarray*}&&     aI(A;B)   \\
	&\leq&   (h+z)I(A;B|C)+cI(A;C|B)+dI(B;C|A)  \\
	&+& aI(A;B|D)+fI(A;D|B)+gI(B;D|A)  \\
	&+& hI(C;D)+iI(C;D|A)+jI(C;D|B)
\end{eqnarray*}
and
\begin{eqnarray*}&&     a'I(A;B)   \\
	&\leq&   (a'+b')I(A;B|C)+c'I(A;C|B)  \\
	&+& (a+j'+z)I(B;C|A)+e'I(A;B|D)  \\
	&+& f'I(A;D|B)+g'I(B;D|A)+h'I(C;D)  \\
	&+& i'I(C;D|A)+j'I(C;D|B)
\end{eqnarray*}
Get:
\begin{eqnarray*}
&&    (a+a'+z)I(A;B)   \\
	&\leq&   (h+a'+b'+j'+2z)I(A;B|C)  \\
	&+& (c+h+j+c'+z)I(A;C|B)  \\
	&+& (d+h+i+z)I(B;C|A)  \\
	&+& (a+a'+c'+e'+f'+z)I(A;B|D)  \\
	&+& (f+b'+e'+f'+j')I(A;D|B)  \\
	&+& (g+g'+j')I(B;D|A)  \\
	&+& (h+h'+j'+z)I(C;D)+(i+i')I(C;D|A)  \\
	&+& jI(C;D|B)
\end{eqnarray*}
 Using:    $RS$ is copy of $CD$ over $AB$    
\newline Substitutions: $AC$ $BC$ $CR$ S;                   $A$ $D$ $R$ $S$

{\bf Rule [23]}
Given:
\begin{eqnarray*}&&     aI(A;B)   \\
	&\leq&   (a+b+f+g+x)I(A;B|C)+cI(A;C|B)  \\
	&+& dI(B;C|A)+eI(A;B|D)+fI(A;D|B)  \\
	&+& (g+g_2)I(B;D|A)+hI(C;D)+iI(C;D|A)  \\
	&+& (j+x)I(C;D|B)
\end{eqnarray*}
Get:
\begin{eqnarray*}
&&    (a+2x)I(A;B)   \\
	&\leq&   (a+b+f+g+3x)I(A;B|C)  \\
	&+& (c+d+e+f+g+g_2+4x)I(A;C|B)  \\
	&+& (a+c+d+f+g+j+2x)I(B;C|A)  \\
	&+& (e+3x)I(A;B|D)+(f+x)I(A;D|B)  \\
	&+& (g_2)I(B;D|A)+(h+2x)I(C;D)  \\
	&+& (i)I(C;D|A)+jI(C;D|B)
\end{eqnarray*}
 Using:    $R$ is copy of $A$ over BC, $S$ is copy of $A$ over $BDR$    
\newline Substitution: $A$ $R$ $C$ $D$

{\bf Rule [24]}
Given:
\begin{eqnarray*}&&     (w-x)I(A;B)   \\
	&\leq&   wI(A;B|C)+cI(A;C|B)+dI(B;C|A)  \\
	&+& (a'+b'+c'+w+z)I(A;B|D)+fI(A;D|B)  \\
	&+& gI(B;D|A)+(w-x)I(C;D)+iI(C;D|A)  \\
	&+& jI(C;D|B)
\end{eqnarray*}
and
\begin{eqnarray*}&&     a'I(A;B)   \\
	&\leq&   (a'+b')I(A;B|C)+c'I(A;C|B)  \\
	&+& (g+i+j'+w+z)I(B;C|A)+e'I(A;B|D)  \\
	&+& f'I(A;D|B)+(g'+z)I(B;D|A)  \\
	&+& h'I(C;D)+i'I(C;D|A)  \\
	&+& (j'+z)I(C;D|B)
\end{eqnarray*}
Get:
\begin{eqnarray*}
&&    (a'+w+z)I(A;B)   \\
	&\leq&   (a'+b'+j'+w+x+z)I(A;B|C)  \\
	&+& (c+j+c'+w)I(A;C|B)  \\
	&+& (d+i+w)I(B;C|A)  \\
	&+& (a'+c'+e'+f'+w+2z)I(A;B|D)  \\
	&+& (f+b'+e'+f'+j'+z)I(A;D|B)  \\
	&+& (g+g'+j'+z)I(B;D|A)  \\
	&+& (h'+j'+w+z)I(C;D)+(i+i')I(C;D|A)  \\
	&+& jI(C;D|B)
\end{eqnarray*}
 Using:    $RS$ is copy of $CD$ over $AB$    
\newline Substitutions: $AC$ $BC$ $CR$ CS;                   $A$ $D$ $R$ $S$

{\bf Rule [25]}
Given:
\begin{eqnarray*}&&     aI(A;B)   \\
	&\leq&   bI(A;B|C)+cI(A;C|B)+dI(B;C|A)  \\
	&+& (h+z)I(A;B|D)+fI(A;D|B)+gI(B;D|A)  \\
	&+& hI(C;D)+iI(C;D|A)+jI(C;D|B)
\end{eqnarray*}
Get:
\begin{eqnarray*}
&&    (a+z)I(A;B)   \\
	&\leq&   (b+z)I(A;B|C)+(c+f+h+j)I(A;C|B)  \\
	&+& (d+g+h+i)I(B;C|A)+(h+2z)I(A;B|D)  \\
	&+& (f+z)I(A;D|B)+(g+z)I(B;D|A)  \\
	&+& (h+z)I(C;D)+(i)I(C;D|A)+jI(C;D|B)
\end{eqnarray*}
 Using:    $RS$ is copy of $CD$ over $AB$    
\newline Substitution: $A$ $B$ $C$ $DR$

{\bf Rule [26]}
Given:
\begin{eqnarray*}&&     aI(A;B)   \\
	&\leq&   bI(A;B|C)+cI(A;C|B)  \\
	&+& (j+e'+f'+w+x+z)I(B;C|A)  \\
	&+& (a+e)I(A;B|D)+fI(A;D|B)  \\
	&+& (g+w)I(B;D|A)+hI(C;D)+iI(C;D|A)  \\
	&+& (j+w+x)I(C;D|B)
\end{eqnarray*}
and
\begin{eqnarray*}&&     a'I(A;B)   \\
	&\leq&   (a'+b')I(A;B|C)+c'I(A;C|B)  \\
	&+& zI(B;C|A)+e'I(A;B|D)+f'I(A;D|B)  \\
	&+& (b+c+g'+j'+w+x+z)I(B;D|A)  \\
	&+& h'I(C;D)+i'I(C;D|A)  \\
	&+& (j'+z)I(C;D|B)
\end{eqnarray*}
Get:
\begin{eqnarray*}
&&    (a+a'+w+z)I(A;B)   \\
	&\leq&   (b+j+2a'+b'+c'+f'+w+2z)I(A;B|C)  \\
	&+& (c+j'+z)I(A;C|B)  \\
	&+& (d+b'+c'+e'+j'+x+z)I(B;C|A)  \\
	&+& (2a+c+e+f+e'+j'+2w+x+z)I(A;B|D)  \\
	&+& (b+e+f+j+g'+w+x)I(A;D|B)  \\
	&+& (g+j+f'+w)I(B;D|A)  \\
	&+& (h+j+h'+j'+w+z)I(C;D)+i'I(C;D|A)  \\
	&+& xI(C;D|B)
\end{eqnarray*}
 Using:    $RS$ is copy of $CD$ over $AB$    
\newline Substitutions: $A$ $D$ $R$ S;                   $B$ $C$ $R$ $S$

{\bf Rule [27]}
Given:
\begin{eqnarray*}&&     aI(A;B)   \\
	&\leq&   bI(A;B|C)+cI(A;C|B)+(d+z)I(B;C|A)  \\
	&+& eI(A;B|D)+fI(A;D|B)+(x+z)I(B;D|A)  \\
	&+& hI(C;D)+iI(C;D|A)+(x+z)I(C;D|B)
\end{eqnarray*}
Get:
\begin{eqnarray*}
&&    (a+2x+z)I(A;B)   \\
	&\leq&   (a+b+c+f+5x+2z)I(A;B|C)  \\
	&+& (-a+b+c+e+3x+z)I(A;C|B)  \\
	&+& (d+x+z)I(B;C|A)+(e+2x+z)I(A;B|D)  \\
	&+& (f)I(A;D|B)+(h+2x+z)I(C;D)  \\
	&+& (i)I(C;D|A)
\end{eqnarray*}
 Using:    $R$ is copy of $C$ over $AB$    $S$ is copy of $R$ over $AC$    
\newline Substitution: $A$ $R$ $C$ $D$

{\bf Rule [28]}
Given:
\begin{eqnarray*}&&     aI(A;B)   \\
	&\leq&   bI(A;B|C)+cI(A;C|B)  \\
	&+& (j+g'+h'+i'+z)I(B;C|A)  \\
	&+& (a+e)I(A;B|D)+fI(A;D|B)+zI(B;D|A)  \\
	&+& hI(C;D)+iI(C;D|A)+(j+z)I(C;D|B)
\end{eqnarray*}
and
\begin{eqnarray*}&&     a'I(A;B)   \\
	&\leq&   b'I(A;B|C)+c'I(A;C|B)+d'I(B;C|A)  \\
	&+& (b+c+j+h'+z)I(A;B|D)+f'I(A;D|B)  \\
	&+& g'I(B;D|A)+h'I(C;D)+i'I(C;D|A)  \\
	&+& j'I(C;D|B)
\end{eqnarray*}
Get:
\begin{eqnarray*}
&&    (a+a'+z)I(A;B)   \\
	&\leq&   (b+b'+z)I(A;B|C)  \\
	&+& (c'+f'+h'+j')I(A;C|B)  \\
	&+& (j+d'+g'+h'+i')I(B;C|A)  \\
	&+& (2a+e+f+j+h'+2z)I(A;B|D)  \\
	&+& (b+e+f+j+f'+z)I(A;D|B)  \\
	&+& (g'+z)I(B;D|A)+(h+h'+z)I(C;D)  \\
	&+& (i+i')I(C;D|A)+(j+j')I(C;D|B)
\end{eqnarray*}
 Using:    $RS$ is copy of $CD$ over $AB$    
\newline Substitutions: $A$ $D$ $R$ S;                   $A$ $B$ $C$ $DR$

{\bf Rule [29]}
Given:
\begin{eqnarray*}&&     aI(A;B)   \\
	&\leq&   bI(A;B|C)+xI(A;C|B)+dI(B;C|A)  \\
	&+& eI(A;B|D)+(f+x)I(A;D|B)+gI(B;D|A)  \\
	&+& hI(C;D)+xI(C;D|A)+jI(C;D|B)
\end{eqnarray*}
Get:
\begin{eqnarray*}
&&    (a+2x)I(A;B)   \\
	&\leq&   (2b+d+e+g+5x)I(A;B|C)+(x)I(A;C|B)  \\
	&+& (b+2d+e+g+3x)I(B;C|A)  \\
	&+& (e+2x)I(A;B|D)+fI(A;D|B)  \\
	&+& gI(B;D|A)+(h+2x)I(C;D)+jI(C;D|B)
\end{eqnarray*}
 Using:    $R$ is copy of $A$ over $BC$    $S$ is copy of $R$ over $AB$    
\newline Substitution: $S$ $B$ $C$ $D$

{\bf Rule [30]}
Given:
\begin{eqnarray*}&&     aI(A;B)   \\
	&\leq&   (a+g+x)I(A;B|C)+cI(A;C|B)  \\
	&+& dI(B;C|A)+eI(A;B|D)+(f+x)I(A;D|B)  \\
	&+& gI(B;D|A)+hI(C;D)+xI(C;D|A)  \\
	&+& jI(C;D|B)
\end{eqnarray*}
Get:
\begin{eqnarray*}
&&    (a+2x)I(A;B)   \\
	&\leq&   (a+g+3x)I(A;B|C)  \\
	&+& (a+2c+d+e+f+2g+3x)I(A;C|B)  \\
	&+& (a+c+2d+e+f+2g+3x)I(B;C|A)  \\
	&+& (e+f+2x)I(A;B|D)+(g)I(B;D|A)  \\
	&+& (h+2x)I(C;D)+jI(C;D|B)
\end{eqnarray*}
 Using:    $R$ is copy of $A$ over $BC$    $S$ is copy of $R$ over $AC$    
\newline Substitution: $S$ $B$ $C$ $D$

{\bf Rule [31]}
Given:
\begin{eqnarray*}&&     aI(A;B)   \\
	&\leq&   bI(A;B|C)+cI(A;C|B)+xI(B;C|A)  \\
	&+& (a+x)I(A;B|D)+fI(A;D|B)+gI(B;D|A)  \\
	&+& aI(C;D)+iI(C;D|A)+xI(C;D|B)
\end{eqnarray*}
Get:
\begin{eqnarray*}
&&    (a+2x)I(A;B)   \\
	&\leq&   (a+b+i+4x)I(A;B|C)  \\
	&+& (2a+b+2c+2f+2g+4x)I(A;C|B)  \\
	&+& (c+x)I(B;C|A)+(a+2x)I(A;B|D)  \\
	&+& (b+c+f+g+x)I(A;D|B)  \\
	&+& (a+f+g+x)I(B;D|A)+(a+2x)I(C;D)  \\
	&+& (i)I(C;D|A)
\end{eqnarray*}
 Using:    $R$ is copy of $A$ over $BC$    $S$ is copy of $A$ over $BDR$    
\newline Substitution: $AR$ $RS$ $CR$ $DR$

{\bf Rule [32]}
Given:
\begin{eqnarray*}&&     (x-w)I(A;B)   \\
	&\leq&   xI(A;B|C)+cI(A;C|B)+dI(B;C|A)  \\
	&+& (x+z)I(A;B|D)+fI(A;D|B)+gI(B;D|A)  \\
	&+& (x-w)I(C;D)+iI(C;D|A)+jI(C;D|B)
\end{eqnarray*}
Get:
\begin{eqnarray*}
&&    (x+z)I(A;B)   \\
	&\leq&   (w+x+z)I(A;B|C)+(c+j+x)I(A;C|B)  \\
	&+& (d+i+x)I(B;C|A)+(x+2z)I(A;B|D)  \\
	&+& (f+z)I(A;D|B)+(g+z)I(B;D|A)  \\
	&+& (x+z)I(C;D)+(i)I(C;D|A)+jI(C;D|B)
\end{eqnarray*}
 Using:    $RS$ is copy of $CD$ over $AB$    
\newline Substitution: $AC$ $BC$ $CR$ $CS$

{\bf Rule [33]}
Given:
\begin{eqnarray*}&&     aI(A;B)   \\
	&\leq&   bI(A;B|C)+cI(A;C|B)+dI(B;C|A)  \\
	&+& eI(A;B|D)+fI(A;D|B)+xI(B;D|A)  \\
	&+& hI(C;D)+xI(C;D|A)+xI(C;D|B)
\end{eqnarray*}
Get:
\begin{eqnarray*}
&&    (a+2x)I(A;B)   \\
	&\leq&   (b+e+4x)I(A;B|C)  \\
	&+& (c+f+h+2x)I(A;C|B)+(d+x)I(B;C|A)  \\
	&+& (e+f+h+3x)I(A;B|D)  \\
	&+& (e+f+x)I(A;D|B)+(h+2x)I(C;D)
\end{eqnarray*}
 Using:    $R$ is copy of $C$ over $AB$    $S$ is copy of $R$ over $AD$    
\newline Substitution: $A$ $B$ $C$ $S$

{\bf Rule [34]}
Given:
\begin{eqnarray*}&&     (a-d)I(A;B)   \\
	&\leq&   (a-w+x)I(A;B|C)+(c+x)I(A;C|B)  \\
	&+& dI(B;C|A)+(a+x)I(A;B|D)+fI(A;D|B)  \\
	&+& gI(B;D|A)+(a-d)I(C;D)+xI(C;D|A)  \\
	&+& jI(C;D|B)
\end{eqnarray*}
Get:
\begin{eqnarray*}
&&    (a+2x)I(A;B)   \\
	&\leq&   (a+d+3x)I(A;B|C)+(a+j+2x)I(A;C|B)  \\
	&+& (a+d+2x)I(B;C|A)+(a+2x)I(A;B|D)  \\
	&+& (a+f+g+x)I(A;D|B)  \\
	&+& (a+d+f+g-w+2x)I(B;D|A)  \\
	&+& (a+2x)I(C;D)+jI(C;D|B)
\end{eqnarray*}
 Using:    $R$ is copy of $C$ over $AB$    $S$ is copy of $B$ over $ADR$    
\newline Substitution: $RS$ $BR$ $CR$ $DR$

{\bf Rule [35]}
Given:
\begin{eqnarray*}&&     (e-w)I(A;B)   \\
	&\leq&   (e+x)I(A;B|C)+cI(A;C|B)+dI(B;C|A)  \\
	&+& eI(A;B|D)+fI(A;D|B)+gI(B;D|A)  \\
	&+& (e-w)I(C;D)+iI(C;D|A)+jI(C;D|B)
\end{eqnarray*}
and
\begin{eqnarray*}&&     (b'-w')I(A;B)   \\
	&\leq&   b'I(A;B|C)+c'I(A;C|B)+d'I(B;C|A)  \\
	&+& (b'-x)I(A;B|D)+f'I(A;D|B)  \\
	&+& g'I(B;D|A)+(b'-w')I(C;D)  \\
	&+& i'I(C;D|A)+j'I(C;D|B)
\end{eqnarray*}
Get:
\begin{eqnarray*}
&&    (e+b')I(A;B)   \\
	&\leq&   (e+b'+w')I(A;B|C)  \\
	&+& (c+b'+c'+j')I(A;C|B)  \\
	&+& (d+b'+d'+i')I(B;C|A)  \\
	&+& (e+w+b')I(A;B|D)  \\
	&+& (e+f+j+f')I(A;D|B)  \\
	&+& (e+g+i+g')I(B;D|A)+(e+b')I(C;D)  \\
	&+& (i+i')I(C;D|A)+(j+j')I(C;D|B)
\end{eqnarray*}
 Using:    $RS$ is copy of $CD$ over $AB$    
\newline Substitutions: $AD$ $BD$ $DR$ DS;                   $AC$ $BC$ $CR$ $CS$

{\bf Rule [36]}
Given:
\begin{eqnarray*}&&     aI(A;B)   \\
	&\leq&   (a+x)I(A;B|C)+cI(A;C|B)+dI(B;C|A)  \\
	&+& (a+z)I(A;B|D)+fI(A;D|B)+gI(B;D|A)  \\
	&+& aI(C;D)+iI(C;D|A)+jI(C;D|B)
\end{eqnarray*}
and
\begin{eqnarray*}&&     a'I(A;B)   \\
	&\leq&   b'I(A;B|C)+c'I(A;C|B)+d'I(B;C|A)  \\
	&+& (h'-x+z)I(A;B|D)+f'I(A;D|B)  \\
	&+& g'I(B;D|A)+h'I(C;D)+i'I(C;D|A)  \\
	&+& j'I(C;D|B)
\end{eqnarray*}
Get:
\begin{eqnarray*}
&&    (a+a'+z)I(A;B)   \\
	&\leq&   (a+b'+z)I(A;B|C)  \\
	&+& (c+c'+f'+h'+j')I(A;C|B)  \\
	&+& (d+d'+g'+h'+i')I(B;C|A)  \\
	&+& (a+h'+2z)I(A;B|D)  \\
	&+& (a+f+j+f'+z)I(A;D|B)  \\
	&+& (a+g+i+g'+z)I(B;D|A)  \\
	&+& (a+h'+z)I(C;D)+(i+i')I(C;D|A)  \\
	&+& (j+j')I(C;D|B)
\end{eqnarray*}
 Using:    $RS$ is copy of $CD$ over $AB$    
\newline Substitutions: $AD$ $BD$ $DR$ DS;                   $A$ $B$ $C$ $DR$

{\bf Rule [37]}
Given:
\begin{eqnarray*}&&     aI(A;B)   \\
	&\leq&   (a+b'+c'+j'+z+z')I(A;B|C)  \\
	&+& cI(A;C|B)+dI(B;C|A)  \\
	&+& (a+e'+f'+z')I(A;B|D)+fI(A;D|B)  \\
	&+& gI(B;D|A)+aI(C;D)+iI(C;D|A)  \\
	&+& jI(C;D|B)
\end{eqnarray*}
and
\begin{eqnarray*}&&     a'I(A;B)   \\
	&\leq&   b'I(A;B|C)+c'I(A;C|B)  \\
	&+& (a+d+i+j'+z+z')I(B;C|A)  \\
	&+& e'I(A;B|D)+f'I(A;D|B)  \\
	&+& (a+g+i+z')I(B;D|A)+h'I(C;D)  \\
	&+& i'I(C;D|A)+(j'+z')I(C;D|B)
\end{eqnarray*}
Get:
\begin{eqnarray*}
&&    (a+a'+z+z')I(A;B)   \\
	&\leq&   (a+b'+j'+2z+z')I(A;B|C)  \\
	&+& (c+c'+z)I(A;C|B)+(d+z)I(B;C|A)  \\
	&+& (a+a'+c'+e'+f'+z+2z')I(A;B|D)  \\
	&+& (a+f+j-a'+b'+e'+f'+j'+z')I(A;D|B)  \\
	&+& (a+g+i+j'+z')I(B;D|A)  \\
	&+& (a+h'+j'+z+z')I(C;D)  \\
	&+& (i+i')I(C;D|A)+(j)I(C;D|B)
\end{eqnarray*}
 Using:    $RS$ is copy of $CD$ over $AB$    
\newline Substitutions: $AD$ $BD$ $DR$ DS;                   $A$ $D$ $R$ $S$

{\bf Rule [38]}
Given:
\begin{eqnarray*}&&     aI(A;B)   \\
	&\leq&   (a+x)I(A;B|C)+(x-w)I(A;C|B)  \\
	&+& dI(B;C|A)+eI(A;B|D)+xI(A;D|B)  \\
	&+& gI(B;D|A)+aI(C;D)+xI(C;D|A)  \\
	&+& jI(C;D|B)
\end{eqnarray*}
Get:
\begin{eqnarray*}
&&    (a+2x)I(A;B)   \\
	&\leq&   (2a+d+g+4x)I(A;B|C)  \\
	&+& (2j+x)I(A;C|B)+(a+d+e+3x)I(B;C|A)  \\
	&+& (a+e+j+3x)I(A;B|D)+()I(A;D|B)  \\
	&+& (a+d+g-w+2x)I(B;D|A)+(a+2x)I(C;D)  \\
	&+& ()I(C;D|A)+jI(C;D|B)
\end{eqnarray*}
 Using:    $R$ is copy of $C$ over $AB$    $S$ is copy of $B$ over $ADR$    
\newline Substitution: $RS$ $BS$ $CS$ $DS$

{\bf Rule [39]}
Given:
\begin{eqnarray*}&&     aI(A;B)   \\
	&\leq&   (a+f+x)I(A;B|C)+cI(A;C|B)  \\
	&+& dI(B;C|A)+eI(A;B|D)+fI(A;D|B)  \\
	&+& xI(B;D|A)+hI(C;D)+iI(C;D|A)  \\
	&+& xI(C;D|B)
\end{eqnarray*}
Get:
\begin{eqnarray*}
&&    (2a+c+f+2x)I(A;B)   \\
	&\leq&   (3a+2c+3f+5x)I(A;B|C)  \\
	&+& (a+2c+2d+2e+3f+3x)I(A;C|B)  \\
	&+& (a+c+d+f+x)I(B;C|A)  \\
	&+& (a+c+e+f+2x)I(A;B|D)+(f)I(A;D|B)  \\
	&+& ()I(B;D|A)+(a+c+f+h+2x)I(C;D)  \\
	&+& (i)I(C;D|A)+I(C;D|B)
\end{eqnarray*}
 Using:    $R$ is copy of $A$ over $BC$    $S$ is copy of $AR$ over $BC$    
\newline Substitution: $A$ $S$ $C$ $D$

{\bf Rule [40]}
Given:
\begin{eqnarray*}&&     aI(A;B)   \\
	&\leq&   bI(A;B|C)+xI(A;C|B)+dI(B;C|A)  \\
	&+& eI(A;B|D)+fI(A;D|B)+gI(B;D|A)  \\
	&+& hI(C;D)+xI(C;D|A)+xI(C;D|B)
\end{eqnarray*}
Get:
\begin{eqnarray*}
&&    (a+b+2x)I(A;B)   \\
	&\leq&   (3b+2d+2h+5x)I(A;B|C)  \\
	&+& (b+3x)I(A;C|B)+(b+d+x)I(B;C|A)  \\
	&+& (b+e+2x)I(A;B|D)+(f)I(A;D|B)  \\
	&+& (g)I(B;D|A)+(b+h+2x)I(C;D)
\end{eqnarray*}
 Using:    $R$ is copy of $A$ over $BC$    $S$ is copy of $AR$ over $BC$    
\newline Substitution: $A$ $B$ $S$ $D$

{\bf Rule [41]}
Given:
\begin{eqnarray*}&&     (a+x)I(A;B)   \\
	&\leq&   (a+b+g+x)I(A;B|C)+cI(A;C|B)  \\
	&+& (d+x)I(B;C|A)+(d+x)I(A;B|D)  \\
	&+& fI(A;D|B)+gI(B;D|A)+(a+x)I(C;D)  \\
	&+& iI(C;D|A)+jI(C;D|B)
\end{eqnarray*}
Get:
\begin{eqnarray*}
&&    (a+b+d+2g+2x)I(A;B)   \\
	&\leq&   (2a+b+2d+4g+i+3x)I(A;B|C)  \\
	&+& (2a+3b+2c+2d+f+4g+i+5x)I(A;C|B)  \\
	&+& (d+g+x)I(B;C|A)  \\
	&+& (a+2b+d+2g+i+3x)I(A;B|D)  \\
	&+& (a+b+c+d+f+g+2x)I(A;D|B)  \\
	&+& (b+g)I(B;D|A)+(a+b+d+2g+2x)I(C;D)  \\
	&+& (i)I(C;D|A)+jI(C;D|B)
\end{eqnarray*}
 Using:    $R$ is copy of $A$ over $BC$    $S$ is copy of $A$ over $BDR$    
\newline Substitution: $ARS$ $BRS$ $CRS$ $DRS$

{\bf Rule [42]}
Given:
\begin{eqnarray*}&&     aI(A;B)   \\
	&\leq&   (a+b+g)I(A;B|C)+cI(A;C|B)  \\
	&+& (d+d')I(B;C|A)  \\
	&+& (a+2d+d'+e)I(A;B|D)+fI(A;D|B)  \\
	&+& gI(B;D|A)+aI(C;D)+iI(C;D|A)  \\
	&+& jI(C;D|B)
\end{eqnarray*}
Get:
\begin{eqnarray*}
&&    (2a+b+2d+d'+e+g)I(A;B)   \\
	&\leq&   (3a+b+3d+1.5d'+2e+g+i)I(A;B|C)  \\
	&+& (4a+3b+2c+4d+2.5d'+e+f+4g+i)  \\
	&& I(A;C|B)+(a+2d+d'+e)I(B;C|A)  \\
	&+& (3a+2b+3d+1.5d'+e+2g+i)I(A;B|D)  \\
	&+& (a+b+c+d+d'+f+g)I(A;D|B)  \\
	&+& (b+g)I(B;D|A)  \\
	&+& (2a+b+2d+d'+e+g)I(C;D)  \\
	&+& (i)I(C;D|A)+jI(C;D|B)
\end{eqnarray*}
 Using:    $R$ is copy of $A$ over $BC$    $S$ is copy of $A$ over $BDR$    
\newline Substitution: $ARS$ $BRS$ $CRS$ $DRS$

{\bf Rule [43]}
Given:
\begin{eqnarray*}&&     aI(A;B)   \\
	&\leq&   bI(A;B|C)+cI(A;C|B)+zI(B;C|A)  \\
	&+& eI(A;B|D)+fI(A;D|B)  \\
	&+& (b'+d'+z)I(B;D|A)+hI(C;D)  \\
	&+& iI(C;D|A)+zI(C;D|B)
\end{eqnarray*}
and
\begin{eqnarray*}&&     a'I(A;B)   \\
	&\leq&   b'I(A;B|C)+c'I(A;C|B)+d'I(B;C|A)  \\
	&+& e'I(A;B|D)+f'I(A;D|B)+g'I(B;D|A)  \\
	&+& h'I(C;D)+i'I(C;D|A)+j'I(C;D|B)
\end{eqnarray*}
Get:
\begin{eqnarray*}
&&    (a+a'+z)I(A;B)   \\
	&\leq&   (a+b+c+f+b'+2z)I(A;B|C)  \\
	&+& (-a+b+c+e+c'+z)I(A;C|B)  \\
	&+& (d'+z)I(B;C|A)+(e+e'+z)I(A;B|D)  \\
	&+& (f+f')I(A;D|B)  \\
	&+& (-a'+b'+e'+g'+i')I(B;D|A)  \\
	&+& (h+h'+z)I(C;D)+(i+i')I(C;D|A)  \\
	&+& (j')I(C;D|B)
\end{eqnarray*}
 Using:    $RS$ is copy of $CD$ over $AB$    
\newline Substitutions: $A$ $C$ $R$ S;                   $AD$ $B$ $R$ $S$